\documentclass[11pt]{article}
\usepackage{jheppub}
\usepackage{tensor}
\usepackage{cleveref}
\usepackage{framed}
\usepackage{amsmath,amssymb,amsfonts,graphicx,tensor}
\usepackage{palatino}
\usepackage{newpxmath}
\usepackage{bbold}
\usepackage{subcaption}
\usepackage{graphicx}
\usepackage{enumitem}
\usepackage{braket}
\usepackage[dvipsnames]{xcolor}
\usepackage{comment}

\usepackage{microtype}

\newcommand{\cH}{\mathcal H}
\newcommand{\cM}{\mathcal M}

\newcommand{\Tr}{\operatorname{Tr}}
\newcommand{\legs}{\mathcal{L}}          
\newcommand{\extv}{\mathcal{B}}          

\newcommand{\tr}{{\rm tr}}
\newcommand{\Hcal}{\mathcal{H}}
\newcommand{\Acal}{\mathcal{A}}
\newcommand{\verts}{\mathsf{V}}
\newcommand{\edge}{\mathsf{E}}
\newcommand{\Id}{\mathbb{1}}

\newcommand{\Ind}{\operatorname{Ind}}
\newcommand{\bdy}{\partial}

\usepackage{amsthm,mathtools}

\newtheorem{theorem}{Theorem}
\newtheorem{lemma}[theorem]{Lemma}
\newtheorem{proposition}[theorem]{Proposition}
\newtheorem{corollary}[theorem]{Corollary}
\theoremstyle{definition}
\newtheorem{definition}{Definition}

\newcommand{\area}{\mathrm{area}}

\usepackage{bbm}
\usepackage{float}
\usepackage{minted}
\usepackage{tcolorbox}

\def\Tr{{\rm Tr }}

\title{Algebraic structure in holographic tensor networks}
\author{Abhisek Sahu, Jeremy van der Heijden, Mark Van Raamsdonk, Rana Zibakhsh}
\emailAdd{abhi@phas.ubc.ca}
\emailAdd{jeremy.vanderheijden@ubc.ca}
\emailAdd{mav@phas.ubc.ca}
\emailAdd{rzibakhsh@phas.ubc.ca}
\affiliation{Department of Physics and Astronomy, University of British Columbia,\\
6224 Agricultural Road, Vancouver, B.C., V6T 1Z1, Canada}

\date{}

\abstract{In a previous work 2511.21852, we explored the hypothesis that in general theories of quantum gravity, spacetime regions known as generalized entanglement wedges (introduced by Bousso and Penington) have an associated subalgebra of observables in a  semiclassical limit. We defined an entropy associated with these subalgebras that obeys the same monotonicity and strong subadditivity properties as the generalized gravitational entropy of generalized entanglement wedges. Here, we test aspects of these ideas making use of tensor network toy models for holography. In this context, the analog of a generalized entanglement wedge is a region for which no region containing it has a smaller number of boundary legs. In the limit of large bond dimension (analogous to the semiclassical limit), we argue that each such region has a canonically associated subalgebra of the full algebra of observables on the boundary Hilbert space, as well as a family of unitarily related von Neumann factor subalgebras. We investigate in detail how various algebraic operations on these subalgebras (e.g. their intersections, joins, and commutants) are related to corresponding geometrical operations on the tensor network regions. We study to what extent the monotonicity and strong subadditivity properties of generalized entropy for generalized entanglement wedges follow from properties of our algebraic entropy. We also introduce an alternative entropic quantity $S(\omega || \omega \circ E_{\mathcal{A}'})$ associated with a subalgebra $\mathcal{A}$ that also obeys monotonicity and strong subadditivity properties with suitable conditions. Finally, we explain how to define a type $II_1$ von Neumann algebra based on an infinite family of random tensor networks with increasing bond dimension such that the regions corresponding to generalized entanglement wedges can be associated with exact subalgebras of this full algebra.}

\begin{document}

\maketitle
\section{Introduction}

This paper is a follow up to \cite{Sahu:2025upe}, motivated by the broad goal of understanding the mathematical structure underlying quantum gravity in general spacetimes. There, we noted that in AdS/CFT the geometrical and causal structure of entanglement wedges reflects the underlying algebraic structure of the dual CFT, with a subalgebra of observables associated to each entanglement wedge, and properties of inclusion and spacelike separation for entanglement wedges corresponding to algebra inclusion and commutation in the dual theory (see also e.g. \cite{Dong2016Reconstruction,Harlow2017RTQEC}). Quantitatively, the Ryu-Takayanagi formula \cite{Ryu2006RT},
together with its quantum corrections 
\cite{Faulkner2013QuantumCorrections,Engelhardt2015QES} gives a relation between generalized entropy for an entanglement wedge and von Neumann entropy for the corresponding CFT subsystem (both appropriatedly regulated).

We asked whether in AdS/CFT there are more general subalgebras associated with more general gravitational subsystems and whether the entanglement wedge/subalgebra connection might extend to more general examples of quantum gravity theories. An interesting candidate for a general class of regions that might be associated with subalgebras is the set of generalized entanglement wedges (GWs) defined by Bousso and Penington \cite{Bousso:2022hlz, Bousso:2023sya}.\footnote{Related recent work proposes a gravitational path-integral
construction of states for gravitating regions without requiring
an asymptotic boundary and finds agreement with the
Bousso--Penington entropy prescription in a class of examples
\cite{Bousso2026QuantumState}.} For a time-reflection symmetric slice in a static or time-reflection symmetric spacetime we have the following definition:
\begin{itemize}
    \item A {\bf generalized entanglement wedge} is a spatial region for which no region containing it has a smaller generalized gravitational entropy $S_{gen} = \text{Area}/(4G) + S_{vN}(A)$.
\end{itemize}
The set of such regions has mathematical properties reminiscent of the set of subalgebras of an algebra:
\begin{itemize}
    \item There is a partial ordering by inclusion
    \item For GWs $W_1$ and $W_2$, the intersection $W_1 \cap W_2$ is a GW.
    \item For GWs $W_1$ and $W_2$ the union $W_1 \cup W_2$ is not generally a GW, but there is a minimal GW $W_1 \vee W_2 = \cap_{W} \{W_1 \cup W_2 \subset W\}$ that contains both, where the intersection is over GWs.
\end{itemize}
The generalized entropy of GWs also obeys certain inequalities reminiscent of properties of certain algebraically defined entropies:
\begin{itemize}
    \item {\bf Monotonicity}: For GWs $W_1 \subset W_2$, $S_{gen}(W_1) \le S_{gen}(W_2)$~.
    \item {\bf Strong subadditivity}: For GWs $W_1, W_2$, $S_{gen}(W_1) + S_{gen}(W_2) \ge S_{gen}(W_1 \cap W_2) + S_{gen}(W_1 \vee W_2)$~.
\end{itemize}
In \cite{Sahu:2025upe}, we considered the hypothesis that 
\begin{itemize}
    \item Each generalized entanglement wedge $W$ can be associated with a pair $(\mathcal{A}_W, \omega_W)$ where $\mathcal{A}_W$ is a subalgebra  of the underlying algebra of observables and $\omega_W$ is a state on $\mathcal{A}_W$.
\end{itemize}
We noted that the various properties of generalized entanglement wedges would follow from underlying algebraic properties if this correspondence obeys
\begin{enumerate}
    \item 
$W_1 \subset W_2 \iff \mathcal{A}_{W_1} \subset  \mathcal{A}_{W_2}~. $
\item
    $\mathcal{A}_{W_1 \cap W_2} = \mathcal{A}_{W_1} \cap \mathcal{A}_{W_2} \; .$
\item 
$\mathcal{A}_{W_1 \vee W_2} = \mathcal{A}_{W_1} \vee \mathcal{A}_{W_2}~$ where $\mathcal{A}_{W_1} \vee \mathcal{A}_{W_2}$ is the algebra generated by $\mathcal{A}_{W_1}$ and $\mathcal{A}_{W_2}$ \; .
\end{enumerate}
We further noted that the strong subadditivity and monotonicity properties for generalized entropies would follow from algebraic properties if the various algebras associated to GWs are von Neumann factors, if the algebras associated with $W_1 \cap W_2$, $W_1$, $W_2$, and $W_1 \vee W_2$ form a non-degenerate commuting square, and if the generalized entropy is related to algebraic quantities by
\begin{equation}
S_{gen}(W_1) - S_{gen}(W_2) = \log \text{Ind} E_{\mathcal{A}_{W_1} \to \mathcal{A}_{W_2}} - S(\omega_1 || \omega_1 \circ E_{\mathcal{A}_{W_1} \to \mathcal{A}_{W_2}})
\label{Sdiff}
\end{equation}
for $W_2 \subset W_1$. Here $\omega_1 \circ E_{\mathcal{A}_{W_1} \to \mathcal{A}_{W_2}}$ is the state $\omega_1$ on $\mathcal{A}_{W_1}$ coarse-grained to keep only the information about observables in $\mathcal{A}_{W_2}$, the relative entropy $S(\omega_1 || \omega_1 \circ E_{\mathcal{A}_{W_1} \to \mathcal{A}_{W_2}})$ measures the information lost in such a coarse-graining, and $\log \text{Ind} E_{\mathcal{A}_{W_1} \to \mathcal{A}_{W_2}}$ (where $\text{Ind} E$ is an index of the conditional expectation $E$) provides a state-independent upper bound for this relative entropy.\footnote{The right-hand side of (\ref{Sdiff}) can also be expressed as a relative entropy $S_{M}(\omega_M || \omega_M \circ \hat{E})$ where the von Neumann algebra $M = \langle \mathcal{A}_{W_1}, e_{\mathcal{A}_{W_2}}\rangle$ is the algebra defined by the Jones basic construction, $\omega_M$ is the state on $M$ associated with $\omega$, and $\hat{E}:M \to \mathcal{A}_{W_1}$ is the ``dual conditional expectation''. See appendix \ref{app:proof1} for more details.} For regions $W_1$ with finite generalized entropy, we can define a corresponding algebraic entropy by taking $W_2$ to be trivial. 

\subsubsection*{Algebras and entropies in random tensor networks}

The proposals we made in \cite{Sahu:2025upe} are speculative. In order to evaluate them in a more concrete setting, we consider in this work random tensor network toy models for holography \cite{Hayden:2016cfa}. Here, we have a Hilbert space ${\cal H}_\partial$ and an algebra of observables $B({\cal H}_\partial)$ associated with the boundary of the tensor network that plays the role of the CFT Hilbert space and observable algebra in holography.\footnote{For any Hilbert space, the notation $B(\mathcal{H})$ indicates the full algebra of bounded operators on $\mathcal{H}$. This is a von Neumann factor. When $\mathcal{H}$ is of finite dimension $n$, this algebra is isomorphic to $M_n(\mathbb{C})$, the algebra of $n \times n$ complex matrices.} The geometry of the tensor network plays the role of the bulk spatial geometry. The tensor network analog of a generalized entanglement wedge is a region of the tensor network such that no region containing it has a smaller boundary area as measured by the number of tensor network legs forming the boundary of the region. We call such regions ${\it extremal}$.

As we review in section 2, for any region $A$ of the tensor network (defined by a subset of vertices), there is a Hilbert space $\mathcal{H}_A$ associated with its boundary legs and a map $T_A : \mathcal{H}_A \to {\cal H}_\partial$ provided by the tensor network. For an extremal region, this map tends to an isometry in the limit of large bond dimension (associated with the semiclassical $G_N \to 0$ limit in the analogy with holography). If $T_A$ is an isometry, the map $T_A : \mathcal{H}_A \to {\cal H}_\partial$ represents a quantum code with $\mathcal{H}_A$ playing the role of the logical Hilbert space and ${\cal H}_\partial$ playing the role of the physical Hilbert space.\footnote{This interpretation follows the holographic quantum
error-correction perspective \cite{Almheiri2015BulkLocality}
and its realization in tensor-network encoding maps
\cite{Pastawski2015HolographicCodes}.}
In this case, we can also represent $B(\mathcal{H}_A)$, the algebra of observables on $\mathcal{H}_A$ as a subalgebra of $B(\mathcal{H}_\partial)$ by defining an operator map
\[
\varphi: B(\mathcal{H}_A) \to B(\mathcal{H}_\partial)
\]
with the properties that $\varphi$ is a *-homomorphism (i.e. respects the product and adjoint operations) and satisfies
\[
\varphi({\cal O}) T = T {\cal O}
\]
so that the action of an operator on $\mathcal{H}_A$ is faithfully represented by the action of the representing operator $\varphi(\mathcal{O})$ on the code subspace $T B(\mathcal{H})$. Any choice of $\varphi$ with these properties gives a natural boundary subalgebra associated with the extremal region $A$, so the analog of the generalized entanglement wedge/subalgebra correspondence is realized in the tensor network.

There are generally many possible choices for the operator map $\varphi$ and thus the associated boundary subalgebra $\varphi(B(\mathcal{H}_A))$. A canonical choice is
\[
\varphi({\cal O}) = T_A {\cal O} T_A^\dagger 
\]
giving the algebra $\mathcal{C}_A = T_A B(\mathcal{H}_A) T_A^\dagger$. However, (assuming all bond dimensions are equal so that $\dim \mathcal{H}_\partial$ is a multiple of $\dim \mathcal{H}_A$) we can also demand the additional structure that $\varphi(B(\mathcal{H}_A))$ is a von Neumann factor subalgebra; in this case, there is a standard quantum subsystem (defined by a tensor factor of the full Hilbert space) associated with the region $A$. There is no unique choice of $\varphi$ that gives such a factor subalgebra but rather a family of equivalent maps related by unitary transformations on $\mathcal{H}_\partial$. We will use the notation $\mathcal{M}_A = \varphi_A(B(\mathcal{H}_A))$ to refer to the boundary von Neumann factor associated to one of these choices.

\subsubsection*{Properties of the subregion-subalgebra map for isometric tensor networks}

In section 3 which forms the core of this paper, we investigate the properties of the extremal subregion/subalgebra family map in this tensor network context. To avoid dealing with approximate isometries and approximate subalgebras that we expect for large but finite bond dimension, we consider {\it isometric tensor networks} for which we simply assume that the tensor network maps associated with all extremal regions are isometric. We also focus on operator maps that give von Neumann factor subalgebras. With these assumptions, we show the following: 
\begin{itemize}
    \item For a set of extremal regions whose inclusion structure is tree-like, it is possible to choose operator maps so that the corresponding boundary subalgebras have the same inclusion structure, that is
    \[
    A \subset B \implies \mathcal{M}_A \subset \mathcal{M}_B \; .
    \]
    \item For a set of regions with a diamond-like inclusion structure $A_1 \subset A_2, A_3 \subset A_4$, it is possible to choose operator maps so that the corresponding boundary subalgebras have the same inclusion structure if and only if the principal angles between the code subspaces $T_{A_2} \mathcal{H}_{A_2}$ and $T_{A_3} \mathcal{H}_{A_3}$ have degeneracies that are multiples of $\dim \mathcal{H}_{A_1}$. We show that this will be true if $A_2 \cup A_3$ is also extremal; in this case the various subalgebras form a non-degenerate commuting square.
    \item For regions $A$ and $B$ with $A \vee B$ the minimal extremal region containing $A$ and $B$, we show that for almost all\footnote{Here and below, ``almost all'' means with probability one for a Haar-random choice in the unitary family.} choices of $\varphi_A$, $\varphi_B$, $\varphi_{A \vee B}$ giving boundary algebras respecting the inclusion structure $A,B, A \vee B$, we have
    \[
    \mathcal{M}_A \vee \mathcal{M}_B = \mathcal{M}_{A \vee B} \; ,
    \]
    that is, the subalgebras associated with $A$ and $B$ together generate the subalgebra associated with $A \vee B$.
    \item For extremal regions $A$ and $B$ in an isometric tensor network with bond dimension larger than 2 and $|\partial A|, |\partial B|$ smaller than the full boundary, almost all choices of $\varphi_A$, $\varphi_B$, $\varphi_{A \cap B}$ giving boundary algebras respecting the inclusion structure  $\mathcal{M}_{A \cap B} \subset \mathcal{M}_A , \mathcal{M}_B$ also satisfy 
    \[
\mathcal{M}_A \cap \mathcal{M}_B = \mathcal{M}_{A \cap B} \; .
\]
\item For non-intersecting and non-adjacent extremal regions $A$ and $B$ with $A \cup B$ also extremal, it is possible to define associated operator maps such that the boundary von Neumann factor algebras $\mathcal{M}_A$ and $\mathcal{M}_B$ commute.
\end{itemize}
In \cite{Sahu:2025upe}, we suggested that the property of algebra commutation might be related to the geometrical property of spacelike separation. However, the last  tensor network result here suggests the refined proposal that {\it 
$\mathcal{A}_{W_1}$ and $\mathcal{A}_{W_2}$ commute if and only if $W_1 \cap W_2$ is empty and $W_1 \cup W_2$ is a GW.}

We also evaluate in section 3 the algebraic entropy formula (\ref{Sdiff}) in the tensor network context. Given a region $A$, we consider a mixed state described by a density matrix $\rho = T_A \rho_A T_A^\dagger$ produced by feeding in an arbitrary density matrix $\rho_A$ on $\mathcal{H}_A$ into the tensor network. We find that, for any choice of the boundary von Neumann factor subalgebra $\mathcal{M}_A$ associated with the region $A$, the associated entropy is  
\[
S_{alg} =  |\partial A|  \log D - \tr\, \rho_A \log \rho_A \; ,
\]
which is precisely the expected tensor network version of generalized entropy. Here, $|\partial A|$ is the number of boundary legs for the region $A$ that plays the role of area and $\log D$, the logarithm of the bond dimension is the tensor network version of $1/4 G_N$. The term $- \tr \rho_A \log \rho_A$ is the von Neumann entropy of the state $\rho_A$ and can be interpreted as the bulk entropy in the region $A$.

For extremal regions $A \subset B$, our results establish that we can choose algebras with the same inclusion structure, so the monotonicity of the generalized gravitational entropies is directly related to the positivity of the right side of (\ref{Sdiff}). For extremal regions $A$ and $B$ where $A \cup B$ is extremal, our results show that we can associate algebras to $A \cap B$, $A$, $B$, and $A \cup B$ with the same inclusion structure as the regions which furthermore form a non-degenerate commuting square. In this case, the algebraic entropies satisfy a strong subadditivity inequality, providing an algebraic reason for the strong subadditivity exhibited by the generalized entropies of the regions. On the other hand, when $A \cup B$ is not extremal, our results do not yet give any algebraic relation underlying the strong subadditivity for generalized entropies of regions $A \cap B$, $A$, $B$, and $A \vee B$.

Sections 4 and 5 contain two related observations that lie somewhat outside the main line of development and can be read independently of each other and of sections 2 and 3:

\subsubsection*{Type $II_1$ algebras from tensor network sequences}

For random tensor networks with large but finite bond dimension, the tensor network maps associated with extremal regions will be close to isometries with high probability, but only exact isometries for special choices of tensors. In section 4, we point out that there is a way to formally work in the large $D$ limit of random tensor networks in such a way that we have exact isometries from Hilbert spaces associated to extremal regions to the boundary Hilbert space. The construction is based on considering an infinite sequence of random tensor networks with increasing bond dimension and a type of limit known as an ultrafilter limit. The resulting algebras associated with the boundary and with other extremal regions are type $II_1$ von Neumann algebras. Type $II$ algebras also arise in gravitational constructions of algebras associated with
de Sitter static patches and more general subregions
(see e.g., \cite{Chandrasekaran2023DeSitter,Jensen2023GeneralSubregions}).

\subsubsection*{Commutant entropy}

In section 5, we note that there is an alternative entropy formula that also exhibits monotonicity and strong subadditivity properties under suitable conditions. Specifically, for an algebra $\mathcal{M}$ with faithful normal state $\omega$, a subalgebra $\mathcal{A}$ with commutant $\mathcal{A}'$ in $\mathcal{M}$ and a faithful normal conditional expectation $E_{\mathcal{A}'} : \mathcal{M} \to \mathcal{A}'$, we define 
the {\it commutant entropy}
\begin{equation}
\label{Sgen}
S_C(\mathcal{M},\mathcal{A},\omega) = S(\omega || \omega \circ E_{\mathcal{A}'}) \; .
\end{equation}
The state $ \omega \circ E_{\mathcal{A}'}$ is a coarse-graining of the global state $\omega$ to a state on the algebra $\mathcal{A}'$. We can think of it as a state coarse-grained to preserve only the information about the observables commuting with those in $\mathcal{A}$. The expression $S(\omega || \omega \circ E_{\mathcal{A}'})$ is the relative entropy between the original state and the coarse-grained state, {\it measuring how much information we have discarded through the coarse-graining to keep only observables in $\mathcal{A}'$}. We show that the entropy formula (\ref{Sgen}) obeys both monotonicity and strong subadditivity properties given appropriate assumptions. Thus, like $S_{alg}$, it behaves like generalized entropy for generalized entanglement wedges. As we will explain, in certain cases $S_C$ and $S_{alg}$ define the same quantity, but in other cases, they are different. Various elements used in the construction in this section have been discussed in  \cite{Gao2020RelativeEntropy} and
\cite{Gao2017Unifying}.

We end in section 6 with a brief discussion where we comment on the related works that have recently appeared in the literature.

\paragraph{Note:} In sections 2 and 3 of this paper, we make use of a variety of rigorous results concerning upgrading networks of quantum codes (as we have in an isometric tensor network) to ``factor code networks'' where each isometry has an associated unital operator map, allowing us to define von Neumann factor subalgebras associated with various Hilbert spaces in the parent Hilbert spaces. Since these results may be of independent interest, we have presented them along with their proofs in a companion paper \cite{FactorCode}.

\section{Algebras for extremal regions in holographic tensor networks}

In \cite{Sahu:2025upe}, we proposed to explore the relation between subregions and subalgebras in the context of random tensor network toy models for holography \cite{Hayden:2016cfa}. Here, we describe a detailed investigation aimed at constructing an explicit map from an extremal subregion to a subalgebra and evaluating whether  various proposed properties of the subregion/subalgebra relation hold in the tensor network setting. 

We recall that the basic construction starts from a graph/network with some number $N$ of ``external'' edges, as shown in Figure \ref{fig:TNremove}.
Let $G=(\verts,\edge)$ be the tensor network graph, with a set $\legs$ of dangling legs ending on external vertices $\extv$. We write $\widehat\edge=\edge\cup\legs$ for the extended edge set.

To this network, we associate a quantum state on a Hilbert space ${\cal H}_\partial = \otimes^N \mathbb{C}^{D}$, where $D$ is known as the {\it bond dimension}. The construction proceeds by independently and randomly choosing states 
\[
\sum_{\{i_k\}} T_{i_1 \cdots i_n} |e_{i_1} \rangle \otimes \cdots \otimes |e_{i_n} \rangle
\]
associated with each vertex of degree $n$ in the network.\footnote{Specifically, we choose a basis $|e_i \rangle$ for $\mathbb{C}^D$, choose a random unitary $U: \otimes^n\mathbb{C}^{D} \to  \otimes^n\mathbb{C}^{D}$ and define $\sum_{\{i_k\}} T_{i_1 \cdots i_n} |e_{i_1} \rangle \otimes \cdots \otimes |e_{i_n} \rangle = U |e_1 \rangle \otimes \cdots \otimes |e_1 \rangle$. }
The state associated with the network is
\[
\psi_{i_1 \cdots i_N} |e_{i_1} \rangle \cdots |e_{i_N} \rangle~,
\]
where $\psi_{i_1 \cdots i_N}$ is the product of the tensors $T$ with pairs of indices associated with each edge in the network contracted.\footnote{Equivalently, we construct the product state from the states associated with each vertex in the network and define maps $\sum_i \langle i| \otimes \langle i| : (\mathbb{C}^N)^2 \to \mathbb{C}$ associated with each edge. The tensor product of these maps takes the full product state to a state in ${\cal H}_\partial$.}

\subsection{Hilbert spaces and boundary maps for tensor network subsystems}

\begin{figure}
    \centering
    \includegraphics[width=\linewidth]{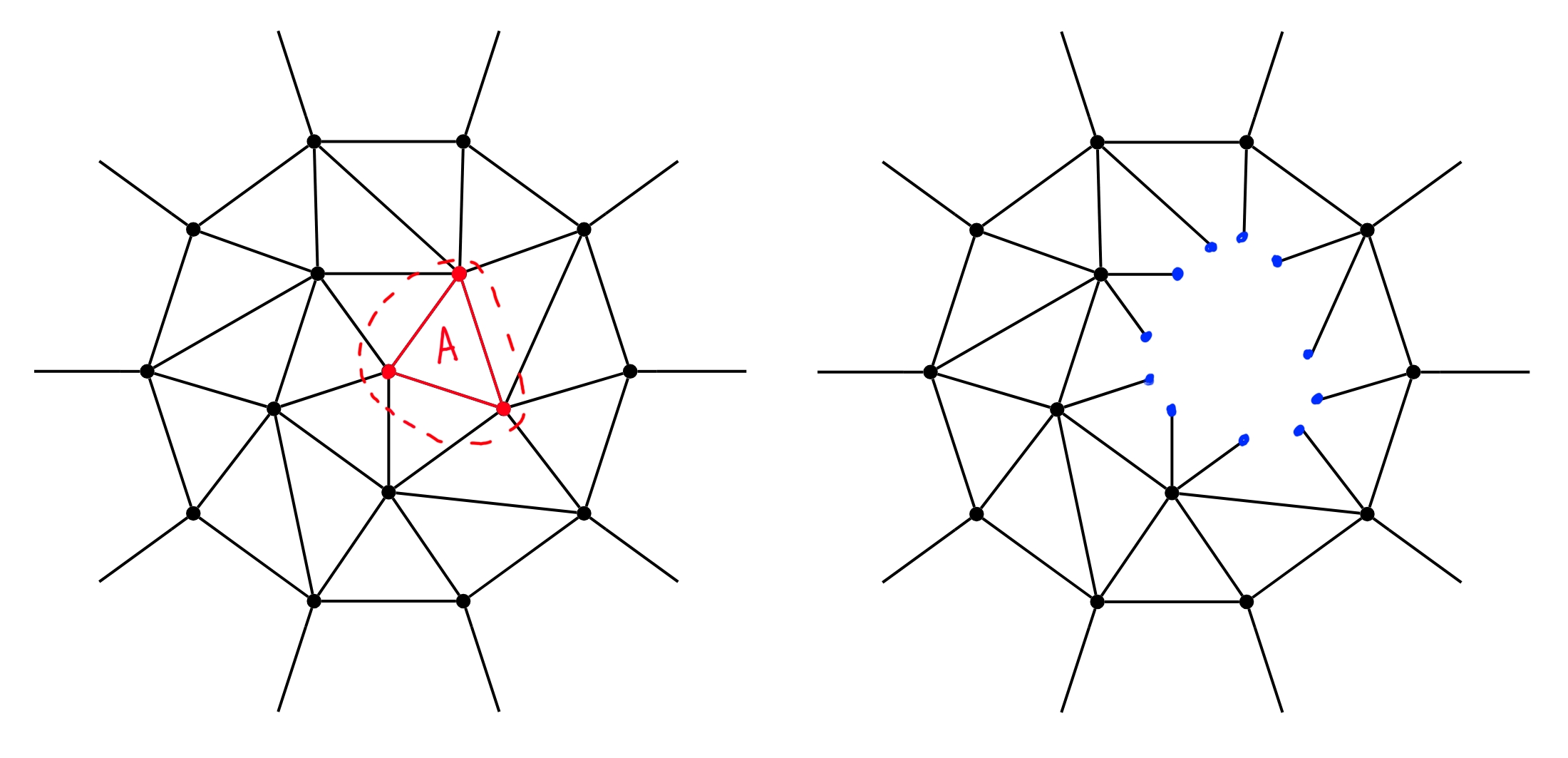}
    \caption{Left: A tensor network with 10 external legs and a region $A$ consisting of three vertices and the edges connecting these. Right: removing the region $A$, the remaining tensor network can be used to define a map $T_A$ from the Hilbert space associated with the revealed legs (blue endpoints) to the boundary Hilbert space. In this case, the map cannot be an isometry, since the dimension of the input space is larger than the dimention of the output space.}
    \label{fig:TNremove}
\end{figure}

To each subset $A$ of vertices in the tensor network (we will call this a \emph{region}), we can associate a Hilbert space ${\cal H}_A = \otimes^{|\partial A|}(\mathbb{C}^D)$,  where $\partial A$ is set of edges (including external edges) with exactly one vertex in $A$. Keeping only the vertices in $A$ and edges connected to these vertices, we have a smaller tensor network that constructs a state $|\psi_A \rangle$ on ${\cal H}_A$.

We can define a map $\hat T_A: {\cal H}_A \to {\cal H}_\partial$ starting from the state $|\psi_{\bar{A}} \rangle$ associated with the complementary region $\bar{A}$ and using the Choi–Jamiołkowski isomorphism to reinterpret this as a map. That is, if $(\psi_{\bar{A}})_{\vec{i} \vec{j}}$ is the tensor associated with the state $|\psi_{\bar{A}} \rangle$ in our chosen basis, with $\vec{i}$ and $\vec{j}$ the indices associated with edges in $\partial A$ and the remaining boundary edges respectively, we define
\begin{equation} \label{eq:CJiso}
\hat T_A |i_1 \rangle \otimes \cdots \otimes |i_M \rangle = \sum_{\vec{j}}(\psi_{\bar{A}})_{\vec{i} \vec{j}} |j_1 \rangle \otimes \cdots \otimes |j_N \rangle~,
\end{equation}
where $M=|\partial A|$, the number of legs crossing the boundary of $A$.

We will now see that (with an appropriate normalization) the map $\hat T_A$ approaches an isometry in the limit of large bond dimension if the region $A$ satisfies the extremality condition that there is no region containing $A$ that has fewer boundary legs.

\subsection{Isometries from generalized entanglement wedges}

We will consider a tensor network map $\hat T_A$ isometric if there is an appropriate normalization $T_A = {\cal N} \hat T_A$ such that
\[
T_A^\dagger T_A = \mathbb{1}_{{\cal H}_A} \;, 
\]or equivalently, the positive operator
\begin{equation}
\label{eq:Tprojector}
    P_A \equiv T_A T_A^\dagger = T_A T_A^\dagger T_A T_A^\dagger  = P_A^2~,
\end{equation}
is a projector.
This leads to
\[
\tr(T_A^\dagger T_A) = D_A \; ,
\]
so if it is possible to get an isometry from $\hat T_A$ after normalization, the properly normalized map will be
\begin{equation} \label{eq:normalized_TA}
T_A \equiv \sqrt{D_A} {\hat T_A \over \tr^{1 \over 2}(\hat T_A^\dagger \hat T_A)} \; .
\end{equation}
The isometry condition can be written alternatively as 
\begin{equation}
\label{eq:purity}
{\tr(\hat \rho_A^2) \over \tr^2(\hat \rho_A)} = {1 \over D_A}~,
\end{equation}
where we have defined the positive operator
\begin{equation}
\hat \rho_A \equiv  \hat T_A^\dagger\hat T_A~.
\end{equation}
Equation \eqref{eq:purity} is the statement about the purity \(\tr (\rho_A^2)\) of the normalized density operator $\rho_A \equiv \hat \rho_A/\tr(\hat \rho_A)$, namely that it takes the minimum value possible. 

We can use the results of \cite{Hayden:2016cfa} to estimate this purity for random tensor network states. In that paper, the authors were interested in R\'enyi entropies of reduced density operators associated with boundary subregions in a network, but the same approach works in our case to calculate the purity of $\rho_A$; we just think of the legs revealed by removing $A$ as a part of the boundary of a new tensor network and compute the R\'enyi entropy of the density operator for this subsystem.

The result of \cite{Hayden:2016cfa} is that in the limit of large bond dimension, the second R\'enyi entropy, averaged over the tensors in the network, is given by
\begin{equation}
\overline{S_2} = -\ln(Z_1) + \ln(Z_0)~,
\end{equation}
where the overline indicates that we have performed the average, and $Z_1$ and $Z_0$ are partition functions for an Ising model,
\begin{equation}
Z = \sum_{\{s_i\}} e^{-H(\{s_i\})}~,
\end{equation}
with a dynamical spin $s_i = \pm 1$ at the site of every tensor and with a fixed spin at each boundary leg. For the partition function $Z_1$, all the spins at the boundary of region $A$ are set to $+1$ while the remaining spins on the boundary of the network are set to $-1$. For $Z_0$, all boundary spins are set to $+1$. The Hamiltonian is given by
\begin{equation}
\label{eq:Ising}
H(\{s_i\}) = -{1 \over 2} \log (D) \sum_{<xy>}  (s_x s_y - 1)~,
\end{equation}
where the coupling strength is determined by the bond dimension $D$, so we get an additional energy $\log D$ for each edge with a different spin value on either end.

In the limit of large bond dimension, this sum is dominated by the configuration of spins that minimizes the Hamiltonian, which is the configuration that minimizes the number of links with differing spins. When $A$ satisfies our extremality condition, i.e. when no region $B$ containing $A$ has fewer boundary legs than $A$, the best we can do is set all dynamical spins to $+1$ for $Z_0$ and to $-1$ for $Z_1$. In this case, the Hamiltonian vanishes in the first case and is equal to $\log(D_A)$ in the second case, so we obtain the result:
\begin{equation}
\overline{S_2} = \log(D_A)~.
\end{equation}
At leading order in large bond dimension, the purity is then
\begin{equation}
\overline{\tr(\hat \rho_A^2) \over \tr^2(\hat \rho_A)}=e^{-\overline{S_2}}= \frac{1}{D_A}~,
\end{equation}
which is the minimum possible value, signalling an isometry. 

On the other hand, if there is some region $B$ whose boundary has a smaller number of links than the boundary of $A$, the best we can do is set all dynamical spins to $-1$ in the region $B$ and $+1$ outside $B$ for $Z_1$, and again to $+1$ everywhere for $Z_0$. In this case we obtain the result
\begin{equation}
\overline{S_2} = \log(D_B)~,
\end{equation}
where $D_B$ is the dimension of the Hilbert space $\mathcal{H}_B$. The purity now satisfies 
\begin{equation}
\overline{\tr(\hat \rho_A^2) \over \tr^2(\hat \rho_A)}=e^{-\overline{S_2}} = \frac{1}{D_B} > \frac{1}{D_A}~,
\end{equation}
so the map $T_A$ is \emph{not} an isometry. This matches our intuition that to map $\mathcal{H}_A$ to the boundary Hilbert space $\mathcal{H}_{\partial}$ one has to go through $\mathcal{H}_{B}$, a Hilbert space of smaller dimension than $\mathcal{H}_{A}$, providing an obstruction for the map to be an isometry.

As a further check, we note that for an isometry $T_A$, $T_A T_A^\dagger$ is a projection, so $(T_A T_A^\dagger)^n = T_A T_A^\dagger$ and $\tr((T_A T_A^\dagger)^n) = \tr(T_A T_A^\dagger) = D_A$. In \cref{app:average}, we verify explicitly that the average of \(\tr({( T_A T_A^\dagger)^n})\) is $D_A$ at leading order in large $D$ for all $n$.

We have thus shown that the map $T_A$ obeys the isometry condition (\ref{eq:purity}) to leading order in large $D$ if and only if the region $A$ is extremal in the sense that there is no region that contains it that has smaller ``area'', which is the number of boundary legs in our tensor network model. For large but finite $D$, the map $T_A$ for typical tensors will be an approximate isometry. In section \ref{sec:largeD}, we will describe a sense in which we can take the large $D$ limit precisely to define an exact isometry associated with an extremal region $A$. Alternatively, through specific choices of tensors, we can ensure that the map associated with an extremal region is an exact isometry for finite $D$. In the discussions below, in order that we can work with precisely defined subalgebras associated with bulk regions in the finite-dimensional context, we will sometimes consider tensor networks such that the map $T_A$ is an exact isometry for every extremal region $A$ and further, that the map $T_{BA}:\mathcal{H}_A \to \mathcal{H}_B$ is an isometry for any extremal $A,B$ with $A \subset B$. We refer to such a network as an {\it isometric tensor network}.

\subsection{Algebras for extremal regions} 

\label{sec:algebras_extremal_regions}

In this section, we will consider some extremal region $A$ and the associated map $T_A$. We will assume that this is an isometry; in the case where it is only an approximate isometry, the various results that we describe will only be true in some approximate sense. 

Let ${\cal M}_\partial = B({\cal H}_\partial)$ be the algebra of operators on the boundary Hilbert space. We would like to argue that there are various subalgebras ${\cal M}_\partial$ that can be naturally associated with an extremal bulk region $A$. 

Let $B({\cal H}_A)$ be the algebra of operators on the Hilbert space ${\cal H}_A$ associated with the region $A$. For any ${\cal O}_A \in B({\cal H}_A)$, we have a natural map to a boundary operator via
\[
\mathcal{O}_A \to T_A \mathcal{O}_A T_A^\dagger  \; .
\]
The image of this map gives some subspace of ${\cal M}_{\partial}$ regardless of whether $T_A$ is an isometry. However, when it is, we have
\[
T_A \mathcal{O}_A T_A^\dagger T_A \mathcal{O}_B T_A^\dagger = T_A \mathcal{O}_A \mathcal{O}_B T_A^\dagger~,
\]
so the image is closed under multiplication and therefore forms a subalgebra. Since the map $T_A$ also preserves the adjoint structure, i.e.
\[
(T_A \mathcal{O}_A T_A^\dagger)^\dagger = T_A \mathcal{O}_A^\dagger T_A^\dagger \; ,
\]
the image of $B({\cal H}_A)$ is more precisely a $*$-subalgebra $\mathcal{C}_A$ of the boundary algebra. This can be equivalently be characterized as the set of operators $P_A B(\mathcal{H}_\partial) P_A$, where $P = T_A T_A^\dagger$ is the projector to the subspace of ${\cal H}_\partial$ that is the image of ${\cal H}_A$. An algebra defined in this way is known as a {\it corner}. 

This subalgebra $\mathcal{C}_A$ is generally \emph{not} {\it unital} i.e. it does not  generally contain the identity operator in $B(\mathcal{H}_\partial)$, so it is not a \emph{von Neumann} subalgebra of the larger algebra.\footnote{Here, we are using the terminology that any von Neumann subalgebra is assumed to be {\it unital}, i.e. to contain the identity of the full algebra. For a finite-dimensional Hilbert space, any unital $C^*$-subalgebra is a von Neumann subalgebra.} We can generate a von Neumann subalgebra by taking the double commutant $\mathcal{N}_A \equiv \mathcal{C}_A''$. This is the subalgebra of operators of the form
\[
P_A {\cal O} P_A + \lambda (\mathbb{1}-P_A) \;,
\]
for arbitrary $\lambda \in \mathbb{C}$. The algebra $\mathcal{N}_A$ is the von Neumann subalgebra of the boundary algebra most canonically associated with the region $A$, but it is not generally a \emph{factor}\footnote{A von Neumann algebra $\mathcal{N}$ is called a factor if there are no non-trivial operators in $\mathcal{N}$ that commute with all operators in $\mathcal{N}$. Operators that do satisfy this property form a subalgebra called the center, $\mathcal{Z}(\mathcal{N})=\{\widetilde{\mathcal{O}} \in \mathcal{N}\, |\, [\widetilde{\mathcal{O}},\mathcal{O}] =0~,\,\forall \mathcal{O}\in \mathcal{N} \}$.}, since operators of the form $a P_A + b (\mathbb{1}-P_A)$ give a non-trivial center. In order to demonstrate the strong subadditivity properties of our proposed entropy formula (\ref{Sdiff}) in \cite{Sahu:2025upe}, it was necessary to assume that the algebra associated with a region to be a von Neumann factor, so we would like to ask whether this is possible to achieve in the tensor network context. 

A subalgebra ${\cal N}$ of the algebra ${\cal M}$ of operators on a finite dimensional Hilbert space ${\cal H}$ is a von Neumann factor if and only if we can represent ${\cal H} = {\cal H}_1 \otimes {\cal H}_2$ and ${\cal N}$ is the set of operators of the form
\[
{\cal N} = U ( B({\cal H}_1) \otimes \mathbb{1}) U^\dagger \; 
\]
where $U: {\cal H}_1 \otimes {\cal H}_2 \to {\cal H}$ is some fixed unitary. This means that after performing the change of basis associated with $U$, the operators in ${\cal N}$ are represented by block diagonal matrices whose blocks are all equal. 

In our tensor network context, we thus require that in some basis, the mapping from $B({\cal H}_A)$ to ${\cal M}_\partial$ takes the form
\[
{\cal O} \to \left( \begin{array}{cccc} {\cal O} & 0 & & 0 \cr 0 & {\cal O} & & 0 \cr
& & \ddots & \cr
0 & 0 & & {\cal O} \end{array} \right)~.
\]
We see that the image operator acts identically on a set of $m = \dim{\cal H}_\partial/\dim{\cal H}_A$ subspaces with the dimension of ${\cal H}_A$ that are orthogonal to one another and together span the full boundary Hilbert space.

This observation suggests the following construction of a von Neumann factor of $B(\mathcal{H}_\partial)$ associated with the region $A$. Choose isometries $T^{(a)}$, $a=1, \dots, m$ such that $T^{(1)} = T_A$ and the $T^{(a)}$ map into orthogonal subspaces of $\mathcal{H}_\partial$. In this case, 
\[
U_A \equiv (T^{(1)} \dots T^{(m)})
\]
defines a unitary from $\mathcal{H}_A \otimes \mathbb{C}^m$ to $\mathcal{H}_\partial$. Using this, we can define an operator map
\begin{equation}
\varphi(\mathcal{O}_A) \equiv U_A ( \mathcal{O} \otimes \mathbb{1}) U_A^\dagger = \sum_a T^{(a)} {\cal O}_A T^{(a)}{}^\dagger \; .
\label{opmap}
\end{equation}
The image defines a subalgebra $\mathcal{M}_A$ of the boundary algebra that is isomorphic to $B(\mathcal{H}_A)$, is a factor, and includes the identity of the boundary algebra. 

The isometries \(T^{(a)}\) define orthogonal projectors $P^{(a)} = T^{(a)}T^{(a)\dagger}$ that satisfy
\[
\sum_a T^{(a)}T^{(a)\dagger}=\mathbb{1}_\partial~,
\qquad 
T^{(a)\dagger}T^{(b)}=\delta^{ab}\mathbb{1}_A~.
\]

A notable feature of this identification of a boundary factor subalgebra with a region $A$ is the non-uniqueness associated with our choices of the auxiliary isometries $T^{(a)}$, $a=2,\dots, m$. We can understand this freedom in the language of quantum error correction. The isometry $T_A :\mathcal{H}_A \to \mathcal{H}_\partial$ defines a quantum code through which state of a logical Hilbert space $\mathcal{H}_A$ are encoded in a code subspace $T_A \mathcal{H}_A$ of a physical Hilbert space $\mathcal{H}_\partial$. The map $\varphi$ defines an operator map, instructing us how to implement operations on the logical Hilbert space via operations on the physical Hilbert space. In general, such an operator map should satisfy
\[
\varphi({\cal O}_1) \varphi({\cal O}_2) = \varphi({\cal O}_1 {\cal O}_2)~,  \qquad \qquad \varphi({\cal O}^\dagger) = \varphi({\cal O})^\dagger~, 
\]
so that the algebraic structure is preserved and
\begin{equation} \label{eq:qec}
T_A {\cal O}_A = \varphi({\cal O}_A) T_A  \; ,
\end{equation}
so that the operations are represented faithfully on the code subspace. If we additionally demand that doing nothing to the logical Hilbert space should do nothing to the physical Hilbert space, then $\varphi$ must be a unital map
\[
\varphi(\mathbb{1}_A) = \mathbb{1}_\partial~.
\]
The operator maps defined in (\ref{opmap}) are the most general ones satisfying these requirements. The freedom in choosing the auxiliary isometries $T^{(a)}, a>1$ is related to the the fact that we do not care how the mapped operators act outside the code subspace apart from the requirement that the image of $\varphi$ is a unital algebra. 

We can describe the freedom in choosing the auxiliary isometries more explicitly by working in a basis corresponding to the block diagonal representation of $\varphi({\cal O})$ given above. The isometry $T_A = T^{(1)}$ maps an orthonormal basis $\{\ket{\psi_j}\}$ of $\mathcal{H}_A$ to an orthonormal set in $\mathcal{H}_\partial$,
\[
\ket{\phi^1_j} = T^{(1)}\ket{\psi_j}, \qquad j = 1, \dots, \dim\mathcal{H}_A~.
\]
The auxiliary isometries $T^{(a)}$ amount to a choice of completion of this set to a full orthonormal basis for $\mathcal{H}_\partial$,
\[
T^{(a)}\ket{\psi_j} = \ket{\phi^a_j}, \qquad a = 1,\dots, m\;;\quad j = 1,\dots, \dim\mathcal{H}_A~,
\]
so that the unitary $U_A : \mathcal{H}_A \otimes \mathbb{C}^m \to \mathcal{H}_\partial$ acts as $U_A\ket{\psi_j, e_a} = \ket{\phi^a_j}$.

Any two admissible completions of $\{\ket{\phi^1_j}\}$ to a full orthonormal basis differ by a unitary $U_\perp$ acting on the orthogonal complement of $\mathrm{span}\{\ket{\phi^1_j}\}$, which has dimension $D_\partial - D_A$. Making use of the basis associated with one choice of isometries, such that the operator map is represented as
\[
{\cal O} \to  \left( \begin{array}{cc} {\cal O} & 0 \cr 0 & {\cal O} \otimes \mathbb{1}_{m-1}  \end{array} \right) \qquad \qquad m \equiv {D_\partial \over D_A} ~.
\]
the operator map associated with any other choice is 
\[
{\cal O} \to  \left( \begin{array}{cc} {\cal O} & 0 \cr 0 & U_\perp ({\cal O} \otimes \mathbb{1}_{m-1}) U_\perp^\dagger \end{array} \right)  ~.
\]
The unitaries $U_\perp$ and $U_\perp (\mathbb{1} \otimes u_{m-1})$ give the same operator map, so overall, the space possible boundary von Neumann factor subalgebras associated with $T_A$ is parameterized by $U(D_\partial - D_A)/U(m - 1)$.

It is interesting to note (see Appendix~\ref{app:CornerAvg}) that this operator map, averaged over all choices of $U_\perp$ according to the Haar measure gives a map
\begin{equation}
    \overline{\varphi}(\mathcal{O}_A) \equiv \mathbb{E}_\perp\!\left[\varphi^{(U_\perp)}(\mathcal{O}_A)\right] 
    = T_A \mathcal{O}_A T_A^\dagger + \tau(\mathcal{O}_A)(\mathbb{1}_\partial - P_A)~,
\end{equation}
where $\tau(\mathcal{O}_A) = \frac{1}{D_A}\mathrm{Tr}(\mathcal{O}_A)$ is the 
normalized trace and $P_A = T_A T_A^\dagger$ projects onto the code subspace. This is the minimal upgrade to the canonical map ${\cal O} \to T_A {\cal O} T_A^\dagger$ to make it unital, but the resulting map $\overline{\varphi}$ no longer preserves the product structure.

\subsubsection{Summary and formal results}

The results in this section are formalized in sections 2 and 3 of \cite{FactorCode}. There, a situation where we have an isometry $T: {\cal H}_{\rm code} \to {\cal H}_{\rm phys}$ is referred to as a {\it quantum code}, and a linear map $\varphi:B({\cal H}_{\rm code}) \to B({\cal H}_{\rm phys})$  preserving the $*$-algebra structure) and satisfying $\varphi({\cal O}) T = T {\cal O}$ is referred to as an {\it operator map}. It is proven that:
\begin{proposition}
\label{factor}
    For a quantum code $T : \mathcal{H}_{\rm code} \to \mathcal{H}_{\rm phys}$ with operator map $\varphi$, the following are equivalent:
\begin{enumerate}
        \item The operator map $\varphi$ is unital, i.e. $\varphi(I_{\rm code}) = I_{\rm phys}$.
        \item $\varphi(U)$ is unitary for every unitary $U$.
        \item The ratio $m = \dim \mathcal{H}_{\rm phys}/\dim \mathcal{H}_{\rm code}$ is an integer and there is a unitary $W:{\cal H}_{\rm code} \otimes \mathbb{C}^m \to \mathcal{H}_{\rm phys}$ such that 
        \[
        \varphi({\cal O}) = W ({\cal O} \otimes I_m ) W^\dagger \;, 
        \]
        and there is some unit vector  $e_1 \in \mathbb{C}^m$ such that 
        \[
        T |\psi \rangle = W (|\psi \rangle \otimes e_1 ) \qquad \forall |\psi \rangle \in {\cal H}_{\rm code} \;. 
        \]
        \item The algebra $\varphi(B(\mathcal{H}_{\rm code}))$ is a unital von Neumann factor subalgebra of $B(\mathcal{H}_{\rm phys})$.
    \end{enumerate}
\end{proposition}
The correspondence between von Neumann factor subalgebras, unital operator maps, and choices of auxiliary isometries is formalized via:
\begin{proposition}
\label{Tphi}
For a quantum code $T : \mathcal{H}_{\rm code} \to \mathcal{H}_{\rm phys}$, with integer $m = \dim  \mathcal{H}_{\rm phys} / \dim \mathcal{H}_{\rm code}$, there is a one-to-one correspondence between 
\begin{enumerate}
    \item unital operator maps, i.e. unital $*$-homomorphisms
$\varphi:B(\mathcal H_{\rm code}) \to B(\mathcal H_{\rm phys})$ satisfying
\[
\varphi(\mathcal O)T=T\mathcal O \qquad
\text{for all } \mathcal O\in B(\mathcal H_{\rm code}) \; .
\]
\item equivalence classes of orthogonal isometric completions, defined to be 
sets  $\{T^{(a)}, a=1\dots m\}$ of isometries where $T^{(1)} = T$ and $T^{(a)} \mathcal{H}_{\rm code}$ are orthogonal subspaces of $\mathcal{H}_{\rm phys}$. Two completions are said to be equivalent if
\[
T_2^{(a)} = \sum_{b=2}^m U_{ab} T_1^{(b)}\;, \qquad \qquad a = 2, \dots , m\;, \qquad U \in U(m-1)\; .
\]
\item unital von Neumann factor subalgebras $\mathcal M\subset B(\mathcal H_{\rm phys})$ such that $\mathcal M\cong B(\mathcal H_{\rm code})$ and the code subspace \(T\mathcal H_{\rm code}\) is invariant under $\mathcal M$.
\end{enumerate}
After choosing a reference operator map $\varphi_0$ where $\varphi_0$ is an arbitrary fixed reference operator map the above spaces can be identified with the quotient $U(nm - n)/U(m-1)$ where $n = \dim \mathcal{H}_{\rm code}$. 
\end{proposition}

These results are further explained and proved in \cite{FactorCode}.

\section{Properties of the extremal region to subalgebra map for isometric tensor networks}
\label{sec:main}
In the previous section, we described how we can associate a family of boundary von Neumann factor subalgebras to any extremal region in an isometric tensor network. In this section, we consider the situation where we have two or more such regions and explore the relationship between the associated families of boundary subalgebras. We will ask when the properties of inclusion, intersection, and joining (taking the minimal extremal region that includes a pair of regions) can be carried over to the associated subalgebras and also investigate conditions under which subalgebras associated with different extremal regions can be taken to commute. Finally, we study in the tensor network context the algebraic entropy formula proposed in \cite{Sahu:2025upe} to correspond to generalized gravitational entropy.

\subsection{The inclusion property}

The most basic property that we proposed for the GW to algebra map in general spacetimes is the inclusion property:
\[
W_1 \subset W_2 \Leftrightarrow \mathcal{A}_{W_1} \subset \mathcal{A}_{W_2} \; .
\]
We would like to understand whether this holds in the tensor network context.

\subsubsection{Inclusion for canonical (non-unital) algebras}

For the canonical map $A \to \mathcal{N}_A$ from extremal regions to boundary von Neumann subalgebras, it is straightforward to show that inclusion of regions leads to algebra inclusion. For extremal regions $A \subset B$, the (assumed) isometries $T_A$ and $T_B$ are related by
\[
T_{A} = T_{B} T_{BA}~,
\]
where $T_{BA}$ is the isometry produced by the tensor sub-network associated to the region $B \setminus A$. Any operator in $\mathcal{N}_A$ takes the form
\[
T_A {\cal O}_A T_A^\dagger + \lambda \mathbb{1} = T_{B} T_{BA} {\cal O}_A T_{BA}^\dagger T_B^\dagger + \lambda \mathbb{1}  = T_{B} {\cal O}_A^{(B)} T_B^\dagger + \lambda \mathbb{1}~, 
\]
where we have defined ${\cal O}_A^{(B)} = T_{BA} {\cal O}_A T_{BA}^\dagger$.
The algebra $\mathcal{N}_B$ includes all operators of the form of the last expression, so $\mathcal{N}_A \subset \mathcal{N}_B$. The same conclusions hold for the corner subalgebras $\mathcal{C}_A$ (via the same argument without the $\mathbb{1}$ terms).

\subsubsection{Inclusion for factor algebras}

Next, we consider the inclusion property for the  von Neumann factor subalgebras defined in the previous section. In this case, we have a whole family of possible boundary algebras associated with each region, so whether the inclusion property holds might depend on our choices for the algebras. In this section, we will consider a set of regions and try to understand when it is possible for some choices of the operator maps to arrange that the boundary algebras have the same inclusion structure as the regions.

For any set ${\cal S}$ of extremal  subregions of the tensor network (equivalently, subsets of the tensor network vertices) we can define a partial ordering by inclusion. We write $i < j$ to indicate that a region $i$ is included in a region $j$.

To visualize the structure of ${\cal S}$ and its partial ordering, we can introduce a \emph{Hasse diagram}, which is a directed graph with a vertex for each set in ${\cal S}$ and a directed edge $i \to j$ between any two vertices $i$ and $j$ with $i < j$ and with no $k$ such that $i < k < j$.  Figures \ref{fig:Arborescence} and \ref{fig:Diamond} show examples of sets of regions and the corresponding Hasse diagrams. The Hasse diagram of Figure \ref{fig:Arborescence} has the special property of being a tree. In this case, we have a unique path between any two vertices that are connected by the graph. 

\subsubsection{Factor code networks}

A partially ordered set ${\cal S}$ of tensor network regions with associated Hilbert spaces $\mathcal{H}_i$ and isometries $T_{ji}:{\cal H}_i \to {\cal H}_j$ for $i < j$ 
can be understood as a network of quantum codes or {\it code network}.  

We can choose boundary algebras associated to the various regions having the desired inclusion structure if the code network can be given additional structure that we call a factor code network:
\begin{definition}
    A {\it factor code network} is a code network $\{\mathcal{H}_i\},\{T_{ji}\}$ of quantum codes together with a corresponding set of unital operator maps $\varphi_{ji}$ for $i \le j$ satisfying
    \[
    \varphi_{kj} \circ \varphi_{ji} = \varphi_{ki}  \qquad i \le j \le k \;,
    \]
    and
    \[
    \varphi_{ji}({\cal O}) T_{ji} = T_{ji} {\cal O}  \qquad \forall {\cal O} \in B(\mathcal{H}_i)\; .
    \]
\end{definition}
It follows immediately from Proposition \ref{factor} that
\begin{proposition}
    In a factor code network, the algebras $\mathcal{M}_i^{(j)} \equiv \varphi_{ji}(B(\mathcal{H}_i))$ are unital von Neumann factor subalgebras of $B(\mathcal{H}_j)$ with the property that $\mathcal{M}_i^{(k)} \subset \mathcal{M}_j^{(k)}$ for $i < j < k$.
\end{proposition}
In particular, if we have a maximal element $0$ of ${\cal S}$ corresponding to the full tensor network, we have a partially-ordered set of von Neumann factor subalgebras of $B(\mathcal{H}_0)$ whose inclusion structure is dictated by the ordering.

\begin{figure}
    \centering
\includegraphics[width=0.5\linewidth]{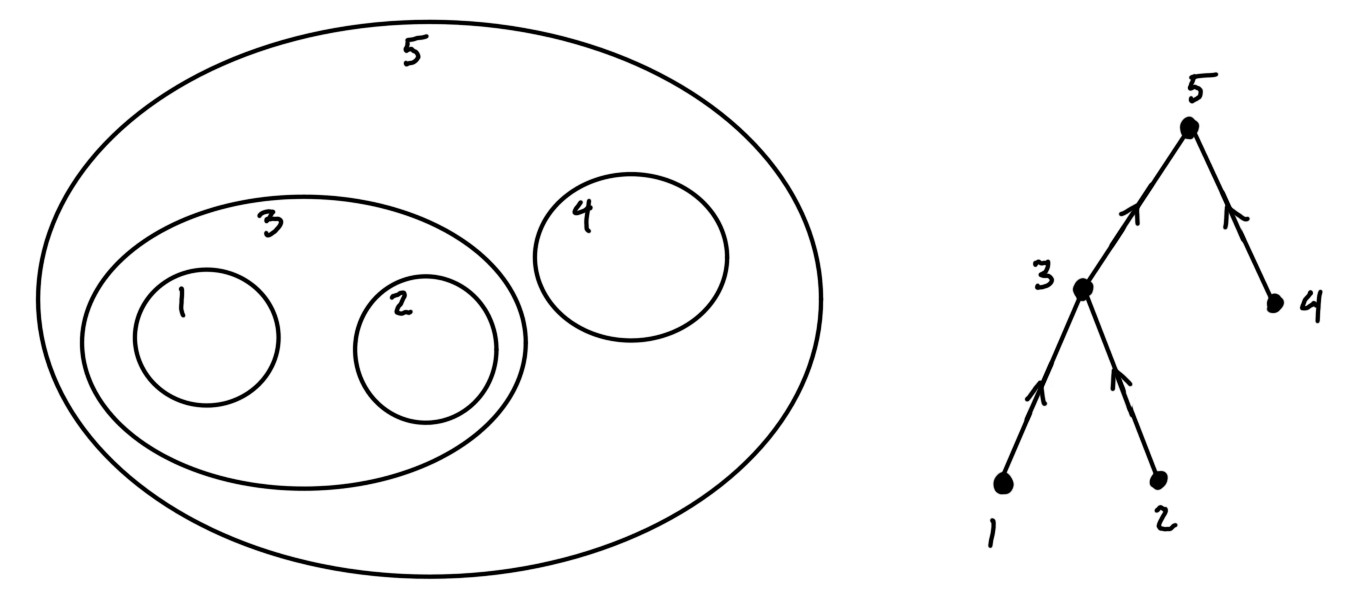}
    \caption{A collection of sets with a partial order by inclusion, and the corresponding Hasse diagram.}
    \label{fig:Arborescence}
\end{figure}

\subsubsection*{Inclusions for tree networks}

If we have a set of tensor network regions whose Hasse diagram forms an tree, and we assume that the Hilbert spaces associated with each vertex have isometries between them according to the directed edges (as we expect in the large $D$ limit of the tensor network), we now show that it is possible to choose isometries as above to associate subalgebras of the largest algebra (associated with the root) to each node with an inclusion structure that matches that of the subsets.

\begin{proposition}{\it Consider a set of Hilbert spaces $\{\mathcal{H}_i\}$ indexed by a partially ordered set whose Hasse diagram is a tree. Suppose that for $i < j$,  ${\rm dim}(\mathcal{H}_i)$ divides ${\rm dim}(\mathcal{H}_j)$ and we have an isometry $T_{ji}:\mathcal{H}_i \to \mathcal{H}_j$ such that $T_{ki} = T_{kj} T_{ji}$ whenever $i<j<k$. Then we can define a set of operator maps  $\varphi_{ji} : B(\mathcal{H}_i) \to B(\mathcal{H}_j)$ with the properties that 
\begin{enumerate}
\item 
$\varphi_{ji}$ is a unital injective *-homomorphism~.
    \item 
$
\varphi_{ki} = \varphi_{kj} \circ \varphi_{ji}
$ \,
for \, $i < j < k$ ~.
\item
$
T_{j i} \mathcal{O}_i |\Psi_i \rangle  = \varphi_{ji}(\mathcal{O}_i) T_{j i} |\Psi_i \rangle 
$\, for\, ${\cal O}_i \in B(\mathcal{H}_i)$~.
\item 
The algebras $\mathcal{M}_i^{(j)} \equiv \varphi_{ji}(B(\mathcal{H}_i))$ are von Neumann factor subalgebras with the property that $\mathcal{M}_i^{(k)} \subset \mathcal{M}_j^{(k)}$ for $i < j < k$~.
\end{enumerate}.}
\label{tree}
\end{proposition}
This assertion is equivalent to Proposition 4 in \cite{FactorCode} and is proved there, but the basic construction is as follows. For any $i < j$ with intermediate vertices, we have a unique path $i \to k_1 \to \cdots \to k_n \to j$ in the Hasse diagram, since we are assuming this is a tree. We can define $m_{ij} = m_{i k_1} m_{k_1 k_2} \cdots m_{k_n j}$ isometries 
\begin{equation} \label{eq:compositions}
T^{(a_1 \cdots a_{n+1})}_{ji} = T^{(a_1)}_{j k_n} T^{(a_2)}_{k_n k_{n-1}} \cdots T^{(a_{n+1})}_{k_1 i} \; .
\end{equation}
Again, these map to orthogonal subspaces. Finally, for any $i < j$ we define the operator map
\[
\varphi_{ji}({\cal O}_i)= \sum_a T^{(a)}_{ji} {\cal O}_i (T^{(a)}_{ji})^\dagger \; .
\]

\begin{figure}
    \centering
\includegraphics[width=0.5\linewidth]{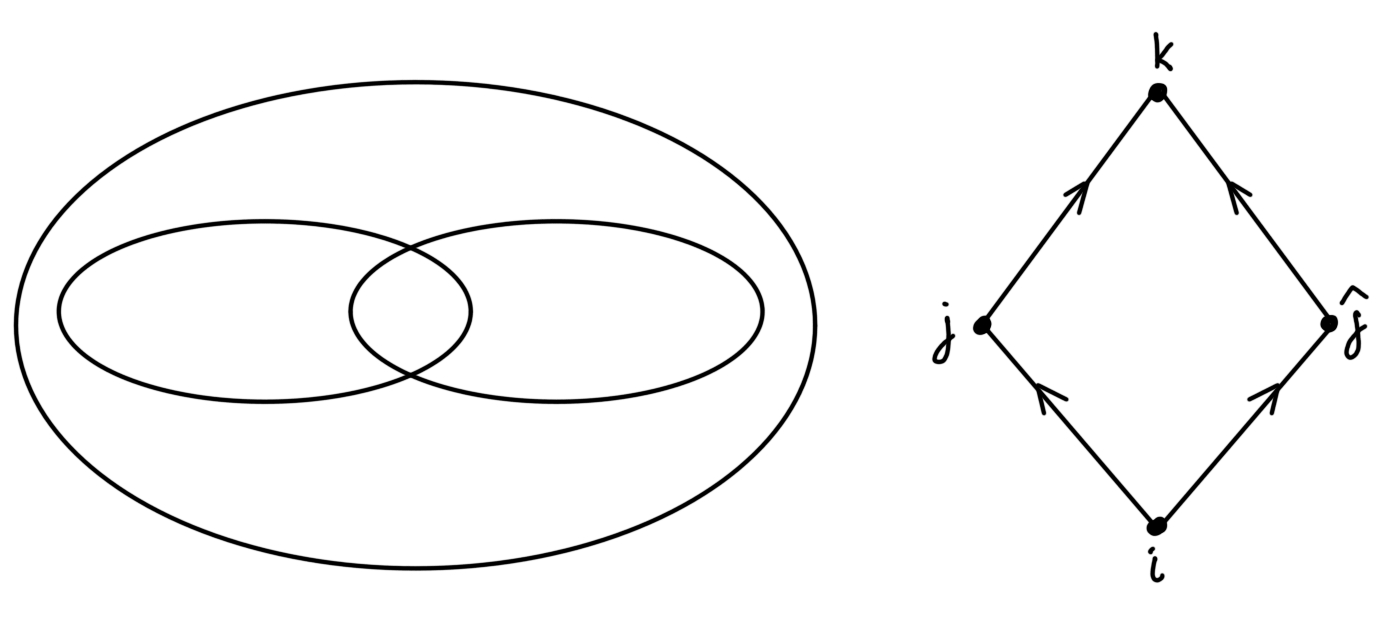}
    \caption{Sets $j$, $\hat{j}$, $i = j \cap \hat{j}$ and $k \supset j, \hat{j}$ with a diamond-shaped Hasse diagram.}
    \label{fig:Diamond}
\end{figure}

\subsubsection*{Inclusions for non-tree networks}

Next, we will show that such a construction of operator maps and nested subalgebras satisfying the properties of a factor code network is not always possible when we consider sets whose Hasse diagram is not a tree graph. We focus on the diamond example of Figure \ref{fig:Diamond} region inclusions $1 < 2,3 < 4$ of region inclusions. Here, it is possible to give precise necessary and sufficient conditions on when the promotion to a factor code network is possible. The success of the construction depends only on how the subspaces $T_{42} \mathcal{H}_2$ and $T_{43} \mathcal{H}_3$ sit inside $\mathcal{H}_4$ (see Figure \ref{fig:DiamondMaps}).

\begin{figure}
    \centering
\includegraphics[width=0.5\linewidth]{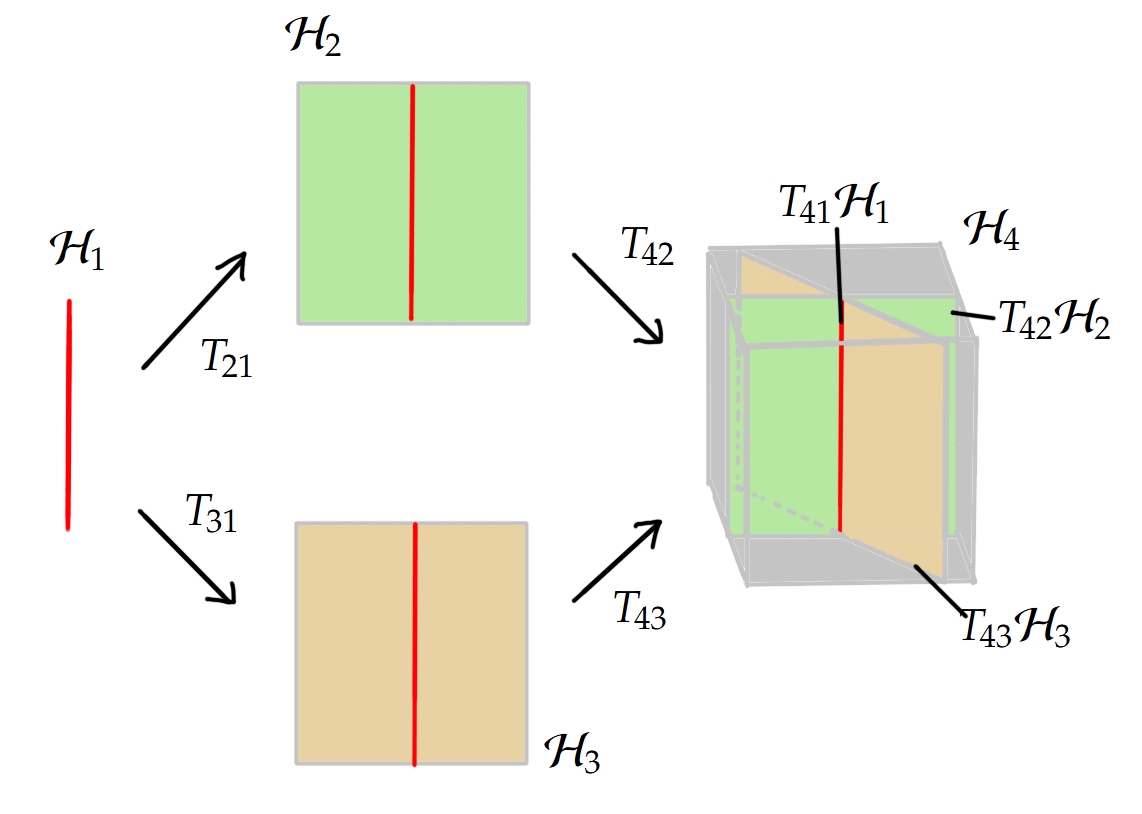}
    \caption{Hilbert spaces and isometries for a diamond-shaped Hasse diagram. The isometries can be promoted to a set of operator maps satisfying the properties of Proposition 1 if and only if the princpal angles between the two subspaces  $T_{42} \mathcal{H}_2$ and $T_{43} \mathcal{H}_3$ have degeneracies that are multiples of $\dim(\mathcal{H}_1)$.}
    \label{fig:DiamondMaps}
\end{figure}

The invariant geometric data describing the relative orientation of the two subspaces $A = T_{42} \mathcal{H}_2$ and $B = T_{43} \mathcal{H}_3$ is a set of {\it principal angles}. To define these, we note that it is always possible to divide the space into a direct sum of subspaces
\[
(A \cap B )\oplus (A \cap B^\perp ) \oplus (B \cap A^\perp )\oplus (\text{span}(A,B)^\perp )\oplus_i H_i \; ,
\]
where $H_i$ is two-dimensional, $\perp$ denotes orthogonal complement, and $A \cap H_i$ and $B \cap H_i$ are one-dimensional subspaces with some relative angle $0 \le \theta_i \le \pi/2$. To this set of angles, we add $\theta = 0$ with multiplicity $\dim (A \cap B)$ and $\theta = \pi/2$ with multiplicity $\min(\dim A \cap B^\perp, \dim B \cap A^\perp)$ to define the full set of principal angles of size $\min (\dim A, \dim B)$.

The principal angles can be defined algebraically considering the projectors $P_2 = T_{42} T_{42}^\dagger$ and $P_3 = T_{43} T_{43}^\dagger$. The set $\{\lambda_i\}$ of nonzero eigenvalues of $P_2 P_3 P_2$ (including multiplicities) corresponds to the set of principal angles with $\theta < \pi/2$ via $\lambda_i = \cos^2 \theta_i$.\footnote{For example, with two one-dimensional Hilbert spaces embedded in a higher dimensional space, we have $P_2 = |v_2 \rangle \langle v_2|$ and $P_3 = |v_3 \rangle \langle v_3|$ so $P_2 P_3 P_2 = \cos^2(\theta) |v_2 \rangle \langle v_2|$ where $\cos^2(\theta) = |\langle v_2 |v_3 \rangle|^2$, so we have a single angle characterizing the relative orientation, as expected.} 

The following result, proved in \cite{FactorCode}, shows that the construction of a set of operator maps yielding boundary subalgebras with the desired inclusion structure  is possible if and only if these angles each come with multiplicity that is a multiple of $\dim(\mathcal{H}_1)$:

\begin{proposition}{\it Consider a set of Hilbert spaces $\{\mathcal{H}_i\}$ associated with sets $i=1,2,3,4$ ordered as $1 < 2,3 < 4$, such that $\dim \mathcal{H}_i$ divides $\dim \mathcal{H}_j$ for $i < j$. Suppose we have isometries $T_{ji}:\mathcal{H}_i \to \mathcal{H}_j$ such that $T_{41} = T_{42} T_{21} = T_{43} T_{31}$. Then it is possible to construct a set of operator maps satisfying the conditions of Proposition \ref{tree} if and only if the principal angles between $T_{42} \mathcal{H}_2$ and $T_{43} \mathcal{H}_3$ have multiplicities which are multiples of $\dim(\mathcal{H}_1)$.}
\label{diamond}
\end{proposition}

\subsubsection{Inclusion for tensor network algebras}

We would now like to understand when the condition on principal angles holds for the code subspaces corresponding to regions in a tensor network with a diamond inclusion structure. Consider extremal regions in a tensor network with inclusion structure
\[
C \subset A,B \subset D \; .
\]
By the same argument as for generalized entanglement wedges in \cite{Sahu:2025upe}, for any extremal regions $A$ and $B$, the intersection $A \cap B$ is also extremal (see Appendix \ref{app:intersection} for more details in the tensor network context). From this it follows that the intersection $E$ of all extremal regions containing $A \cup B$ is also extremal, giving a smallest extremal region containing $A \cup B$. We therefore have an inclusion 
\[
C \subset A \cap B \subset A,B \subset E \subset D \; 
\]
of extremal regions. In the limit of large bond dimension, the tensor network maps from $C$ to $A \cap B$ and from  $D$ to $E$ become isometries, so if we can promote the code network associated with regions 
\[
A \cap B \subset A,B \subset E
\]
to a factor code network, the same should be true for the larger network of inclusions including $C$ and $D$.

In order to investigate whether the necessary and sufficient conditions of Proposition \ref{diamond} are likely to be satisfied, we need to understand the spectrum of $P_A P_B P_A$, where $P_A = T_AT_A^{\dagger}$ and $P_B=T_{B}T_B^{\dagger}$. The information about the spectrum is contained in the moments $\tr\left[(P_A P_B)^n\right]$. Indeed, assuming that $P_A$ and $P_B$ are projectors we have:
\begin{equation} \label{eq:PAPBdegeneracies}
\tr\left[(P_AP_B)^n\right] = \tr\left[(P_AP_BP_A)^n\right]=\sum_{i}d_i\lambda_i^n~,
\end{equation}
where $d_i$ is the multiplicity of the eigenvalue $\lambda_i=\cos^2 \theta_i$. To meet the conditions of Proposition \ref{diamond}, the multiplicities $d_i$ need to be an integer multiple of $D_{A\cap B}=\dim \mathcal{H}_{A\cap B}$. 

We can deduce information about the spectrum of $P_A P_B P_A$ for a typical instance of the tensor network by computing its averaged moments. We will now see that in the case where $A\cup B$ is extremal, these calculations suggest that in the large $D$ limit, $P_A P_B$ tends to a projector with a rank that is a multiple of $D_{A \cap B}$, in this case, the conditions of Proposition \ref{diamond} are satisfied. In the case where the intersection geometry is as in Figure \ref{fig:ABoverlap}, with no edges connecting $A - A\cap B$ and $B - A \cap B$ directly, we show directly using Corollary \ref{overlap} that the conditions of Proposition \ref{diamond} are satisfied so that (assuming an isometric tensor network) we have boundary algebras with the same inclusion structure as the regions.

In the case where $A \cup B$ is not extremal, our calculations do not provide evidence that the conditions of Proposition \ref{diamond} are satisfied. 

\subsubsection*{Spectral moment calculations}

In order to perform the desired calculations for averaged traces of projectors, we will make use of the following general result, described earlier in \cite{Penington:2022dhr,Hayden:2016cfa} and reviewed in \cref{app:haar}:

\paragraph{Master formula for tensor network averages}

Consider the quantity
\begin{equation} \label{eq:proj_moments}
\prod_{\alpha=1}^{m}
\operatorname{tr}\!\left(
P_{A_{i_{\alpha,1}}}
P_{A_{i_{\alpha,2}}}\cdots
P_{A_{i_{\alpha,k_\alpha}}}
\right)
\end{equation}
where 
\[
(i_{1,1}, \dots, i_{1, k_1},i_{2,1}, \dots, i_{2, k_2},\dots,i_{m,1}, \dots, i_{m, k_m}) = (1,2,\dots,n) \; ,
\]
and $P_A$ is the projector associated with the region $A$; the regions $A_i$ may be distinct or equal. In the limit of large bond dimension $D$, the leading order contribution to the average of this quantity can be computed by minimizing the energy of a certain ``spin'' model, where the spins are elements of the permutation group $S_n$ where $n = \sum_\alpha k_\alpha$ is the total number of projectors. The permutations live on the vertices of the tensor network graph which we denote as $\verts$ as well as a set of external vertices $\extv$ at the ends of the external legs. For each external vertex, the permutations is constrained to be 
\[
\kappa = \prod_{\alpha=1}^m (i_{\alpha,1}, \dots, i_{\alpha, k_\alpha}) \; ,
\]
determined by the trace structure in \eqref{eq:proj_moments}. The permutation at vertex $v$ is constrained to be an element of the subgroup
\[
H_v
=
\left\{
\sigma\in S_n:
\sigma(i)=i
\ \text{whenever}\ 
v\in A_i
\right\} \; .
\]
The energy of a permutation configuration is defined to be the sum over all edges of the ``Cayley distance''
$$
d_{\mathrm C}(\sigma,\rho)
\equiv
n-\#\operatorname{cyc}\!\left(\sigma^{-1}\rho\right),
$$
between the permutations $\sigma$ and $\rho$ at the two endpoints of the edge, where \(\#\operatorname{cyc}(\sigma)\) denotes the number of cycles of \(\sigma\), including fixed points. Equivalently, \(d_{\mathrm C}(\sigma,\rho)\) is the minimum number of transpositions required to transform \(\sigma\) into \(\rho\).

The minimum possible energy is defined by
$$
\mathcal E_{\kappa}(A_1,\ldots,A_n)
=
\min_{\substack{\pi_v\in H_v,\;v\in \verts\\
                 \pi_b=\kappa,\;b\in \extv}}
\;
\sum_{(x,y)\in\widehat E}
d_{\mathrm C}(\pi_x,\pi_y)~, 
$$
where the sum is over the extended set of edges $\widehat E$ which includes both the edges between internal vertices as well as external dangling edges. 

We have that
\begin{equation}
\begin{aligned}
&\overline{
\prod_{\alpha=1}^{m}
\operatorname{tr}\!\left(
P_{A_{i_{\alpha,1}}}
P_{A_{i_{\alpha,2}}}\cdots
P_{A_{i_{\alpha,k_\alpha}}}
\right)}
\\[2mm]
&\hspace{2cm}
=
D^{\,\sum_{i=1}^{n}|\partial A_i|
   -\mathcal E_{\kappa}(A_1,\ldots,A_n)}
\left(
g_\kappa+O(D^{-1})
\right)~, 
\end{aligned} 
\label{eq:master-formula1}
\end{equation}
where \(g_\kappa\) is the number of minimizing permutation configurations.

\paragraph{Note on normalisation.}
From the definition of the isometry $\hat T_A$ in \eqref{eq:normalized_TA},
the normalisation of each projector also depends on the random tensors. In
writing \eqref{eq:master-formula1} we have assumed that the average of the
trace of the normalized projectors can be replaced by the product of the
average of the unnormalized trace and the average of the normalization. As
we discuss in \cref{app:ratio-of-averages approximation}, this is justified because the covariance
of the unnormalized trace and the normalization is subleading in $D$.

\subsubsection*{Diamond inclusions for $A \cap B \subset A,B \subset A \cup B$ with extremal $A \cup B$}

Using the general result for averages we can now show the following in the case that $A\cup B$ is {\it strongly extremal}, which we define to mean that no region containing it has a smaller or equal number of boundary legs: 
\begin{proposition}
\label{prop:overlap-projectors-extremality}
Let \(A\)
and \(B\) be strongly extremal regions with associated operators \(P_A=T_AT_A^{\dagger}\) and \(P_B=T_BT_B^{\dagger}\). If the region \(A\cup B\) is strongly extremal, then for every \(n\in\mathbb N\) we have 
\begin{equation}
\overline{\tr\!\left[(P_AP_B)^n\right]}
=\frac{D_AD_B}{D_{A\cup B}}\left[1+\mathcal{O}(1/D)\right]~,
\label{eq:projector-overlap-extremal}
\end{equation}
where for example $D_A = \dim {\cal H}_A$ and $\mathcal{O}(1/D)$ denotes terms that are subleading in the large bond dimension limit.
\end{proposition}
\begin{proof}
Making use of the general result above, in order to compute the average, we need to find an assignment of permutations to the vertices of the network such that the Cayley distance energy defined above is minimized. In this case, since we have only a single trace, the fixed permutation on the boundary vertices is 
\[
\kappa = c \equiv (1 2 \cdots 2n) \; .
\]
The subgroups $H_v$ that restrict our choice of permutations are the $H_v = \{1\}$ for $v \in A \cap B$, and $H_v$ equal to the permutations acting only on even and odd elements for $v \in A$  and $v \in B$ respectively. As we show in \Cref{app:averageTN}, the energy-minimizing configuration sets all permutations in $\overline{A \cup B}$ (the complement of $A\cup B$ in the tensor network) to $\kappa$, all permutations in $A \setminus B$ to $c_\text{even} = (2 \, 4 \cdots 2n)$, all permutations in $B \setminus A$ to $c_\text{odd} = (1\, 3 \cdots 2n-1)$, and those in $A \cap B$ to the identity, as shown in Figure \ref{fig:extremal-nonextremal-comparison}. Using that
\begin{align*}
d_C(c,c_\text{even}) &= d_C(c,c_\text{odd}) = n~, \quad d_C(c,\mathbb{1}) = 2n-1~, \\
d_C(\mathbb{1},c_\text{even}) &= d_C(\mathbb{1},c_\text{odd}) = n - 1~, \quad d_C(c_\text{odd},c_\text{even}) = 2n-2~,
\end{align*}
and defining $|\partial(A_1,A_2)|$ to be the number of tensor network edges connecting adjacent non-overlapping regions $A_1$ and $A_2$, we then have 
\begin{align*}
\mathcal{E}_{\kappa} &= n |\partial(A \setminus B,\overline{A \cup B})| + n |\partial(B \setminus A,\overline{A \cup B})| + (2n-1) |\partial(A \cap B,\overline{A \cup B})| + (2n-2) |\partial(A \setminus B, B \setminus A)| \\
& \hspace{10pt} + (n-1) |\partial(A \setminus B,A \cap B)| + (n-1) |\partial(B \setminus A,A \cap B)|~.
\end{align*}
Using the general result above we find that up to order $1/D$ corrections, 
\begin{align}
\overline{\tr\!\left[(P_AP_B)^n\right]}
    &=
    D^{n|\partial A| + n|\partial B| - \mathcal{E}_\kappa}  \cr
   &= D^{|\partial(A \cap B,\overline{A \cup B})| + 2 |\partial(A \setminus B, B \setminus A)| + |\partial(A \setminus B,A \cap B)| + |\partial(B \setminus A,A \cap B)|} \cr
   &= D^{|\partial A| + |\partial B| - |\partial(A \cup B)|}~,
    \label{eq:pApB-n-leading}
\end{align}
where in going to the second line, we have used the geometrical result that
\begin{eqnarray*}
|\partial A| &=& |\partial(A \setminus B,\overline{A \cup B})| + |\partial(A \setminus B,B \setminus A)| + |\partial(A \cap B,B \setminus A)| + |\partial(A \cap B,\overline{A \cup B})| \cr
|\partial B| &=& |\partial(B \setminus A,\overline{A \cup B})| + |\partial(B \setminus A,A \setminus B)| + |\partial(A \cap B,A \setminus B)| + |\partial(A \cap B,\overline{A \cup B})|  
\end{eqnarray*}
and in going to the third line, we have used these results again together with
\[
|\partial(A \cup B)| = |\partial(B \setminus A,\overline{A \cup B})| + |\partial(A \setminus B,\overline{A \cup B})| + |\partial(A \cap B,\overline{A \cup B})| \; .
\]
These geometrical results can be verified from Figure \ref{fig:extremal-nonextremal-comparison}. The claim in the proposition now follows since 
\[
D_A = D^{|\partial A|}~, \qquad D_B = D^{|\partial B|}~, \qquad D_{A \cup B} = D^{|\partial (A \cup B)|}~.
\]
\end{proof}

\begin{figure}[t]
    \centering

    \begin{subfigure}[t]{0.4\textwidth}
        \centering
        \includegraphics[width=\textwidth]{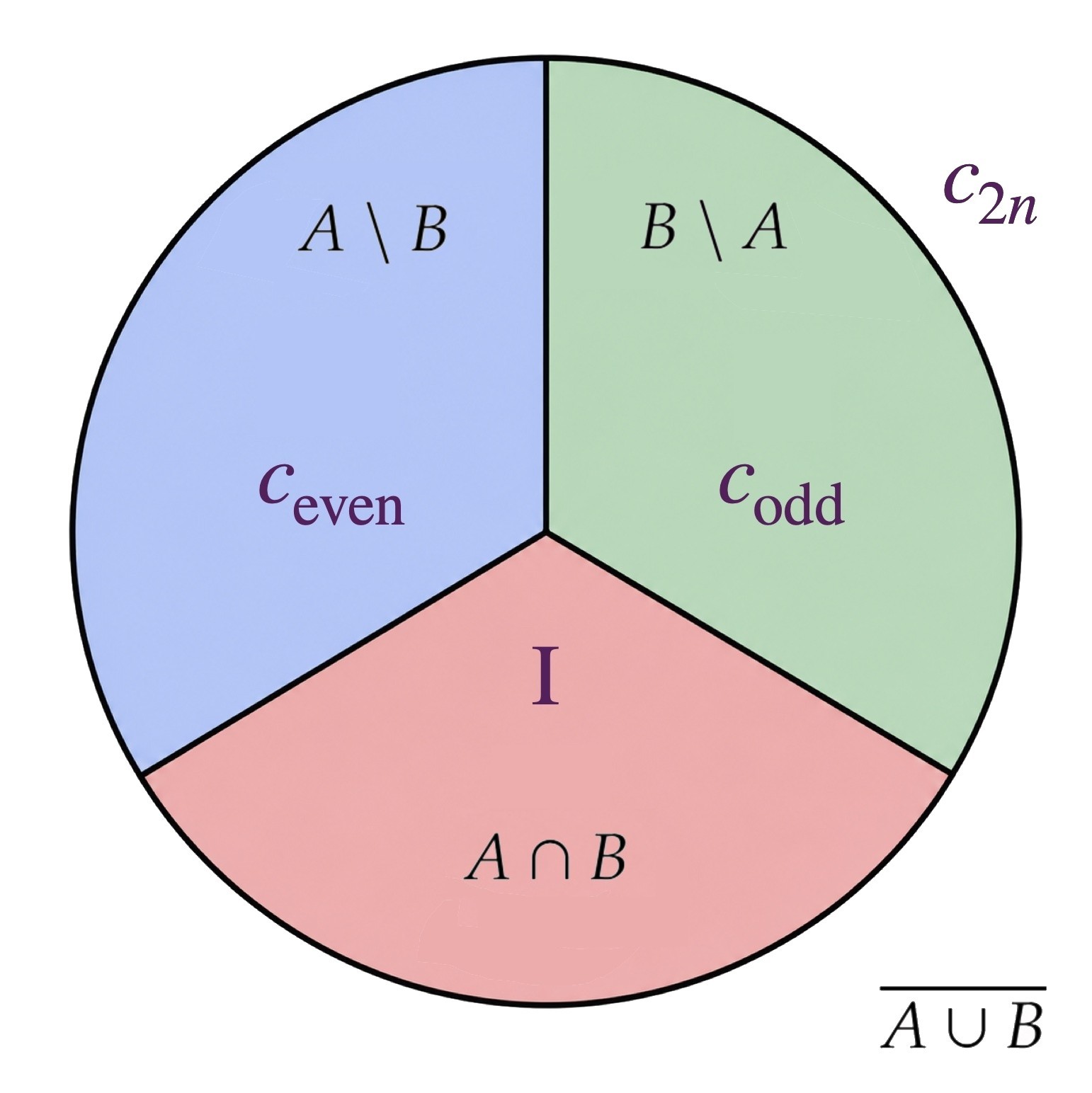}
    \end{subfigure}
    \hfill
    \begin{subfigure}[t]{0.48\textwidth}
        \centering
        \includegraphics[width=\textwidth]{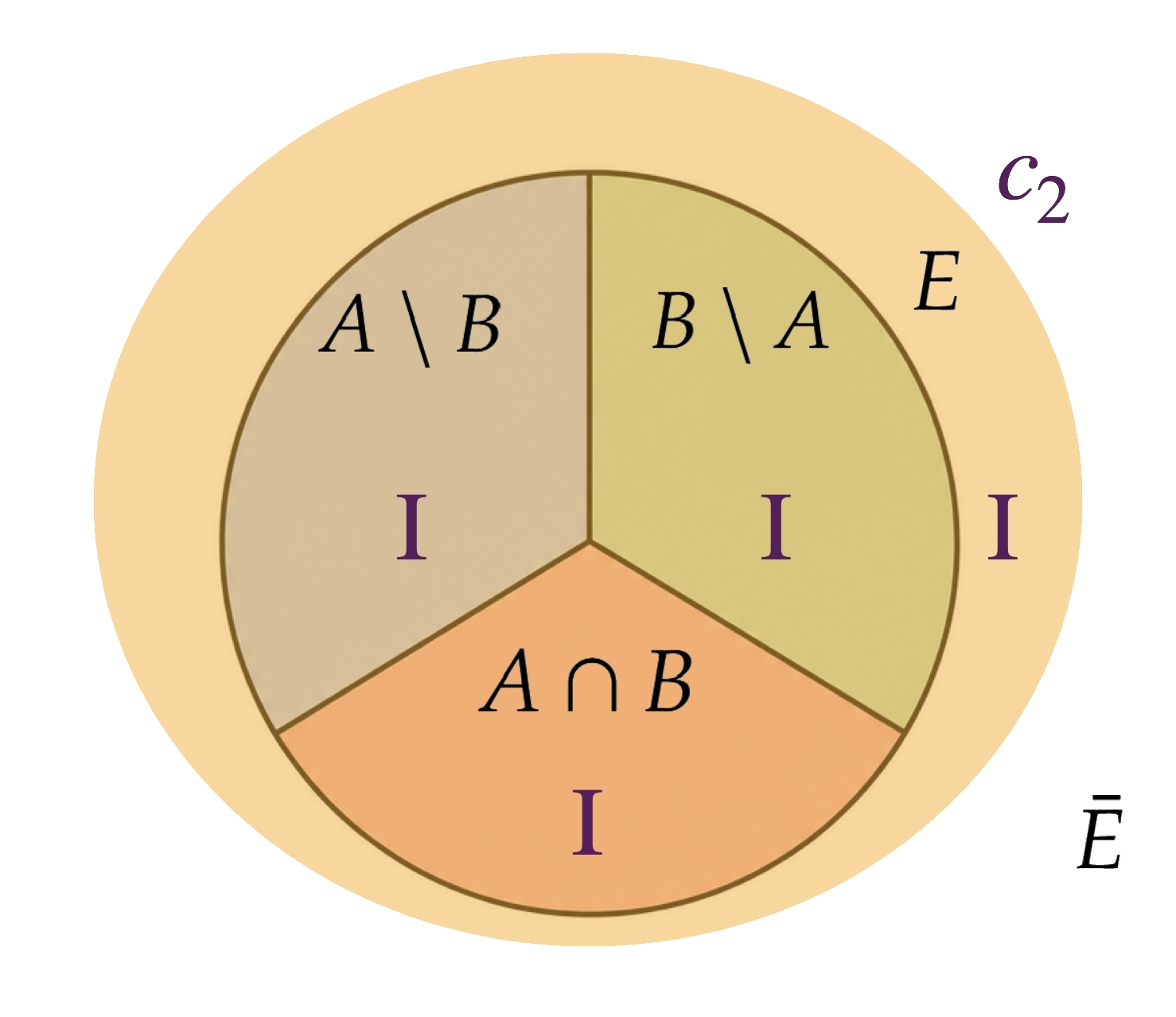}
    \end{subfigure}

    \caption{
        Dominant permutation configurations in calculating $\overline{\tr\!\left[(P_AP_B)^n\right]}$ for extremal and non-extremal
        unions. Left: for an extremal union and general $n$, the regions
        $A \setminus B$ and $B \setminus A$ are filled with the even and odd cyclic
        permutations, respectively, while $A\cap B$ carries $\Id$ and the
        exterior carries $c_{2n}$. Right: for a non-extremal union at $n=1$,
        $c_{\mathrm{odd}}=c_{\mathrm{even}}=\Id$, so the identity domain
        extends to the minimal containing region $E$, and the domain wall
        lies on $\partial E$ rather than on $\partial(A\cup B)$.
    }
    \label{fig:extremal-nonextremal-comparison}
\end{figure}
We can use this result to understand the relative position of $T_A\mathcal{H}_A$ and $T_{B}\mathcal{H}_B$ in $\mathcal{H}_{\partial}$. For extremal $A\cup B$, Proposition \ref{prop:overlap-projectors-extremality} shows that up to order $1/D$ corrections,\footnote{In particular, when $A \setminus B$ and $B \setminus A$ do not share a boundary, i.e. $\partial(A \setminus B,B \setminus A)=\varnothing$, we have the identity
$D_A D_B=D_{A\cup B}D_{A\cap B}$. In that case, the result reduces to $\overline{\tr\!\left[(P_AP_B)^n\right]}
=
D_{A\cap B}$.}
$$
\overline{\tr\!\left[(P_AP_B)^n\right]}=\frac{D_A D_B}{D_{A\cup B}} = D_{A\cap B} D^{2 |\partial(A \setminus B, B \setminus A)|}
$$
is independent of $n$. Similar calculations can be used to establish that the variance of these quantities is small. For example, a replica computation carried out in Appendix \ref{app:averageTN}, shows that the variance of $\tr((P_A P_B)^n)$
goes to zero as $D\to \infty$, i.e. 
$$
\frac{\overline{\tr[(P_AP_B)^n]\tr[(P_AP_B)^n]}}{\left(\overline{\tr[(P_AP_B)^n]}\right)^2} - 1=\mathcal{O}(1/D)~.
$$
Thus, the result for $\tr((P_A P_B)^n)$ in a single instance of the tensor network should be close to the average for any $n$. 

For a single tensor network with $\tr((P_A P_B)^n) = D_{A\cap B}D^{2 |\partial(A \setminus B, B \setminus A)|}$ for all $n$, it follows from \eqref{eq:PAPBdegeneracies} that $P_A P_B$ has a single non-zero eigenvalue $\lambda=1$ with multiplicity $d= D_{A\cap B}D^{2 |\partial(A \setminus B, B \setminus A)|}$, which is divisible by $D_{A\cap B}$. Thus, the conditions of Proposition \ref{diamond} are satisfied for this example.
Our results show that for extremal $A\cup B$ the results for $\tr((P_A P_B)^n)$  should, in a typical instance of the tensor network for large $D$ be close to one satisfying the conditions of Proposition \ref{diamond} such that we have (in some approximate sense for large but finite $D$) algebra inclusions of the form 
$$
\mathcal{M}_{A\cap B}\subset \mathcal{M}_A,\mathcal{M}_B\subset \mathcal{M}_{A\cup B}~.
$$
If we assume an isometric tensor network and a region geometry as in Figure \ref{fig:ABoverlap}, the existence of boundary algebras with these inclusion properties can be shown directly, as we will do below Corollary \ref{overlap}. 

\subsubsection*{Results for non-extremal $A \cup B$}

We now consider the case that $A \cup B$ is \emph{not} extremal, which means that there is a region \(E\) with a smaller number of boundary legs than $A\cup B$ that properly contains $A\cup B$. We have the following result:
\begin{proposition}
\label{prop:overlap-projectors-variance}
Let \(A\)
and \(B\) be extremal regions with associated operators \(P_A=T_AT_A^{\dagger}\) and \(P_B=T_BT_B^{\dagger}\). Let $E$ be the extremal region defined by the intersection of all extremal regions containing $A\cup B$. Assuming that $E$ is strongly extremal, we have
\begin{equation}
\overline{\tr\!\left[P_AP_B\right]}
=\frac{D_AD_B}{D_{E}}\left[1+\mathcal{O}(1/D)\right]~,
\label{eq:projector-overlap-extremal}
\end{equation}
where $\mathcal{O}(1/D)$ denotes terms that are subleading in the large bond dimension limit.
\end{proposition}
\begin{proof}

To see this, note that for $n=1$ the odd- and even-replica permutations are
both trivial, $c_{\rm odd} = c_{\rm even} = \Id$ and one can extend the identity domains beyond
\(A\) and \(B\) towards the minimal  region \(E\). Thus, the leading saddle will consist of an identity domain inside $E$, and a cyclic permutation domain outside. The result follows via calculations analogous to those above.
\end{proof}

\begin{figure}[t]
    \centering
    \includegraphics[
        width=0.4\textwidth,
        trim={0 0 0 0},
        clip
    ]{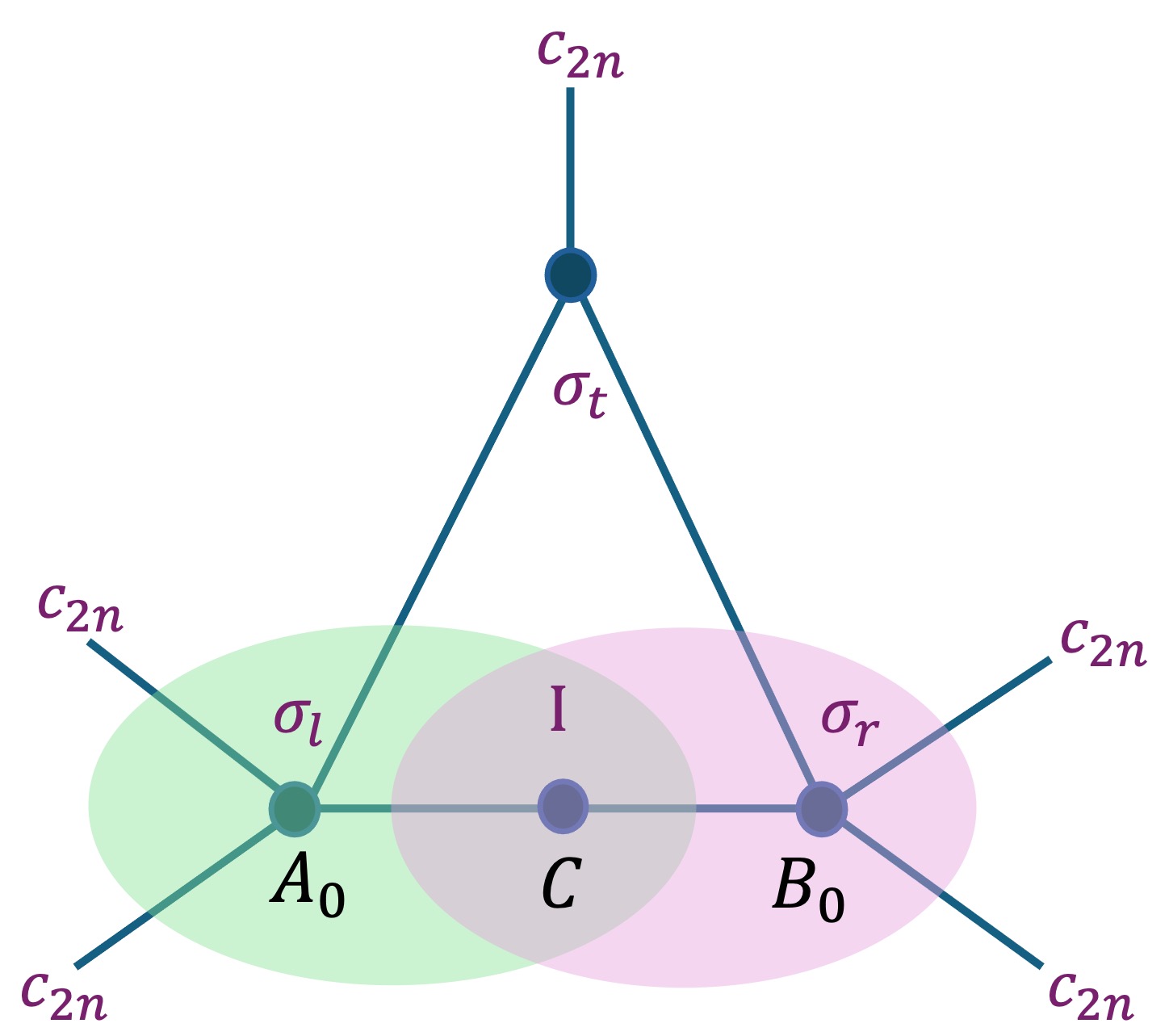}

    \caption{
        An example of a tensor network with extremal regions $A$ (green) and $B$ (purple) such that $A\cup B$ is not extremal. We have denoted $A_{0}=A \setminus B$ and $B_0 = B\setminus A$. We assume that all legs contribute the same bond dimension $D$. We have furthermore labeled the vertices to indicate the allowed permutation assignments used in the calculation of $\overline{\tr\!\left[(P_AP_B)^n\right]}$ (see Appendix \ref{app:nonextremal-counterexample}): The physical boundary is fixed to be $c_{2n}$ and the intersection vertex is fixed to the identity $\Id$. The top permutation $\sigma_t$ is summed over $S_{2n}$, while $\sigma_\ell$ and $\sigma_r$ are summed over the subgroups acting on the even and odd replicas, respectively.
    }
    \label{fig:nonextremal-counterexample}
\end{figure}

Furthermore, when $A\cup B$ is not extremal, the higher moments $\overline{\tr\!\left[(P_AP_B)^n\right]}$ need not be independent of $n$. To demonstrate explicitly that the moments can depend on \(n\), we consider the simple tensor network example shown in Figure \ref{fig:nonextremal-counterexample}. We assume that all bond dimensions are equal to \(D\). The regions $A$ and $B$ each cut four legs, while $A\cup B$ cuts six. The larger region $E\supset A\cup B$, obtained by including the top tensor, cuts only five legs. Thus
\[
D_A=D_B=D^4~,
\qquad
D_E=D^5<D_{A\cup B}=D^6~,
\]
which shows that $A\cup B$ is not extremal.

For this network, equation~\eqref{eq:projector-overlap-extremal} gives
\[
\overline{\tr\left[P_AP_B\right]} \sim \frac{D_AD_B}{D_E} = \frac{D^8}{D^5} = D^3~, 
\]
to leading order in the large bond dimension limit. Using a replica computation presented in Appendix~\ref{app:nonextremal-counterexample}, we show that the dominant contribution to the second moment is given by \[
\overline{\tr\!\left[(P_AP_B)^2\right]} \sim 5D^2~. 
\]
We therefore find \begin{equation} \overline{\tr\left[P_AP_B\right]}\neq \overline{\tr\!\left[(P_AP_B)^2\right]}~, \end{equation} showing that the leading large-\(D\) moments are \(n\)-dependent in this non-extremal example. Hence, the \(n\)-independence found when \(A\cup B\) is extremal does not extend to general non-extremal unions.

Given the $n$-dependence of the moments, the argument that we used to verify the conditions of Proposition \ref{diamond} when $A\cup B$ is extremal no longer applies. Therefore, there could possibly be an obstruction to finding a consistent inclusion of the algebra $\mathcal{M}_{A\cap B}$ in $\mathcal{M}_E$ through $\mathcal{M}_A$ and $\mathcal{M}_B$ respectively. At this point, we have not found conclusive evidence that the conditions in Proposition \ref{diamond} are generally satisfied. 

It would be interesting to see if we could get more detailed information about the principal angles and their degeneracies of by a more thorough investigation of the tensor network averaged quantity $\tr\left[(P_AP_B)^n\right]$.\footnote{For example, since all eigenvalues satisfy  $\lambda_i\in [0,1]$ the $n$-th moment for large $n$ will be dominated by the largest eigenvalue, i.e. $\overline{\tr\left[(P_A P_B)^{n}\right]}\sim d_{\max }\lambda_{\max}^n$. We thus find that the following ratio of moments extracts the leading order eigenvalue:
$$
\frac{\overline{\tr\left[(P_A P_B)^{n+1}\right]}}{\overline{\tr\left[(P_A P_B)^{n}\right]}} \sim \lambda_{\max}~, \qquad n\to \infty~.
$$} We leave such an investigation for future work. 

\subsection{The union property}

For generalized entanglement wedges $W_1$ and $W_2$, the union $W_1 \cup W_2$ is not generally a GW, but the intersection of all GWs containing $W_1 \cup W_2$ is a GW that we call $W_1 \vee W_2$. By its definition, $W_1 \vee W_2$ is contained in any GW that includes $W_1$ and $W_2$. Almost identical statements apply to subalgebras $\mathcal{A}_1$ and $\mathcal{A}_2$ of an algebra: the union $\mathcal{A}_1 \cup \mathcal{A}_2$ is not generally an algebra, but the intersection of all algebras containing $\mathcal{A}_1 \cup \mathcal{A}_2$ is an algebra that that is usually denoted $\mathcal{A}_1 \vee \mathcal{A}_2$ and called the join of $\mathcal{A}_1$ and $\mathcal{A}_2$. This is also the algebra generated by $\mathcal{A}_1$ and $\mathcal{A}_2$.

In \cite{Sahu:2025upe}, we suggested that a natural property for our proposed GW to subalgebra map to have is
\[
W_1 \vee W_2 \Leftrightarrow \mathcal{A}_{W_1} \vee \mathcal{A}_{W_2} \; .
\]
We would like to understand whether this holds in our tensor network context for some choices of operator boundary algebras.

Suppose that we have extremal subsystems $A$ and $B$ of a tensor network and that $A \vee B$ is the extremal region containing $A \cup B$ with the smallest boundary area. If the region $T$ represents the full network, the sets $A,B \subset A \vee B \subset T$ have a tree-like inclusion structure, so (assuming isometries for the various tensor network maps) we can define operator maps and associate boundary factor algebras in a consistent way for the various regions to give a factor code network. We would like to understand if the algebra associated with $A \vee B$ is the one generated by the algebras associated with $A$ and $B$, 
\begin{equation}
\label{union1}
\mathcal{M}_{A\vee B} = \mathcal{M}_A\vee \mathcal{M}_B~.
\end{equation}

In \cite{FactorCode}, we prove the following general result:

\begin{proposition}
    Suppose we have Hilbert spaces ${\cal H}_A$, ${\cal H}_B$, and ${\cal H}$ with isometries $T_A : {\cal H}_A \to {\cal H}$ and  $T_B : {\cal H}_B \to {\cal H}$. 
    In the space of possible extensions of the isometries to operator maps $\varphi_A$, $\varphi_B$ satisfying the properties of Proposition \ref{tree}, a Haar random choice will have the property that the subalgebras $\mathcal{M}_A = \varphi_A(B(\mathcal{H}_A))$ and $\mathcal{M}_B= \varphi_B(B(\mathcal{H}_B))$ together generate $B(\mathcal{H})$ unless \begin{itemize}
        \item $\dim \mathcal{H}_A = 1$ and $  \dim \mathcal{H}_B < \dim \mathcal{H}$~.
        \item $\dim \mathcal{H}_B = 1$ and $ \dim \mathcal{H}_A < \dim \mathcal{H}$~.
        \item
        $T_A \mathcal{H}_A = T_B \mathcal{H}_B \ne {\cal H}$~.
        \item
        $T_A \mathcal{H}_A$ is the orthogonal complement of $T_B \mathcal{H}_B$~.
    \end{itemize}
     In all of these cases, no choices for $\mathcal{M}_A$ and $\mathcal{M}_B$ generate $B(\mathcal{H})$.    
\end{proposition}

In the tensor network context, assuming an isometric tensor network, we can apply this result to the case where ${\cal H}_A$ and ${\cal H}_B$ are Hilbert spaces associated with regions $A$ and $B$, $\mathcal{H}$ is the Hilbert space associated with the extremal region $A \vee B$, and $T_A$ and $T_B$ are the tensor network isometries from $A$ and $B$ to $A \vee B$. 

For non-trivial regions, we certainly have $\dim \mathcal{H}_A, \dim \mathcal{H}_B > 1$. The fourth special case requires $\dim \mathcal{H} = \dim \mathcal{H}_A + \dim \mathcal{H}_B$, which will never hold in a tensor network for bond dimension greater than 2 since all the Hilbert space dimensions are some power of the bond dimension and $D^a + D^b = D^c$ has no solutions for $D > 2$ and positive integer $a,b,c$. The third special case in the proposition corresponds to $\tr((P_A P_B)^n) = \tr(P_A) = \tr(P_B)$. We can use our results above to show that this will typically not hold in the case of large bond dimension.

If $A \vee B \ne A \cup B$, we have that on average to leading order in large bond dimension (see Proposition \ref{prop:overlap-projectors-variance}), 
\[
\overline{\tr(P_A P_B)} = {D_A D_B \over D_{A \vee B}}~.
\]
For isometries, $\tr(P_A) = D_A$ and $\tr(P_B) = D_B$, so $\tr(P_A P_B) = \tr(P_A) = \tr(P_B)$ gives $\dim {\cal H}_A = \dim {\cal H}_B = \dim {\cal H}_{A \vee B}$. In this case, we must have $T_A \mathcal{H}_A = T_B \mathcal{H}_B = \mathcal{H}$, violating the last part of condition 3. 

If $A \vee B = A \cup B$, then on average to leading order in large bond dimension, 
\[
\overline{\tr(P_A P_B)} = D_{A \cap B} \; .
\]
Then $\tr(P_A P_B) = \tr(P_A) = \tr(P_B)$ gives $D_A = D_B = D_{A \cap B}$ implying that the boundary lengths for regions $A$, $B$, and $A \cap B$ are equal. Using the general relation $|\partial A| + |\partial B| \ge |\partial A \cap B| + |\partial A \cup B|$ on boundary lengths, we then have
\[
|\partial A| = |\partial B| \ge |\partial A \cup B|
\]
but we must have equality since $A,B \subset A \cup B$ and $A$ and $B$ are extremal. Thus, $D_{A \cup B} = D_A = D_B$, and again we violate the last part of condition 3. 

We conclude that for extremal regions 
\[
A, B \subset A \vee B \subset T
\]
in a tensor network, typical choices of operator maps giving rise to boundary von Neumann factor algebras $\mathcal{M}_A, \mathcal{M}_B \subset
\mathcal{M}_{A\vee B} \subset \mathcal{M}_\partial$ will have the union property (\ref{union1}).

\subsection{The intersection property}

In \cite{Sahu:2025upe}, we showed that the intersection of any two GWs is also a GW. Similarly, the intersection of any two subalgebras of an algebra is another subalgebra. Thus, we suggested in \cite{Sahu:2025upe} that a natural property of the proposed GW to subalgebra map would be the intersection property
\[
\mathcal{A}_{W_1 \cap W_2} = \mathcal{A}_{W_1} \cap \mathcal{A}_{W_2} \; .
\]
In this section, we will understand when the analogous statement is true in the tensor network situation.

Suppose that we have extremal subsystems $A$ and $B$ of a tensor network. Then as we argued earlier, the intersection $A \cap B$ is also an extremal region. In cases covered by Proposition \ref{diamond} above, we can associate subalgebras $\mathcal{M}_A$, $\mathcal{M}_B$, and $\mathcal{M}_{A \cap B}$ of the boundary algebra $\mathcal{M}_\partial$ with the inclusions $\mathcal{M}_{A \cap B} \subset \mathcal{M}_A$, $\mathcal{M}_B \subset \mathcal{M}_\partial$.

We would like to understand when for some choice of the auxiliary isometries that define these boundary algebras, it is possible to obtain additionally the intersection property
\[
\mathcal{M}_A \cap \mathcal{M}_B = \mathcal{M}_{A \cap B} \; .
\]
In order to understand this, we can make use of the following general result that we prove in \cite{FactorCode}:
\begin{proposition}
\label{intersection1}
    Suppose we have Hilbert spaces $\mathcal{H}_i$, $i = 1,..,4$ and isometries $T_{ji}$ satisfying the conditions of Proposition \ref{diamond}. Then by an appropriate choice of isometries, we can ensure that $\mathcal{M}_1 = \mathcal{M}_2 \cap \mathcal{M}_3$ unless one of following holds:
    \begin{itemize}
    \item $\dim \mathcal{H}_1 < \dim \mathcal{H}_2 \le \dim  \mathcal{H}_3 = \dim \mathcal{H}_4$~.
    \item $\dim \mathcal{H}_1 < \dim\mathcal{H}_3 \le \dim \mathcal{H}_2 = \dim \mathcal{H}_4$~.
    \item $\dim \mathcal{H}_1 < \dim \mathcal{H}_2 = \dim \mathcal{H}_3 = \dim \mathcal{H}_4/2$ with $T_{42} \mathcal{H}_2=T_{43} \mathcal{H}_3$ or $T_{42} \mathcal{H}_2=T_{43} \mathcal{H}_3^\perp$~.
    \end{itemize}
    Furthermore, in all of these cases, generic choices of auxiliary isometries lead to the desired intersection property.
\end{proposition}
Now suppose that we have an isometric tensor network with extremal regions $A$ and $B$ and take ${\cal H}_1 = {\cal H}_{A \cap B}$, ${\cal H}_2 = {\cal H}_A$, ${\cal H}_3 = {\cal H}_B$, and ${\cal H}_4$ to be the boundary Hilbert space. In this case, the third special case in the proposition will never hold for bond dimension greater than 2, since all dimensions are powers of the bond dimension and thus $\dim \mathcal{H}_4/ \dim \mathcal{H}_3$ cannot be 2. The first two conditions can only hold if one of the regions $A$ or $B$ has a Hilbert space dimension that is already equal to the full Hilbert space dimension. Thus, we conclude that:
\begin{itemize}
    \item For extremal regions $A$ and $B$ in an isometric tensor network with bond dimension larger than 2 and $|\partial A|, |\partial B|$ smaller than the full boundary, generic choices of the associated boundary algebras satisfying the inclusion property $\mathcal{M}_{A \cap B} \subset \mathcal{M}_A , \mathcal{M}_B$ also satisfy the intersection property
    \[
\mathcal{M}_A \cap \mathcal{M}_B = \mathcal{M}_{A \cap B} \; .
\]
\end{itemize}

\subsection{Algebra commutation}

In a quantum field theory, subalgebras associated with spacelike separated regions commute. For a potential GW to subalgebra map in the gravitational context, one possibility is that this commutation property continues to hold, that is, the subalgebras associated with disjoint GWs on the same spatial slice commute with each other. In this section, we will investigate whether this holds in the tensor network context, finding that a different stronger condition is required for commutation.

In the context of an isometric tensor network, we can make use of the following result that we prove in \cite{FactorCode}:
\begin{proposition}
\label{commuting1}
Let $\mathcal H$ be an $n$-dimensional Hilbert space, and let subspaces
$A,B\subset \mathcal H$ have dimensions $a$ and $b$. Let $P_A$ and $P_B$ be the orthogonal
projections onto $A$ and $B$. There exist commuting von Neumann factor subalgebras
\[
  \mathcal{M}_A \cong M_a(\mathbb{C})~, \quad \mathcal{M}_A A = A~, \qquad \mathcal{M}_B \cong M_b(\mathbb{C})~, \quad \mathcal{M}_B B = B~,
\]
if and only if one of the following holds 
\begin{itemize}
    \item 
     $n = abr, r=1$ and $\text{rank}  (P_AP_B)=1$ and $\|P_AP_B\|_1=1$~, \item 
     $n = abr, r=2,3,\dots$ and $\text{rank}  (P_AP_B) \le r$ and $\|P_AP_B\|_1 \le 1$~,
\end{itemize}
where $\|\cdot\|_1$ denotes the trace norm.
\end{proposition}

We can now apply this result to the tensor network situation. Suppose we have two extremal regions $A$ and $B$ with isometries from the associated Hilbert space to the boundary Hilbert space. In the case where $A$ and $B$ have a non-trivial intersection, Proposition \ref{intersection1} tells us that the algebras $\mathcal{M}_A$ and $\mathcal{M}_B$ typically have a non-scalar intersection and thus do not commute. 

If $A$ and $B$ are disjoint and also non-adjacent, the Hilbert space associated with $A \cup B$ is the tensor product $\mathcal{H}_A \otimes \mathcal{H}_B$.
   We have isometries
   \[
   T_{A \to AB}:|\psi \rangle_A \to  |\psi \rangle_A \otimes |B \rangle~, \qquad T_{B \to AB}:|\psi \rangle_b \to |B \rangle \otimes |\psi \rangle_B~,
   \]
   where $|A \rangle$ and $|B \rangle$ are the states on $\mathcal{H}_A$ and $\mathcal{H}_B$ produced by the tensor network. In this case, we can define operator maps $\phi_{\cup A}:{\cal O}_A \to {\cal O}_A \otimes 1_b$ and $\phi_{\cup B}:{\cal O}_B \to 1_b \otimes {\cal O}_B$ that give commuting algebras $\mathcal{M}_A^{\cup} = M_a(\mathbb{C}) \otimes 1_b$ and $\mathcal{M}_B^{\cup} = 1_a  \otimes M_b(\mathbb{C})$. When $A \cup B$ is extremal, we have an isometry $T_{AB}$ mapping $\mathcal{H}_A \otimes \mathcal{H}_B$ to the boundary Hilbert space $\mathcal{H}$. Extending this to an operator map $\phi_{\partial \cup}$ according to Proposition \ref{tree}, we have boundary subalgebras $\phi_{\partial \cup}(\mathcal{M}_A^{\cup})$ and $\phi_{\partial \cup}(\mathcal{M}_B^{\cup})$ that commute. We have thus shown that:
   \begin{itemize}
       \item If $A$ and $B$ are non-intersecting and non-adjacent extremal regions of an isometric tensor network and $A \cup B$ is extremal, it is possible to define associated operator maps defining commuting boundary von Neumann factor algebras $\mathcal{M}_A$ and $\mathcal{M}_B$. 
   \end{itemize}
   
   If $A \cup B$ is not extremal, there is some extremal region $C$ containing $A \cup B$ whose associated Hilbert space has smaller dimension than the Hilbert space associated with $A \cup B$. Applying Proposition \ref{commuting1} for $\mathcal{H} = \mathcal{H}_C$, $A = T_{CA} \mathcal{H}_A$ and $B = T_{CB} \mathcal{H}_B$, we find that the conditions cannot hold since
 $\dim \mathcal{H}_C < \dim \mathcal{H}_A \dim \mathcal{H}_B$. Thus, if we define operator maps satisfying the properties of Proposition \ref{tree} for the four Hilbert spaces associated with $A,B,C$ and the full tensor network, the algebras associated with $A$ and $B$ will not commute. 

 We can also ask whether it is possible to define commuting von Neumann factor subalgebras of the boundary algebra if we do not require compatibility with the intermediate Hilbert space $\mathcal{H}_C$. In the limit of large bond dimension, we have from our earlier tensor network calculations that $\overline{\tr(P_A P_B)} = {\dim {\cal H}_A \dim {\cal H}_B \over \dim {\cal H}_C} > 1$. Also, since all the eigenvalues of $P_A P_B$ are in $[0,1]$ we have $||P_A P_B||_1 \ge \tr(P_A P_B)$. Thus, typically we will violate the conditions of Proposition \ref{commuting1} and defining commuting algebras will not be possible.

These results motivate a proposal for the potential GW to subalgebra map in the gravity context, that for disjoint wedges $W_1$ and $W_2$, the associated subalgebras $\mathcal{A}_{W_1}$ and $\mathcal{A}_{W_2}$ commute if and only if $W_1 \cup W_2$ is also a GW.

\subsection{Non-degenerate commuting squares}
\label{sec:CommutingSquare}

In \cite{Sahu:2025upe}, it was pointed out that the strong subadditivity relation for GWs would follow given a GW-subalgebra map, the relation (\ref{Sdiff}) between generalized entropies and algebraic entropies, and the assumption that the algebras associated with $W_1 \vee W_2$, $W_1$, $W_2$, and $W_1 \cap W_2$ form a \emph{non-degenerate commuting square}.

In this section, we would like to understand more generally, given regions $A_1 \subset A_2,A_3 \subset A_4$ such that the conditions of Proposition \ref{diamond} are satisfied, when it is possible to choose the corresponding boundary algebras ${\cal M}_1 \subset {\cal M}_2,{\cal M}_3 \subset {\cal M}_4$
to form a non-degenerate commuting square.

The commuting square condition is a property of the conditional expectations
\[
E_{j \to i} : {\cal M}_j \to {\cal M}_i
\]
which in this finite-dimensional setting can be understood  geometrically as a projection to ${\cal M}_i$ considered as a subspace of ${\cal M}_j$. For a commuting square, we require
\[
E_{4 \to 3} \circ E_{4 \to 2} = E_{4 \to 2} \circ E_{4 \to 3} = E_{4 \to 1} \; .
\]
The commuting square is \emph{non-degenerate} if the commutant of the algebras within ${\cal M}_4$ also form a commuting square. 

Thinking in terms of the conditional expectations as projectors to subspaces, we can see that the commuting square condition is equivalent to the requirement that
\[
{\cal M}_2 \ominus {\cal M}_1 \perp {\cal M}_3 \ominus {\cal M}_1 \; .
\]
That is, we require that the orthogonal complement of ${\cal M}_1$ in ${\cal M}_2$ is perpendicular to the orthogonal complement of ${\cal M}_1$ in ${\cal M}_3$. This implies in particular that ${\cal M}_1$ is the intersection of ${\cal M}_2$ and ${\cal M}_3$.

When the conditions of Proposition \ref{diamond} are satisfied, we can obtain further necessary and sufficient conditions on the isometries such that the diamond of algebras form a non-degenerate commuting square, making use of the following proposition that we prove in \cite{FactorCode}:

\begin{proposition}[Diamond criterion for a non-degenerate commuting square]
\label{prop:diamond_commuting_square}
Consider Hilbert spaces
\[
\mathcal H_1,\mathcal H_2,\mathcal H_3,\mathcal H_4
\]
ordered as \(1<2,3<4\).  Let
\[
T_{ji}:\mathcal H_i\to \mathcal H_j
\]
be isometries satisfying
\[
T_{41}=T_{42}T_{21}=T_{43}T_{31}~.
\]
Let $P_i$ denote the orthogonal projector onto \(T_{4i}\mathcal H_i\subset \mathcal H_4\), for \(i=1,2,3\).  Set
\[
n=\dim\mathcal H_1~,\qquad
k=\frac{\dim\mathcal H_2}{\dim\mathcal H_1}~,\qquad
m=\frac{\dim\mathcal H_3}{\dim\mathcal H_1}~.
\]
Then the isometries can be completed by auxiliary isometries so that the associated factor algebras
\[
\mathcal M_1\subset \mathcal M_2,\mathcal M_3
\subset \mathcal M_4=B(\mathcal H_4)
\]
form a non-degenerate commuting square if and only if
\[
\dim\mathcal H_4
=
\frac{\dim\mathcal H_2\,\dim\mathcal H_3}{\dim\mathcal H_1}
=
nkm
\]
and
\[
P_2 P_3=P_3P_2=P_1~.
\]
Equivalently, the subspaces
\[
T_{42}\mathcal H_2
\qquad\text{and}\qquad
T_{43}\mathcal H_3
\]
have intersection $T_{41} \mathcal{H}_1$ and the orthogonal complements $T_{42}\mathcal H_2 \ominus T_{41} H_1$ and $T_{43}\mathcal H_3 \ominus T_{41} H_1$ are orthogonal to one another. 
\end{proposition}
In terms of the principal angles between $T_{42}\mathcal H_2$ and $T_{43}\mathcal H_3$, the condition is that we have $\dim \mathcal{H}_1$ principal angles equal to zero and the rest equal to $\pi/2$.

\begin{figure}
    \centering
    \includegraphics[width=0.4\linewidth]{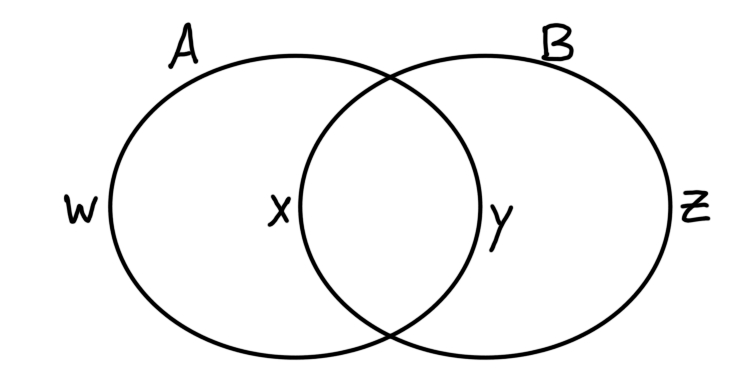}
    \caption{Overlapping regions $A$ and $B$ in a tensor network. Here $w$,$x$,$y$,and $z$ represent the number of edges in the various boundary components.}
    \label{fig:ABoverlap}
\end{figure}

Based on this result, we can immediately show that in an isometric tensor network, the non-degenerate commuting square condition will hold for extremal regions $A$ and $B$ intersecting as shown in Figure \ref{fig:ABoverlap} where $A \cup B$ is also extremal. We make use of the following:

\begin{corollary}
    Suppose $\mathcal{H}_1 = \mathcal{H}_L \otimes \mathcal{H}_R$, $\mathcal{H}_2 = \hat{\mathcal{H}}_L \otimes \mathcal{H}_R$, $\mathcal{H}_3 = \mathcal{H}_L \otimes \hat{\mathcal{H}}_R$, and $\mathcal{H}_4 = \hat{\mathcal{H}}_L \otimes \hat{\mathcal{H}}_R$ with $\dim \hat{\mathcal{H}}_L = k \dim \mathcal{H}_L$ and $\dim \hat{\mathcal{H}}_R = m \dim \mathcal{H}_R$ for integers $k,m$. Assume we have isometries $T_L : \mathcal{H}_L \to \hat{\mathcal{H}}_L$ and $T_R : \mathcal{H}_R \to \hat{\mathcal{H}}_R$. Then we can use these to define isometries $T_{ij}$ and operator maps $\varphi_{ij}$ satisfying the properties of Proposition \ref{diamond}, and the resulting algebras $\mathcal{M}_i$ form a commuting and co-commuting square. 
    \label{overlap}
\end{corollary}
\begin{proof}
We have
\[
\dim \mathcal{H}_4 = \dim \mathcal{H}_{\hat{L}} \dim \mathcal{H}_{\hat{R}} = {(\dim \mathcal{H}_{\hat{L}} \dim \mathcal{H}_R) (\dim \mathcal{H}_L \dim \mathcal{H}_{\hat{R}}) \over \dim \mathcal{H}_L \dim \mathcal{H}_R} = {\dim \mathcal{H}_2 \dim \mathcal{H}_3 \over \dim \mathcal{H}_1} \; .
\]
Also, 
\[
T_{21} = T_L \otimes \mathbb{1}_R~, \qquad T_{31} = \mathbb{1}_L \otimes T_R~, \qquad T_{42} = \mathbb{1}_{\hat{L}} \otimes T_R~, \qquad T_{43} = T_L \otimes \mathbb{1}_{\hat{R}}~,
\]
and
\[
P_2 =  \mathbb{1}_{\hat{L}} \otimes T_R T_R^\dagger~, \qquad P_3 =   T_L T_L^\dagger \otimes \mathbb{1}_{\hat{R}}~, \qquad P_1 =   T_L T_L^\dagger \otimes T_R T_R^\dagger \; .
\]
It follows immediately that $T_{42} T_{21} = T_{43} T_{31}$ and that $P_2 P_3 = P_3 P_2 = P_1$. Thus, the necessary and sufficient conditions of Proposition \ref{prop:diamond_commuting_square} are satisfied.
\end{proof}

For extremal regions $A$ and $B$ as shown in Figure \ref{fig:ABoverlap}, the condition that $A$ and $B$ are extremal give $w \ge x$ and $z \ge y$, so the conditions of the proposition are satisfied with $\mathcal{H}_1 = \mathcal{H}_{A \cap B} = \mathcal{H}_x \otimes \mathcal{H}_y$, $\mathcal{H}_2 = \mathcal{H}_A = \mathcal{H}_w \otimes \mathcal{H}_y$, $\mathcal{H}_3 = \mathcal{H}_B = \mathcal{H}_x \otimes \mathcal{H}_z$ and $\mathcal{H}_4 = \mathcal{H}_{A \cup B} = \mathcal{H}_w \otimes \mathcal{H}_z$, so we can define boundary algebras $\mathcal{M}_{A \cap B} \subset \mathcal{M}_A, \mathcal{M}_B \subset \mathcal{M}_{A \cup B}$ associated with the regions $A \cap B \subset A,B \subset A \cup B$, and these algebras will form a non-degenerate commuting square.

For an isometric tensor network with extremal regions $A$ and $B$ and $A\cup B$ \emph{not} extremal, the non-degenerate commuting square condition for the algebras $\mathcal{M}_{A \cap B} \subset \mathcal{M}_A, \mathcal{M}_B \subset \mathcal{M}_{A \vee B}$ (assuming that the inclusion property is satisfied) associated with the regions $A \cap B \subset A,B \subset A \vee B$ is violated. Indeed, Proposition \ref{prop:overlap-projectors-variance} shows that on average to leading order in $1/D$ we have
\[
\overline{\tr(P_AP_B)} =\frac{D_AD_B}{D_{A\vee B}} > D_{A\cap B}~,
\]
where we have used the general relation $|\partial A| + |\partial B| \ge |\partial A \cap B| + |\partial A \cup B| > |\partial A \cap B| + |\partial A \vee B|$ for $A\cup B$ not extremal. This contradicts the statement that we have $D_{A\cap B}$ principal angles equal to zero and the rest equal to $\pi/2$. Hence, from Proposition \ref{prop:diamond_commuting_square} we learn that the boundary algebras cannot form a non-degenerate commuting square. 

\subsection{Generalized entropy}
\label{subsec:generalized-entropy}

In \cite{Sahu:2025upe}, we proposed that the generalized entropy associated to a wedge $W$ can be expressed in terms of an entropic quantity associated to the corresponding algebra $\mathcal{A}_W$. Our generalized entropy prescription is given by
\begin{equation}
S_{\mathrm{gen}}(W)= S(\omega_W|\Acal_W)-\log \operatorname{Ind}E_{\Omega\rightarrow \Acal_W}+K_\Omega~,
\label{eq:generalized-entropy-tensor-network}
\end{equation}
where \(S(\omega_W|\Acal_W)\) denotes the algebraic entropy for the pair \((\Acal_W,\omega_W)\), and \(K_\Omega\) is a state-independent normalization constant associated with the ambient algebra \(\Omega\). This formula gives the result (\ref{Sdiff}) for entropy differences, and can also be obtained from (\ref{Sdiff}) by taking the smaller region to be empty. We now evaluate the formula \eqref{eq:generalized-entropy-tensor-network} in our finite-dimensional tensor network model.

Suppose we have an extremal subregion $A$ in the tensor network with associated boundary von Neumann factor algebra $\mathcal{M}_A\subset \mathcal{M}_{\partial}=B(\Hcal_\partial)$. Its construction depends on a choice of mutually orthogonal isometries $T^{(a)}:\mathcal{H}_A\to \mathcal{H}_{\partial}$ satisfying 
\[
T^{(a)\dagger}T^{(b)}
=
\delta_{ab}\mathbb 1_A~,
\qquad
\sum_{a=1}^{m}T^{(a)}T^{(a)\dagger}
=
\mathbb 1_\partial~,
\]
where 
\[
m = D_{\partial}/D_A~, \qquad 
D_\partial=\dim\Hcal_\partial~,
\qquad
D_A=\dim\Hcal_A~.
\]
A natural choice for the algebra $\Omega$ in the tensor network is the algebra associated to the boundary system $\Omega = B(\mathcal{H}_{\partial})$. We now give an explicit expression for conditional expectation from $\Omega$ to $\mathcal{M}_A$:

\begin{proposition}
\label{prop:tensor-network-expectation}
The unique trace-preserving conditional expectation $E_{\Omega\rightarrow \mathcal{M}_A}:
B(\Hcal_\partial)\longrightarrow\mathcal M_A$
is given by
\begin{equation}
E_{\Omega\rightarrow \mathcal{M}_A}(\mathcal{O}_{\partial})
=
\frac{1}{m}
\sum_{a,b=1}^{m}
T^{(b)}T^{(a)\dagger}
\mathcal{O}_{\partial}
T^{(a)}T^{(b)\dagger}~.
\label{eq:explicit-conditional-expectation}
\end{equation}

\end{proposition}

\begin{proof}

We verify that \eqref{eq:explicit-conditional-expectation} is the orthogonal projection to the subspace $\mathcal{M}_A$ with respect to the inner product $\langle B|A\rangle = \tr(B^\dagger A)$. It is clear from the definition of $E$ that its image is in $\mathcal{M}_A$. It remains to check that \( \mathcal{O}_{\partial} - E(\mathcal{O}_{\partial}) \) is orthogonal to $\mathcal{M}_A$. To see this, note that
\begin{align*}
    \tr\big(\varphi(\mathcal{O}_{A})E(\mathcal{O}_{\partial})\big) &= \frac{1}{m}\sum_{a,b,c} \tr\left( T^{(c)}\mathcal{O}_{A}T^{(c)\dagger}T^{(a)}T^{(b)\dagger}\mathcal{O}_{\partial} T^{(b)}T^{(a)\dagger}\right) \\
    &= \frac{1}{m}\sum_{a,b,c} \tr\left( T^{(c)}\mathcal{O}_{A}\delta^{ac}T^{(b)\dagger}\mathcal{O}_{\partial} T^{(b)}T^{(a)\dagger}\right) \\
    &= \frac{1}{m}\sum_{a,b} \tr\left( T^{(a)}\mathcal{O}_{A}T^{(b)\dagger}\mathcal{O}_{\partial} T^{(b)}T^{(a)\dagger}\right) \\
    &= \sum_{b} \tr\left( T^{(b)} \sum_a \frac{1}{m}T^{(a)\dagger}T^{(a)}\mathcal{O}_{A}T^{(b)\dagger}\mathcal{O}_{\partial} \right) = \tr\left(\sum_b T^{(b)}\mathcal{O}_{A} T^{(b)\dagger}\mathcal{O}_{\partial}\right)\\
    &= \tr(\varphi(\mathcal{O}_{A})\mathcal{O}_{\partial})~.
\end{align*}
Thus \( \mathcal{O}_{\partial} - E(\mathcal{O}_{\partial}) \) is orthogonal to the entire algebra generated by the operator map $\varphi$ and hence $E(\mathcal{O}_{\partial})$ is the projection of $\mathcal{O}_{\partial}$ onto this algebra. 

\end{proof}

One can equivalently express the conditional expectation \eqref{eq:explicit-conditional-expectation} in terms of a more standard partial trace operation. Under the unitary identification
\[
\Hcal_\partial
\simeq
\Hcal_A\otimes\mathbb C^m
\]
determined by the isometries \(T^{(a)}\), the algebra \(\mathcal{M}_A\) takes the form
\[
\mathcal{M}_A
\simeq
B(\Hcal_A)\otimes\mathbb 1_m~.
\]
In this representation, the conditional expectation takes the standard form
\[
E_{\Omega \to \mathcal{M}_A}(\mathcal{O}_{\partial})\simeq \tau_m(\mathcal{O}_{\partial})\otimes \mathbb{1}_m~,
\]
where $\tau_m(Y)=\frac{1}{m}\tr_m(Y)$ is the normalized partial trace over $\mathbb{C}^m$. The expression \eqref{eq:explicit-conditional-expectation} is written directly in terms of the isometry $T_A$ and its orthogonal counterparts, without assuming an explicit Hilbert space factorization. 

The index of the conditional expectation agrees with the Jones index, which for the finite-dimensional tensor network algebras is simply given by a ratio of dimensions:
\[
\operatorname{Ind}E_{\Omega\rightarrow \mathcal{M}_A}
=\left(\frac{D_{\partial}}{D_A} \right)^2 =m^2~.
\]

We next evaluate the generalized entropy formula. Recall that the algebraic entropy in \eqref{eq:generalized-entropy-tensor-network} is defined as \cite{Sahu:2025upe}:
\[
S(\omega | \mathcal{A}) :=  - \tau(\rho_\mathcal{A} \log \rho_\mathcal{A})~,
\] 
where $\tau$ is a normalized trace, i.e. $\tau(\mathbb{1})=\mathbb{1}$, and $\rho_\mathcal{A}=E(\rho_{\Omega})$ is the density matrix of the state $\omega$ coarse-grained to the subalgebra $\mathcal{A}$ through the conditional expectation $E$. A natural choice for the state $\omega$ in the entropy formula is the code state 
\[
|\Psi\rangle =
T_A|\psi_A\rangle \in \mathcal{H}_{\partial}~,
\]
where $|\psi_A\rangle$ is the state constructed from part of the tensor network corresponding to the subregion $A$. Note that the state $|\Psi\rangle$ is determined by the full tensor network geometry and therefore independent of the choice of subregion. We can extend this to an orthonormal set of states
\[
|\Psi^{(a)}\rangle
=
T^{(a)}|\psi_A\rangle~, \qquad a=1,..., m~,
\]
with $|\Psi^{(1)}\rangle= |\Psi\rangle$. We use that the normalized trace on $\Omega$ is given by:
\[
\tau(\mathcal{O}_{\partial})
=
\frac{1}{D_\partial}
\tr(\mathcal{O}_{\partial})~.
\]
Note that the ordinary density matrix $\rho_\Psi
=
|\Psi\rangle\langle\Psi|$
is not normalized with respect to \(\tau_\partial\), so we define $\widehat\rho_\Psi
=
D_\partial\rho_\Psi$, which the property $\tau(\widehat\rho_\Psi)=1$. We show that:

\begin{proposition}
\label{prop:tensor-network-generalized-entropy}
For the code state \(|\Psi\rangle=T_A|\psi_A\rangle\), the generalized entropy formula
\begin{equation} \label{eq:SgenA}
S_{\mathrm{gen}}(A)= S(\Psi_A|\mathcal{M}_A)-\log \operatorname{Ind}E_{\Omega\rightarrow \mathcal{M}_A}+K_\Omega~,
\end{equation}
with the choice $K_\Omega=2\log D_\partial$, reduces to
\[
S_{\mathrm{gen}}(A)=\log D_A~.
\]
\end{proposition}

\begin{proof}
The density operator representing the restricted state is obtained by applying
the conditional expectation:
\begin{align}
\rho_{\mathcal M_A}
&=
E_{\Omega\rightarrow \mathcal{M}_A}(\widehat\rho_\Psi)
\nonumber\\
&=
\frac{D_\partial}{m}
\sum_{a,b=1}^{m}
T^{(b)}T^{(a)\dagger}
T^{(1)}
|\psi_A\rangle\langle\psi_A|
T^{(1)\dagger}
T^{(a)}T^{(b)\dagger}~. \nonumber
\end{align}
Using that $T^{(a)\dagger}T^{(1)}
=
\delta_{a1}\mathbb 1_A$ we obtain
\[
\rho_{\mathcal M_A}
=
\frac{D_\partial}{m}
\sum_{b=1}^{m}
T^{(b)}
|\psi_A\rangle\langle\psi_A|
T^{(b)\dagger}=
\frac{D_\partial}{m}
\sum_{b=1}^{m}
|\Psi_b\rangle\langle\Psi_b|~.
\]
Since the vectors \(|\Psi_b\rangle\) are orthonormal,
\[
\Pi
=
\sum_{b=1}^{m}
|\Psi_b\rangle\langle\Psi_b|
\]
is a
rank\(-m\) projection, and
\[
\rho_{\mathcal M_A}\log\rho_{\mathcal M_A}
=
\frac{D_\partial}{m}
\log\left(\frac{D_\partial}{m}\right)\Pi~.
\]
The algebraic entropy therefore evaluates to
\begin{align}
S(\Psi_A|\mathcal{M}_A)
&=
-\tau(\rho_{\mathcal M_A}\log\rho_{\mathcal M_A})\nonumber \\
&=
-
\frac{D_\partial}{m}
\log\left(\frac{D_\partial}{m}\right)
\tau(\Pi) = -\log\left(\frac{D_\partial}{m}\right)~, \nonumber
\end{align}
where we used that the normalized trace of a rank$-m$ projection is $\tau(\Pi)
=
\frac{m}{D_\partial}$.

Substituting this result and $\log\operatorname{Ind}E_{\Omega\rightarrow \mathcal{M}_A}
=
2\log m$ into \eqref{eq:SgenA} gives
\begin{align*}
S_{\mathrm{gen}}(A)
&=
\log\frac{m}{D_\partial}
-
2\log m
+
K_\Omega \\
&=
K_\Omega-2\log D_\partial+\log D_A~. \nonumber
\end{align*}
The choice \(K_\Omega=2\log D_\partial\) therefore yields
\[
S_{\mathrm{gen}}(A)
=
\log D_A~.
\]
\end{proof}

Note that the negative value of algebraic entropy is a consequence of using
density operators normalized with respect to the normalized trace
\(\tau_\partial\), rather than the ordinary matrix trace. The index term and
the state-independent constant \(K_\Omega\) combine with this \(S(\Psi_A\vert \mathcal{M}_A)\) to reproduce the expected positive contribution
\(\log D_A\), the tensor network analogue of the area of the subregion $A$.

Instead of considering the pure code state $|\Psi\rangle=T_A|\psi_A\rangle$, we may start from a general
mixed state $\rho_A\in B(\Hcal_A)$
and define the corresponding boundary state by \[
\rho_\Psi
=
T_A\rho_AT_A^\dagger~.
\]
Such a situation arises naturally when we extend our tensor network to include additional
bulk legs. These bulk legs can be used to feed in a mixed state. We prove the following:

\begin{proposition}
\label{prop:tensor-network-generalized-entropy-mixed}
Let $\rho_\Psi=T_A\rho_AT_A^\dagger$ be the boundary state associated with a density operator \(\rho_A\) on
\(\Hcal_A\). Then the generalized entropy formula
\begin{equation}
S_{\mathrm{gen}}(A)
=
S(\Psi_A|\mathcal{M}_A)
-
\log\operatorname{Ind}
E_{\Omega\rightarrow\mathcal{M}_A}
+
K_\Omega
\label{eq:SgenA-mixed}
\end{equation}
with the choice $K_\Omega=2\log D_\partial$ reduces to
\begin{equation}
S_{\mathrm{gen}}(A)
=
\log D_A
+
S(\rho_A)~,
\label{eq:SgenA-mixed-final}
\end{equation}
where $S(\rho_A)
=
-\Tr(\rho_A\log\rho_A)$ is the von Neumann entropy of the state $\rho_A$.
\end{proposition}

\begin{proof}
The argument proceeds as in
Proposition~\ref{prop:tensor-network-generalized-entropy}, with the rank-one
operator
\(
|\psi_A\rangle\langle\psi_A|
\)
replaced by the more general mixed density operator \(\rho_A\).

Since \(\rho_\Psi\) has unit trace with respect to the ordinary matrix trace,
the density operator normalized with respect to $\tau$
is given by $
\widehat\rho_\Psi
=
D_\partial T_A\rho_AT_A^\dagger$.
Applying the conditional expectation gives
\[
\rho_{\mathcal M_A}
=
E_{\Omega\rightarrow\mathcal M_A}
(\widehat\rho_\Psi)
\nonumber=
\frac{D_\partial}{m}
\sum_{b=1}^{m}
T^{(b)}\rho_AT^{(b)\dagger}~.
\]

Now let
\[
\rho_A
=
\sum_i\lambda_i|i\rangle\langle i|~, \qquad \lambda_i\geq0 ~, \qquad
\sum_i\lambda_i=1~,
\]
be the spectral decomposition of \(\rho_A\). Defining
\[
|\Psi_{b,i}\rangle
=
T^{(b)}|i\rangle~,
\]
we can write the coarse-grained density operator as
\[
\rho_{\mathcal M_A}
=
\frac{D_\partial}{m}
\sum_{b,i}
\lambda_i
|\Psi_{b,i}\rangle
\langle\Psi_{b,i}|~.
\]
The orthogonality relations for the isometries imply
\[
\langle\Psi_{a,i}|\Psi_{b,j}\rangle
=
\delta_{ab}\delta_{ij}~.
\]
Thus, each eigenvalue \(\lambda_i\) of \(\rho_A\) gives rise to the eigenvalue $\frac{D_\partial}{m}\lambda_i$ of \(\rho_{\mathcal M_A}\), with multiplicity \(m\). It follows that
\begin{align*}
S(\Psi_A|\mathcal M_A)
&=
-\tau\left(
\rho_{\mathcal M_A}
\log\rho_{\mathcal M_A}
\right)
\nonumber\\
&=
-
\frac{1}{D_\partial}
\sum_{b,i}
\frac{D_\partial}{m}\lambda_i
\log\left(
\frac{D_\partial}{m}\lambda_i
\right)
\nonumber\\
&=
-
\sum_i\lambda_i
\log\left(
\frac{D_\partial}{m}\lambda_i
\right)
\nonumber\\
&=
-\log\left(\frac{D_\partial}{m}\right)
-
\sum_i\lambda_i\log\lambda_i
\nonumber\\
&=
\log\frac{m}{D_\partial}
+
S(\rho_A)~.
\end{align*}

The remaining steps are the same as in the proof of 
Proposition~\ref{prop:tensor-network-generalized-entropy}. Using $\log\operatorname{Ind}
E_{\Omega\rightarrow\mathcal M_A}
=
2\log m$ and $K_\Omega=2\log D_\partial$ we obtain
\begin{align*}
S_{\mathrm{gen}}(A)
&=
\log\frac{m}{D_\partial}
+
S(\rho_A)
-
2\log m
+
K_\Omega\\
&=
\log D_A
+
S(\rho_A)~.
\end{align*}

\end{proof}

\section{Type $II_1$ algebras from tensor networks in the large bond dimension limit}
\label{sec:largeD}

In this section, we would like to define an algebra based on an infinite family of random tensor networks with the same geometry but with increasing bond dimension. Properties that are approximately true for the finite dimensional algebras with errors that vanish in an appropriate way in the limit where the bond dimension goes to infinity will be exactly true in the final algebra. 

The use of limiting operator algebras provides a natural framework for
describing infinite tensor networks and large-$N$ holographic systems
\cite{Faulkner:2022ada,Chandrasekaran:2022eqq,Chemissany:2025inf}.
In the present work, we implement this limit using a standard tool in the theory of operator algebras known as the ``ultraproduct'' \cite{Ando:2014ultraproducts}. Related applications
of ultraproduct constructions to large-$N$ limits and
holography have recently been discussed in
\cite{Chen:2025third,Bahiru:2026ultra}.

\subsection{Ultraproduct algebras}

Consider any infinite sequence of Hilbert spaces ${\cal H}_n$ of dimension
$D_n$. Let $M_n$ be the algebra of operators on ${\cal H}_n$.

We would like to define a new algebra by taking the Cartesian product of all the $M_n$. The elements of this correspond to infinite sequences ${\cal O} = ({\cal O}_n)$ with one operator from each $M_n$. Addition, multiplication by a constant, multiplication, and adjoint happen component-wise. We define a norm via
\[
||{\cal O}|| = \sup_n ||{\cal O}_n|| \; ,
\]
where $||{\cal O}_n||$ is the standard operator norm, the infimum of the set of all numbers $\lambda$ for which $|{\cal O}_n v| \le \lambda |v|$ for all $v \in {\cal H}_n$.
We restrict to the bounded operators according to this norm to define the new algebra $\mathcal{M}_\infty$:
\[
\mathcal{M}_\infty = \{(O_n)|\sup_n ||{\cal O}_n|| < \infty\} \; .
\]

\subsubsection*{Ultrafilter limits}

Next, we  define a quotient algebra by identifying elements whose sequences of operators coincide asymptotically. The point of this is that many properties that become true asymptotically for the initial sequence of algebras will be exactly true in this quotient algebra.

The standard tool to express this asymptotic convergence is the idea of an ultrafilter limit.
We begin with some definitions. 

\paragraph{Definition (Filter).}
A \emph{filter} $\mathcal{F}$ on $\mathbb{N}$ is a collection of subsets of $\mathbb{N}$ such that:
\begin{enumerate}
  \item $\varnothing \notin \mathcal{F}$ and $\mathbb{N} \in \mathcal{F}$.
  \item If $A,B \in \mathcal{F}$, then $A \cap B \in \mathcal{F}$.
  \item If $A \in \mathcal{F}$ and $A \subseteq B \subseteq \mathbb{N}$, then $B \in \mathcal{F}$.
\end{enumerate}

\paragraph{Definition (Ultrafilter).}
An \emph{ultrafilter} $\omega$ on $\mathbb{N}$ is a filter with the additional property that
for every subset $A \subseteq \mathbb{N}$, exactly one of the following holds:
\[
  A \in \omega \quad \text{or} \quad \mathbb{N}\setminus A \in \omega~.
\]
Equivalently, an ultrafilter is a filter that is maximal in the sense that no other filter includes it as a proper subset.

There are two types of ultrafilters.
\begin{itemize}
  \item A \emph{principal ultrafilter} at $n_0 \in \mathbb{N}$ is the collection
  \[
    \omega_{n_0} = \{ A \subseteq \mathbb{N} : n_0 \in A \}~.
  \]
  \item A \emph{nonprincipal ultrafilter} is an ultrafilter containing no finite set. 
        The existence of such ultrafilters requires the axiom of choice.
\end{itemize}

\paragraph{Definition (Ultrafilter Limit).}
Let $\omega$ be a nonprincipal ultrafilter on $\mathbb{N}$, and let
$(x_n)_{n \in \mathbb{N}}$ be a bounded sequence of complex numbers.
We say that $L \in \mathbb{C}$ is the \emph{ultrafilter limit} of $(x_n)$ along $\omega$,
and write
\[
  L = \lim_{n \to \omega} x_n~,
\]
if for every neighbourhood $U$ of $L$ in $\mathbb{C}$,
\[
  \{\, n \in \mathbb{N} : x_n \in U \,\} \in \omega~.
\]

We have two important properties of the ultrafilter limit:
\begin{itemize}
  \item If $(x_n)$ converges in the usual sense to $L$, then 
        $\lim_{n \to \omega} x_n = L$.
  \item Every bounded sequence admits an ultrafilter limit.
\end{itemize}
The latter property is crucial and is the key reason that we will consider the ultrafilter limit.

\subsubsection*{The ideal}

Now each matrix algebra has a normalized trace
\[
\tau_n(X_n)=\frac{1}{D_n}\Tr(X_n)
\]
and associated \(2\)-norm
\[
\|X_n\|_{2,n}
=
\bigl(\tau_n(X_n^*X_n)\bigr)^{1/2}.
\]
Define
\[
\mathcal I_\omega^{(2)}
=
\left\{
(X_n)\in\mathcal M_\infty:
\lim_{n\to\omega}\|X_n\|_{2,n}=0
\right\}.
\]
If \(X\in\mathcal I_\omega^{(2)}\) and
\(Y\in\mathcal M_\infty\), then
\[
\lim_{n\to\omega}\|X_nY_n\|_{2,n}
\leq
\|Y\|
\lim_{n\to\omega}\|X_n\|_{2,n}
=0~,
\]
and similarly \(YX\in\mathcal I_\omega^{(2)}\).
The triangle inequality, invariance under adjoint, and norm closure show
that \(\mathcal I_\omega^{(2)}\) is a norm-closed two-sided
\(*\)-ideal.

\subsubsection*{The quotient}

We can form the quotient ${\cal M}_\omega$, also denoted $\prod_\omega (M_n, \tau_n)$ as 
\[
\mathcal{M}_\omega \equiv \mathcal{M}_\infty / {\cal I}^{(2)}_\omega
\]
whose elements are equivalence classes
\[
[X] = \{Y | \lim_{n \to \omega} ||Y_n - X_n||_{2,n} = 0 \} \; .
\]
Addition, multiplication by a constant, multiplication, and adjoint are defined by taking these operations on any representative, i.e.
\begin{eqnarray*}
[X_1] + [X_2] &=& [X_1 + X_2] \cr
[X_1][X_2] &=& [X_1 X_2] \cr
c [X_1] &=& [c X_1] \cr
[X]^* &=& [X^*] 
\end{eqnarray*}
We can check that the right hand sides here do not depend on which representative we choose. The operators in $\mathcal{M}_\omega$ have a 2-norm
\[
||[X]||_{2,n} = \lim_{n \to \omega} ||X_n||_{2,n} \; .
\]
The algebra $\mathcal{M}_\omega$ is a finite von Neumann algebra, and for $D_n \rightarrow_\omega \infty$ is specifically a type $II_1$ factor  \cite{anantharaman2017introduction} with a faithful normal tracial state
\[
\tau_\omega([X]) = \lim_{n \to \omega} \tau_n (X_n) \; .
\]

\subsection{Tensor network application}

We now apply the tracial ultraproduct construction to a sequence of random
tensor networks with fixed geometry and increasing bond dimension. 

Let \(G\) be a fixed finite tensor-network geometry with $b = |\partial G|$. For each \(n\), assign a
Hilbert space of dimension \(d_n\) to every edge of \(G\), choose an independent
random tensor at each vertex,
and denote the resulting tensor network by \(T_n\). We assume that $d_n \to \infty$.
The boundary Hilbert
space is
\[
\cH_n^\partial \cong (\mathbb C^{d_n})^{\otimes b}~,
\]
and we write
\[
\cM_n^\partial = B(\cH_n^\partial)
\]
for the algebra of boundary operators. We equip \(\cM_n^\partial\) with its
normalized trace
\[
\tau_n^\partial(X)
=
\frac{1}{\dim \cH_n^\partial}\Tr(X)
\]
and the associated \(2\)-norm
\[
\|X\|_{2,n}
=
\bigl(\tau_n^\partial(X^*X)\bigr)^{1/2}.
\]
The boundary ultraproduct algebra is
\[
\cM_\omega^\partial
=
\prod_\omega (\cM_n^\partial,\tau_n^\partial)~.
\]

\subsubsection{Extremal regions}

For an extremal region \(A\), let
\[
a_A=|\partial A|~,
\qquad
\mathcal H_n^A
\cong
(\mathbb C^{d_n})^{\otimes a_A}~,
\qquad
D_{A,n}
\equiv
\dim\mathcal H_n^A
=d_n^{a_A}~,
\]
and define
\[
\mathcal N_n^A
=
B(\mathcal H_n^A)~.
\]
We equip \(\mathcal N_n^A\) with its normalized trace \(\tau_n^A\) and
associated \(2\)-norm \(\|\cdot\|_{2,n}\). Since the full graph \(G\)
contains \(A\), extremality implies \(a_A\leq b\).

Let
\[
T_n^A:\mathcal H_n^A\longrightarrow\mathcal H_n^\partial
\]
be the normalized tensor-network map defined as in
\eqref{eq:normalized_TA}. In particular,
\[
\Tr\bigl((T_n^A)^\dagger T_n^A\bigr)
=
D_{A,n}~.
\]
For any map \(S:\mathcal H_n^A\to\mathcal H_n^\partial\), we use the
normalized Hilbert--Schmidt norm
\[
\|S\|_{2,n}
=
\bigl(\tau_n^A(S^\dagger S)\bigr)^{1/2}.
\]

We now state explicitly the large-bond-dimension estimates that will be used
below. These are assumptions on the random tensor ensemble; the leading-order
averaged purity calculation motivates them but does not by itself prove the
uniform high-probability estimates.

\paragraph{Large-bond-dimension assumption.}
For every extremal region \(A\), assume that there is constant
\(C_A<\infty\) and functions
\[
f_A(D)\longrightarrow0~,
\qquad
g_A(D)\longrightarrow0~,
\qquad
\textrm{for} \qquad D\longrightarrow\infty~,
\]
such that, for every \(n\), with probability at least
\(1-g_A(d_n)\),
\begin{align}
\bigl\|(T_n^A)^\dagger T_n^A-\mathbb{1}_{A,n}\bigr\|_{2,n}
&\leq f_A(d_n)~.
\label{eq:largeD-isometry-defect}
\end{align}

Since \(G\) is finite, it has only finitely many extremal regions. Let
\[
g_{\mathrm{all}}(D)
=
\sum_{A\ {\rm extremal}}g_A(D)~.
\]
By the union bound, with probability at least
\(1-g_{\mathrm{all}}(d_n)\), the estimate
\eqref{eq:largeD-isometry-defect} holds for every extremal region
simultaneously.

Choose the sequence \(d_n\) to grow sufficiently rapidly that
\[
\sum_{n=1}^\infty g_{\mathrm{all}}(d_n)<\infty~.
\]
For each \(n\), let \(E_n\) be the event that at least one of the required
estimates fails for at least one extremal region. Then
\[
\mathbb P(E_n)
\leq
g_{\mathrm{all}}(d_n)~,
\qquad
\sum_{n=1}^\infty\mathbb P(E_n)<\infty~.
\]
The first Borel--Cantelli lemma\footnote{The first Borel--Cantelli lemma states that, for a sequence of events
\((E_n)_{n\in\mathbb N}\), if
\(\sum_{n=1}^{\infty}\mathbb P(E_n)<\infty\), then
\(\mathbb P(E_n\ \text{infinitely often})=0\). No independence assumption is
required.} therefore implies
\[
\mathbb P(E_n\ {\rm infinitely\ often})=0~.
\]
Consequently, with probability one, there is an \(n_0\) such that the
required estimates hold for every extremal region and every \(n\geq n_0\).
We henceforth fix a tensor-network realization with this property.

\subsubsection{The regional ultraproduct algebra}

For each \(n\), we can construct an exact isometry $V_n^A$ based on \(T_n^A\) as follows. We note that 
\((T_n^A)^\dagger T_n^A\) gives an invertible map from $(\ker T_n^A)^\perp$ to itself, and that on this domain
\[
\hat{V}_n^A = T_n^A ((T_n^A)^\dagger T_n^A)^{-1/2}
\]
is an isometry, since the definition immediately gives $(\hat{V}_n^A)^\dagger \hat{V}_n^A = \mathbb{1}$. We can extend this to an isometry 
\[
V_n^A:\mathcal H_n^A\longrightarrow\mathcal H_n^\partial~
\]
by mapping $\ker T_n^A$ into the orthogonal complement of the range of $T_n^A$. For any choice of extension, we have
\[
T_n^A
=
V_n^A\bigl((T_n^A)^\dagger T_n^A\bigr)^{1/2}.
\]
The inequality
\(\lvert1-\sqrt{x}\rvert\leq\lvert1-x\rvert\) for \(x\geq0\) then gives
\begin{align*}
\|V_n^A-T_n^A\|_{2,n}
&=
\left\|
\mathbb{1}_{A,n}
-
\bigl((T_n^A)^\dagger T_n^A\bigr)^{1/2}
\right\|_{2,n}
\\
&\leq
\bigl\|
\mathbb{1}_{A,n}
-
(T_n^A)^\dagger T_n^A
\bigr\|_{2,n}
\\
&\leq
f_A(d_n)
\end{align*}
for all sufficiently large \(n\). Thus \(V_n^A\) is an exact isometry that
agrees with the tensor-network map in normalized \(2\)-norm in the
large-bond-dimension limit.

Defining the integer
\[
r_{A,n}
=
\frac{D_{\partial,n}}{D_{A,n}}
=
d_n^{\,b-a_A}~,
\]
we can choose a unitary
\[
U_{A,n}:
\mathcal H_n^A\otimes\mathbb C^{r_{A,n}}
\longrightarrow
\mathcal H_n^\partial
\]
whose restriction to the first copy of \(\mathcal H_n^A\) is \(V_n^A\):
\[
U_{A,n}(\xi\otimes|1\rangle)
=
V_n^A\xi~.
\]
As in \eqref{opmap}, define
\[
\Phi_{A,n}:
\mathcal N_n^A
\longrightarrow
\mathcal M_n^\partial~,
\qquad
\Phi_{A,n}(O)
=
U_{A,n}
\bigl(O\otimes\mathbb{1}_{r_{A,n}}\bigr)
U_{A,n}^\dagger~.
\]
This map is a unital \(*\)-homomorphism and satisfies
\[
\Phi_{A,n}(O)V_n^A
=
V_n^A O~.
\]
It is also trace preserving:
\[
\tau_n^\partial(\Phi_{A,n}(O))
=
\tau_n^A(O)~ .
\]
The explicit form of the map gives
\begin{equation}
\|\Phi_{A,n}(O)\|=\|O\|~,
\label{ID1}
\end{equation}
while trace preservation gives
\begin{equation}
\|\Phi_{A,n}(O)\|_{2,n}
=
\|O\|_{2,n}~.
\label{ID2}
\end{equation}
Together with the estimate above, the intertwining relation also gives
\[
\|\Phi_{A,n}(O)T_n^A-T_n^A O\|_{2,n}
\leq
2\|O\|
\|V_n^A-T_n^A\|_{2,n}
\longrightarrow0
\]
uniformly for \(O\) in the operator-norm unit ball.

Define the regional tracial ultraproduct
\[
\mathcal N_\omega^A
=
\prod_\omega(\mathcal N_n^A,\tau_n^A)~.
\]
The maps \(\Phi_{A,n}\) induce
\[
\Phi_A:
\mathcal N_\omega^A
\longrightarrow
\mathcal M_\omega^\partial
\]
by
\[
\Phi_A([(O_n)])
=
[(\Phi_{A,n}(O_n))]~.
\]
The operator-norm identity (\ref{ID1}) shows that uniformly bounded sequences are mapped
to uniformly bounded sequences, while the \(2\)-norm identity in (\ref{ID2}) shows that the
map is well defined on equivalence classes.

To show that
\[
\mathcal M_\omega(A)
=
\Phi_A(\mathcal N_\omega^A)
\]
is a von Neumann subalgebra of
\(\mathcal M_\omega^\partial\), it is enough to show that
\(\Phi_A\) is a normal unital \( * \)-homomorphism. Unitality and the
\( * \)-homomorphism property ensure that the image is a unital
\( * \)-subalgebra, while normality implies that its image is
closed in the ultraweak topology. We will prove the stronger statement that \(\Phi_A\) is in addition
trace preserving and injective. Injectivity identifies its image with
\(\mathcal N_\omega^A\), allowing us to conclude that the image is a type
\(II_1\) factor whenever \(\mathcal N_\omega^A\) is.

The fact that \(\Phi_A\) is a
unital trace-preserving \(*\)-homomorphism follows directly from the fact that each \(\Phi_{A,n}\) has these properties. Moreover,
\[
\|\Phi_A(X)\|_{2,\omega,\partial}
=
\|X\|_{2,\omega,A}~,
\]
so $\Phi_A(X) = 0 \implies X=0$ and thus \(\Phi_A\) is injective. Since \(\Phi_A\) preserves the faithful normal tracial states, it is
normal. The range of a normal \(*\)-homomorphism between von Neumann
algebras is ultraweakly closed.

We have thus shown that $\mathcal M_\omega(A)
=
\Phi_A(\mathcal N_\omega^A)$ is a von Neumann subalgebra of
\(\mathcal M_\omega^\partial\). If \(a_A>0\) then \(D_{A,n}=d_n^{a_A}\to\infty\), so
\(\mathcal N_\omega^A\), and therefore \(\mathcal M_\omega(A)\), is a type
\(II_1\) factor. 

Different choices of the unitary extensions \(U_{A,n}\) give unitarily conjugate subalgebras. We conclude that, almost surely, every extremal region
with \(a_A>0\) determines a family of unitarily related type
\(II_1\) subfactors of the boundary ultraproduct algebra. 

As discussed above, properties that hold asymptotically at finite bond dimension become exact statements in the ultraproduct algebra $\mathcal M_\omega(A)$. The ultraproduct construction therefore provides a natural framework for describing the infinite bond dimension limit of the tensor network. 

\section{An alternative entropic quantity satisfying monotonicity and strong subadditivity}

In this section, we introduce an alternative formula for the entropy associated with a von Neumann subalgebra. As we will show, this quantity likewise satisfies the monotonicity and strong subadditivity properties, under suitable assumptions on the associated conditional expectations. We also discuss how this definition is related to, and differs from, our previous proposal.

\subsection{Definition and properties}

Let $\mathcal M$ be a von Neumann algebra and let $\omega$ be a faithful
normal state on $\mathcal M$. Suppose that $\mathcal A'\subset\mathcal M$
is a von Neumann subalgebra for which there exists a faithful normal
conditional expectation $E_{\mathcal A'}:\mathcal M\rightarrow\mathcal A'$. We define the {\it commutant entropy}
associated with $\mathcal A$ by
\begin{equation}
S_C(\mathcal A)
:=
S(\omega\|\omega\circ E_{\mathcal A'})~.
\label{SubalgebraEntropy}
\end{equation}
In this section, we will drop the $C$ subscript and refer to this simply as $S(\mathcal A)$.

The state $\omega\circ E_{\mathcal A'}$ represent a coarse-graining of the state $\omega$ to the subset of observables that commute with $\mathcal A$; these are observables that can be known at the same time as any observable in $\mathcal A$, i.e. have no uncertainty relation with any such observable. The relative entropy is a measure of how much information is lost in this coarse-graining process.

The above definition extends naturally beyond unital von Neumann
subalgebras. Indeed, if $\mathcal A\subset B(\mathcal H)$ is any (possibly
non-unital) $*$-subalgebra, then its commutant $\mathcal A'$ is automatically a von Neumann algebra, since $\mathcal A'''=\mathcal A'$.\footnote{Since every element of $\cal A$ commutes with every element of $\cal A'$, we have $\cal A\subset \cal A''$.
Taking commutants reverses inclusions, so $\cal A'''=(\cal A'')' \subset \cal A'$.
On the other hand, every element of $\cal A'$ commutes with every element of $\cal A''=(\cal A')'$, so $\cal A'\subset (\cal A'')'=\cal A'''$. Hence, $\cal A'''=\cal A'$.} In particular, $\mathbb 1\in\mathcal A'$, so whenever there exists a faithful normal conditional expectation $E_{\mathcal A'}:\mathcal M\rightarrow\mathcal A'$,
the entropy (\ref{SubalgebraEntropy}) remains well defined. This observation
allows one to associate entropies not only to von Neumann algebras, but also to
non-unital corner algebras, which arise naturally in the tensor network context.

We first show that the entropy as defined in \eqref{SubalgebraEntropy} is monotone under inclusion
of subalgebras.

\begin{proposition}[Monotonicity]
Let $\mathcal A\subset\mathcal B\subset\mathcal M$ be von Neumann subalgebras, and assume there exist faithful normal conditional expectations
\[
E_{\mathcal A'}:\mathcal M\rightarrow\mathcal A'~,
\qquad
E_{\mathcal B'}:\mathcal M\rightarrow\mathcal B'~,
\]
satisfying the compatibility condition
\[
E_{\mathcal B'}E_{\mathcal A'}
=
E_{\mathcal B'}~.
\]
Then
\begin{equation} \label{eq:Pythagorean}
S(\mathcal B)
=
S(\mathcal A)
+
S(\omega_{\mathcal A}\|\omega_{\mathcal B})~,
\end{equation}
where $\omega_{\mathcal A}
=
\omega\circ E_{\mathcal A'}, \omega_{\mathcal B}
=
\omega\circ E_{\mathcal B'}$. In particular, 
\[
S(\mathcal A)\le S(\mathcal B)~.
\]

\end{proposition}

\begin{proof}
Set
\[
E:=E_{\mathcal A'}~,\qquad
\omega_{\mathcal A}:=\omega\circ E~,
\qquad
\omega_{\mathcal B}:=\omega\circ E_{\mathcal B'}~.
\]
The compatibility condition implies
\[
\omega_{\mathcal B}\circ E
=
\omega\circ E_{\mathcal B'}E_{\mathcal A'}
=
\omega\circ E_{\mathcal B'}
=
\omega_{\mathcal B}~.
\]
Thus $\omega_{\mathcal B}$ is invariant under $E$.

By the Pythagorean identity for Araki relative entropy with respect to a
faithful normal conditional expectation \cite{petz1991certain}, for every faithful normal state
$\sigma$ satisfying $\sigma\circ E=\sigma$ one has
\[
S(\omega\|\sigma)
=
S(\omega\|\omega\circ E)
+
S(\omega\circ E\|\sigma)~.
\]
Applying this with $\sigma=\omega_{\mathcal B}$ gives
\[
S(\omega\|\omega_{\mathcal B})
=
S(\omega\|\omega_{\mathcal A})
+
S(\omega_{\mathcal A}\|\omega_{\mathcal B})~.
\]
Equivalently,
\[
S(\mathcal B)
=
S(\mathcal A)
+
S(\omega_{\mathcal A}\|\omega_{\mathcal B})~.
\]
Since relative entropy is nonnegative,
\[
S(\mathcal A)\le S(\mathcal B)~.
\]
\end{proof}

Our second result is a strong subadditivity inequality. It holds whenever
the conditional expectations associated with the commutant algebras form a
commuting square:

\begin{proposition}[Strong subadditivity]
Let $\mathcal A,\mathcal B\subset\mathcal M$
be von Neumann subalgebras and define
\[
\mathcal C:=\mathcal A\cap\mathcal B~,
\qquad
\mathcal D:=\mathcal A\vee\mathcal B~.
\]
Assume there exist faithful normal conditional expectations $E_{\mathcal A'},
E_{\mathcal B'},
E_{\mathcal C'},
E_{\mathcal D'}$.
and that
\[
E_{\mathcal A'}E_{\mathcal B'}
=
E_{\mathcal B'}E_{\mathcal A'}
=
E_{\mathcal D'}~.
\]
Assume further that 
\[
E_{\mathcal A'} E_{\mathcal C'} = E_{\mathcal A'} \qquad E_{\mathcal B'} E_{\mathcal C'} = E_{\mathcal B'} \; .
\]
This is automatic if $\mathcal{M}$ has a faithful normal tracial state $\tau$ and all conditional expectations are $\tau$-preserving.
Then
\[
S(\mathcal A)+S(\mathcal B)
\ge
S(\mathcal C)+S(\mathcal D)~.
\]
\end{proposition}

\begin{proof}
Introduce the coarse-grained states
\[
\omega_{\mathcal A}
=
\omega\circ E_{\mathcal A'}~,
\qquad
\omega_{\mathcal B}
=
\omega\circ E_{\mathcal B'}~,
\qquad
\omega_{\mathcal C}
=
\omega\circ E_{\mathcal C'}~,
\qquad
\omega_{\mathcal D}
=
\omega\circ E_{\mathcal D'}~.
\]
We note that 
\[
(\mathcal A' \vee \mathcal B')' = \mathcal A'' \cap \mathcal B'' = \mathcal A \cap \mathcal B = \mathcal C \; ,
\]
so that $\mathcal{C}' = \mathcal{A}' \vee \mathcal{B}'$. Thus, $\mathcal{A}', \mathcal{B}' \subset \mathcal{C}'$. By monotonicity of relative entropy under the conditional expectation
\(E_{\mathcal A'}\) we have:
\begin{equation}
\label{mon}
S(\omega_{\mathcal C}\|\omega_{\mathcal B})
\ge
S(\omega_{\mathcal C}\circ E_{\mathcal A'}
\,\|\,
\omega_{\mathcal B}\circ E_{\mathcal A'})~.
\end{equation}
Since $\mathcal{A}' \subset \mathcal{C}'$, we have 
\[
\omega_{\mathcal C} \circ E_{\mathcal A'} = \omega \circ E_{\mathcal{C}'} \circ E_{\mathcal A'} = \omega \circ  E_{\mathcal A'} = \omega_{\mathcal A}~.
\]
Using the commuting square relation $E_{\mathcal B'} E_{\mathcal A'}
=
E_{\mathcal D'}$ and 
the right-hand side of (\ref{mon}) becomes $S(\omega_{\mathcal A}\|\omega_{\mathcal D})$. Hence,
\[
S(\omega_{\mathcal C}\|\omega_{\mathcal B})
\ge
S(\omega_{\mathcal A}\|\omega_{\mathcal D})~.
\]

Applying the Pythagorean identity \eqref{eq:Pythagorean} to the inclusions
\[
\mathcal B'\subset \mathcal C'~,
\qquad
\mathcal D'\subset\mathcal A'~,
\]
gives
\[
S(\omega_{\mathcal C}\|\omega_{\mathcal B})
=
S(\mathcal B)-S(\mathcal C)~, \quad S(\omega_{\mathcal A}\|\omega_{\mathcal D})
=
S(\mathcal D)-S(\mathcal A)~.
\]
Inserting these into our previous inequality then gives
\[
S(\mathcal A)+S(\mathcal B)
\ge
S(\mathcal D)
+
S(\mathcal C)~,
\]
as desired.
\end{proof}

\subsection{Relation to the earlier entropy formula}

The entropy (\ref{SubalgebraEntropy}) should be contrasted with the entropy formula
\begin{equation}
\widetilde{S}(\mathcal A) := S(\omega|\Acal)
-
\log\Ind(E_{\Omega \to\mathcal A})
+
K_{\Omega}
\label{SubalgebraEntropy2}
\end{equation}
which was introduced in \cite{Sahu:2025upe} and analyzed in the tensor network setting in section \ref{subsec:generalized-entropy}. While both are defined in terms of relative entropy and conditional expectations, they are associated with different coarse-graining procedures. The previous definition \eqref{SubalgebraEntropy2} coarse-grains to the subalgebra $\mathcal A$ and corrects by the index, whereas (\ref{SubalgebraEntropy}) coarse-grains to its commutant $\mathcal A'$. Consequently, the two entropies are different a priori. In this section, we will identify under what conditions the two definitions coincide.

Let us evaluate both formulas in the tensor network setting, where we have finite-dimensional factor algebras. Given an extremal region $A$ and corresponding boundary algebra $\mathcal{M}_A$ with mixed input state $\rho_{A}$, it was shown in Proposition \ref{prop:tensor-network-generalized-entropy-mixed} that choosing $K_{\Omega}=2\log D_{\partial}$ in \eqref{SubalgebraEntropy2} gives:
\begin{equation} \label{eq:Stilde}
\widetilde{S}(\mathcal{M}_A) = \log D_A + S(\rho_A)~.
\end{equation}

We can also evaluate the alternative entropy formula
\eqref{SubalgebraEntropy} directly for the mixed tensor network state. The
orthogonal isometries \(T^{(a)}\) define a unitary identification
\[
\mathcal H_\partial
\simeq
\mathcal H_A\otimes\mathbb C^m~,
\qquad
T^{(a)}|\phi\rangle
\longleftrightarrow
|\phi\rangle\otimes|a\rangle~,
\]
where $m=D_\partial/D_A$. In this representation, the algebra associated with the region \(A\) and
its commutant take the standard tensor-product form
\[
\mathcal M_A
=
B(\mathcal H_A)\otimes 1_m~,
\qquad
\mathcal M_A'
=
1_A\otimes B(\mathbb C^m)~.
\]
Moreover, choosing \(T_A=T^{(1)}\), the mixed boundary state
\[
\rho_\Psi=T_A\rho_AT_A^\dagger
\]
is represented by\footnote{Although we assumed faithfulness above \eqref{SubalgebraEntropy2} for convenience, it is not required for the definition of Araki relative entropy. The definition extends naturally to non-faithful normal states by restricting the relative modular operator to the appropriate support subspaces. In particular, the relative entropy is finite only if the support of the first state is contained in that of the second. In finite dimensions, this reduces to the usual relative entropy formula $S(\rho\|\sigma)=\Tr(\rho(\log\rho - \log\sigma))$ for density matrices $\rho, \sigma$.}
\[
\rho_\Psi
=
\rho_A\otimes |1\rangle\langle1|~.
\]

The trace-preserving conditional expectation onto the commutant is
\[
E_{\mathcal M_A'}(X)
=
1_A\otimes\tr_A(X)~,
\]
where $\tr_A=\frac{1}{D_A}\Tr_A$ is the normalized trace on \(B(\mathcal H_A)\). Consequently,
\[
(\omega\circ E_{\mathcal M_A'})(X)
=
\Tr\!\left[
\left(
\frac{1_A}{D_A}\otimes|1\rangle\langle1|
\right)X
\right],
\]
so that the coarse-grained state is represented by
\[
\frac{1_A}{D_A}\otimes|1\rangle\langle1|~.
\]

By a direct computation, we obtain
\begin{align*}
S(\omega\|\omega\circ E_{\mathcal M_A'})
&=
S\!\left(
\rho_A\otimes|1\rangle\langle1|
\,\middle\|\,
\frac{1_A}{D_A}\otimes|1\rangle\langle1|
\right) \\
&=
\Tr(\rho_A\log\rho_A)
-
\Tr\!\left[
\rho_A
\log\!\left(\frac{1_A}{D_A}\right)
\right] \\
&=
\log D_A-S(\rho_A)~.
\end{align*}
Thus, for the mixed tensor network state,
\begin{equation}
S(\mathcal M_A)
=
S(\omega\|\omega\circ E_{\mathcal M_A'})
=
\log D_A-S(\rho_A)~.
\end{equation}
Comparing this with \eqref{eq:Stilde} we find
\[
\widetilde S(\mathcal M_A)-S(\mathcal M_A)
=
2S(\rho_A)~.
\]
The two entropy prescriptions therefore coincide whenever the input state
\(\rho_A\) is pure, while they differ by twice the bulk entropy for a general mixed input state.

\subsection*{Extending the boundary algebra}

It is important to note that the value of $S(\mathcal{A})$ in \eqref{SubalgebraEntropy} depends on the choice of ambient algebra $\mathcal{M}$ in which the commutant $\mathcal{A}'$ is taken. In comparing to our previous entropy proposal, we choose $\mathcal{M}$ to be equal to the boundary algebra $\Omega = B(\mathcal{H}_{\partial})$. In principle, however, different choices of ambient algebra are possible.  Another natural choice is to introduce an additional copy of the boundary system $\overline{\mathcal{H}}_{\partial}$ that acts as a purifying system for states on $\mathcal{H}_{\partial}$. The Hilbert space $\mathcal{H}_{\partial}\otimes \overline{\mathcal{H}}_{\partial}$ containing two copies of the boundary system is exactly the GNS representation of the boundary algebra $\Omega$ (for a faithful state) on which the algebra $\Omega$ is represented by $ B(\mathcal{H}_{\partial})\otimes \mathbb{1}$. Let   $\mathcal{M}=B(\mathcal{H}_{\partial}\otimes \overline{\mathcal{H}}_{\partial})$ be the algebra associated to the full system. Besides the isometry $T_A: \mathcal{H}_A\to \mathcal{H}_{\partial}$, we now have a second isometry $\overline{T}_A: \overline{\mathcal{H}}_A\to \overline{\mathcal{H}}_{\partial}$ that is constructed from a copy of the tensor network.

Given a mixed state $\rho_{A}$ on $\mathcal{H}_A$ we can purify it to a state $|\phi_A\rangle$ on $\mathcal{H}_A\otimes\overline{\mathcal{H}}_A$, and use the isometries $T_A$ and $\overline{T}_A$ to map it into two copies of the boundary Hilbert space:
\[
\rho_{\Phi}=T_{A}\otimes \overline{T}_A |\phi_A\rangle \langle \phi_A|T_{A}^{\dagger}\otimes \overline{T}_A^{\dagger}~.
\]
Using the identification 
\[
\mathcal H_\partial \otimes \overline{\mathcal H}_\partial 
\simeq 
\mathcal H_A\otimes \overline{\mathcal H}_A\otimes \mathbb C^m\otimes \mathbb C^m~,
\qquad
T^{(a)}\otimes \overline{T}^{(b)}|\phi\rangle
\longleftrightarrow
|\phi\rangle\otimes|a\rangle \otimes |b\rangle~,
\]
we can identify the above state with 
$$
\rho_{\Phi} \simeq|\phi_{A}\rangle \langle \phi_{A}|\otimes |1\rangle \langle 1| \otimes |\overline{1}\rangle \langle \overline{1}|~.
$$
Taking the commutant within the ambient algebra $\mathcal{M}$, we have
\[
\mathcal M_A
=
B(\mathcal H_A)\otimes 1_{\overline{A}}\otimes 1_m \otimes 1_{\overline{m}}~,
\qquad
\mathcal M_A'
=
1_A\otimes B(\overline{\mathcal H}_A) \otimes  B(\mathbb C^m)\otimes B(\mathbb C^m)~.
\]
Note that the density matrix associated to the coarse-grained state $\omega\circ E_{\mathcal{M}\to \mathcal M_A'}$ is now given by:
$$
\frac{1_A}{D_A}\otimes \overline{\rho}_{A}\otimes |1\rangle \langle 1| \otimes |\overline{1}\rangle \langle \overline{1}|~,
$$
where $\overline{\rho}_{A} = \Tr_{\mathcal{H}_A}|\phi_{A}\rangle \langle \phi_{A}|$. Therefore,
\begin{align*}
S(\omega\|\omega\circ E_{\mathcal M \to \mathcal M_A'})
&=
S\!\left(
\rho_{\Phi}
\,\middle\|\,
\frac{1_A}{D_A}\otimes \overline{\rho}_{A}\otimes |1\rangle \langle 1| \otimes |\overline{1}\rangle \langle \overline{1}|
\right) \\
&=
\Tr\left[|\phi_{A}\rangle \langle \phi_{A}|\log|\phi_{A}\rangle \langle \phi_{A}|\right]
-
\Tr\!\left[
|\phi_{A}\rangle \langle \phi_{A}|
\log\!\left(\frac{1_A}{D_A}\otimes \overline{\rho}_{A}\right)
\right] \\
&=
\log D_A+S(\overline{\rho}_A) = \log D_A+S(\rho_A)~,
\end{align*}
where at the last equality we used that $S(\rho_A)=S(\overline{\rho}_A)$, since the state $|\phi_A\rangle$ on the full system is pure. We conclude that the entropy defined in \eqref{SubalgebraEntropy} reproduces our generalized entropy proposal:
\[
S(\mathcal M_A)= \widetilde S(\mathcal M_A)~,
\]
when the commutant $\mathcal{M}_A'$ is defined with respect to the doubled boundary system. Note that the entropy on the right-hand side is defined with respect to the state $\omega_{\Omega}$ on $\Omega$, while the left-hand side uses  its purification $\omega_{\Phi}$ on $\mathcal{M}$.

We have thus established the equivalence of the two entropy proposals for finite-dimensional factors in the tensor network framework, provided that we work in the GNS representation of the boundary algebra $\Omega$. In Appendix~\ref{app:entropy_equivalence}, we investigate the extent to which this relation extends to more general von Neumann algebras.

\subsection{Entropies for corner algebras}

In order to compare with our earlier entropy proposal, the calculation in the previous subsection assumed that $\mathcal A$ was a factor. However, one of the benefits of the alternative proposal \eqref{SubalgebraEntropy} is that it holds more generally, e.g. in the case where $\mathcal{A}$ is a corner algebra defined in section \ref{sec:algebras_extremal_regions}. We now
consider the situation in which
\[
\mathcal C_A=P_A B(\mathcal H_\partial)P_A~,
\]
is the corner algebra associated with the projector $P_A$. This subalgebra does \emph{not} generally contain the identity operator in $B(\mathcal{H}_\partial)$. Since
$\mathcal C_A$ is non-unital, its commutant $\mathcal C_A'$ has a non-trivial center. Writing
\[
\mathcal H_\partial
=
P_A\mathcal H_\partial
\oplus
(1-P_A)\mathcal H_\partial~,
\]
and passing to the GNS representation $\mathcal M
=
B(\mathcal H_\partial\otimes\overline{\mathcal H}_\partial)$, the corresponding commutant within $\mathcal{M}$ is
\[
\mathcal C_A'
=
\left(
\mathbb C P_A
\oplus
B((1-P_A)\mathcal H_\partial)
\right)
\otimes
B(\overline{\mathcal H}_\partial)~,
\]
with center $Z(\mathcal{C}_A')=\left(\mathbb C P_A
\oplus
\mathbb{C}(1-P_A)\right)\otimes 1$.

The trace-preserving conditional expectation onto $\mathcal C_A'$ acts by
depolarizing the $P_A$-sector while leaving the orthogonal complement
unchanged:
\begin{equation}
E_{\mathcal M\to\mathcal C_A'}(X)
=
\frac{P_A}{D_A}
\otimes
\Tr_{P_A\mathcal H_\partial}
\!\left[
(P_A\otimes1)
X
(P_A\otimes1)
\right]
+
((1-P_A)\otimes1)
X
((1-P_A)\otimes1)~.
\label{CornerExpectation}
\end{equation}

For the tensor-network state
\begin{equation} \label{eq:tensor_network_state}
\rho_\Phi
=
(T_A\otimes\overline T_A)
|\phi_A\rangle\langle\phi_A|
(T_A^\dagger\otimes\overline T_A^\dagger)~,
\end{equation}
the support of the first boundary factor lies entirely inside the image of
$P_A$. Equivalently,
\[
(P_A\otimes1)\rho_\Phi(P_A\otimes1)=\rho_\Phi~,
\qquad
((1-P_A)\otimes1)\rho_\Phi=0~.
\]
Hence, the second term in \eqref{CornerExpectation} vanishes, and the
conditional expectation simply depolarizes the degrees of freedom inside the
corner algebra:
\[
\omega\circ E_{\mathcal M\to\mathcal C_A'}
\quad\leftrightarrow\quad
\frac{P_A}{D_A}
\otimes
\Tr_{P_A\mathcal H_\partial}(\rho_\Phi)= \frac{1_A}{D_A}
\otimes
\overline\rho_A
\otimes
|1\rangle\langle1|
\otimes
|\overline1\rangle\langle\overline1|~.
\]

Since $\rho_\Phi$ is pure, $S(\rho_\Phi)=0$, and the relative entropy reduces to the entropy of the coarse-grained state:
\begin{align}
S(\omega\|\omega\circ E_{\mathcal M\to\mathcal C_A'})
&=
S\!\left(\rho_{\Phi}\,\middle\|\,
\frac{1_A}{D_A}
\otimes
\overline\rho_A
\otimes
|1\rangle\langle1|
\otimes
|\overline1\rangle\langle\overline1|
\right)
\nonumber\\
&=
-\Tr\left[|\phi_A\rangle \langle \phi_A|\left(-\log D_A (1_A\otimes 1) + 1\otimes \log \overline{\rho}_A \right) \right]
\nonumber\\
&=
\log D_A
+
S(\rho_A)~,
\end{align}
where we used that
$|\phi_A\rangle$ is a purification of $\rho_A$, implying
$S(\rho_A)=S(\overline\rho_A)$ as in the previous subsection. 

We conclude that the entropy formula for the corner algebra agrees with the one for the factor algebra on the particular tensor network state in \eqref{eq:tensor_network_state}. If we had considered a more general input state (e.g., with support on the orthogonal complement $(1-P_A)\mathcal{H}_{\partial}$), we would have generically found differing results.

\section{Discussion}

In this paper, we have explored our suggested connection \cite{Sahu:2025upe} between generalized entanglement wedges and subalgebras for the semiclassical limit of a quantum gravity theory by making use of the random tensor network toy model for holographic quantum gravity. We found that this correspondence is sharp for random tensor networks in the limit of large bond dimension or more generally for isometric tensor networks with isometries associated to extremal regions. An interesting aspect of the connection is that in order to assign a von Neumann factor subalgebra to an extremal region, it is necessary to make a choice from a unitary family. This is directly related to the freedom of making a choice of operator map in a quantum code \cite{FactorCode}. 

For several of the suggested properties of the extremal region/subalgebra map we found that by appropriate choices of the associated boundary von Neumann algebras (and sometimes for almost all choices), it is possible to realize these properties in the tensor network situation. However, in some cases, this realization required a restriction on the choice of regions. For example, for extremal regions $A$ and $B$, we were able to show that boundary algebras associated with $A \cap B$, $A$, $B$, and $A \vee B$ can be chosen with the same inclusion structure as the regions when $A \vee B = A \cup B$ (i.e. when $A \cup B$ is extremal) but we were not able to show this more generally. Relatedly, we were able to show that the strong subadditivity property of extremal regions follows from algebraic strong subadditivity in this case with $A \vee B = A \cup B$, but were not able to show this more generally. 

Our results in the tensor network setting can be used to refine the proposed GW/subalgebra connection for general spacetimes. For example, based on the tensor network results, it appears natural to suggest that for algebras associated with GW regions 
$\mathcal{A}_{W_1}$ and $\mathcal{A}_{W_2}$, the algebras can be chosen to commute if and only if $W_1 \cap W_2$ is empty and $W_1 \cup W_2$ is a GW. 

In the present work, we have focused on random tensor networks without bulk legs. These model spacetime situations where the generalized gravitational entropy is dominated by the area term. It would be interesting to consider the more general situation with bulk legs. In this case, it is natural to consider the more general situation where the bulk state (and thus also the boundary state) is mixed. Then the boundary mixed state can be represented by the density operator $T_{Ii} \rho^i {}_{j} (T^\dagger)^{j J}$ where $i$ and $j$ represent the bulk legs of the tensor network, $I$ and $J$ represent the boundary legs. For an interior region $A$ of the tensor network, we can write $T_{Ii} = (T_A)_{i_A a} (T_{\bar{A}})^{a} {}_{ I i_{\bar{A}}}$ where we have divided $i = (i_A i_{\bar{A}})$ and $a$ represent indices associated with the boundary legs for the region. In this case, removing the tensors $(T_A)_{i_A a}$ associated with a region $A$ leaves a map
\[
(T_{\bar{A}})^{a} {}_{ I i_{\bar{A}}} \rho^{i_A i_{\bar{A}}} {}_j (T^\dagger)^{j J}
\]
from the Hilbert space associated with indices $(a, i_A)$ (i.e. the boundary and bulk legs for the region) to the doubled Hilbert space $(I,J)$ associated with the canonical purification of the region. This appeared in the work \cite{Kaya2025HollowGrams} that we discuss below. It would be interesting to extend our investigations from this paper to this context. 

We have furthermore used the ultraproduct construction to associate a limiting von Neumann algebra to a sequence of finite-dimensional tensor network algebras. While the tensor network algebras at finite bond dimension are finite-dimensional matrix algebras, their tracial ultraproduct is a type $II_1$ factor. Relations that hold only asymptotically at finite bond dimension, such as approximate inclusion or commutation, can become exact relations in the ultraproduct algebra. It would be interesting to study the inclusion structure of the resulting limiting algebra in more detail. For example, one could try to find a natural sequence of finite-index inclusions with a uniform bound on the index to obtain a finite-index inclusion of the corresponding ultraproduct algebras. Understanding the structure of these inclusions may provide a useful way to characterize how geometric and entanglement properties of the tensor network are encoded in the operator algebras at large bond dimension.

\paragraph{Comments on relation to recent work}

Our construction is closely related to the hollow-grams random tensor network (RTN) approach of \cite{Kaya2025HollowGrams}. Both provide tensor network realisations of the Bousso--Penington proposal for generalized wedges.\footnote{This tensor network picture has also been used in \cite{Bozanic:2026ios} to model observers in quantum gravity who occupy a spacetime subregion. Whether our wedge algebras may be interpreted as observer algebras in the fundamental description is an interesting direction for future work.} In \cite{Kaya2025HollowGrams}, a bulk region is hollowed out of an RTN with bulk legs. This yields a state that depends on the region. The entropy of the legs associated with the region is then computed using the replica trick, and the dominant permutation saddle selects the generalized wedge. Our networks, by contrast, have no bulk legs. We identify generalized wedges directly from extremal cuts. Each wedge is associated with a boundary algebra, and its generalized entropy is an algebraic entropy \cref{eq:SgenA,Sgen} that depends on the tensor network state.

This raises a natural question. What algebraic description should be given to an arbitrary bulk region, and how should its information content be computed? We do not associate algebras with arbitrary regions. The tensor network nevertheless assigns a set of operators to any region, through the map built from the vertices exterior to the region (see \cref{fig:TNremove}). This map is not an isometry, as the map proceeds through a minimal cut, namely the generalized wedge of the region. The image of the operators on the region is therefore a linear subspace of boundary operators rather than an algebra. This subspace spans the algebra of the generalized wedge of the region. Its information content can be obtained in our prescription by considering the algebra of the generalised wedge and computing the entropy of that algebra with our formulae \cref{eq:SgenA,Sgen}. The prescription therefore returns the area of the minimal cut enclosing the region.

In \cite{Kaya2025HollowGrams}, the authors propose that strong subadditivity in the hollowed RTN construction holds for regions $a, b, c$ that satisfy additional homology constraints on their boundaries. Under these constraints, the relevant entropies can be viewed as subregion entropies of a single state $\ket{\psi_{abc}}$, created by hollowing out all three regions. The restricted inequality reduces to the one proposed by Bousso and Penington \cite{Bousso:2022hlz} if the regions satisfy an independence condition, namely that no region enters the entanglement wedge of the union of the other two.

In our algebraic framework, strong subadditivity is formulated directly for the boundary algebras associated with generalized wedges. All of these algebras act on a single boundary Hilbert space in a single state. As discussed in \cref{sec:CommutingSquare}, the algebraic SSA inequality holds when the corresponding algebras form a non-degenerate commuting square. In our RTN construction, this occurs when the union of the generalized wedges $A$ and $B$ is itself strongly extremal. For three regions $a$, $b$, $c$, we take $A = E(ac)$ and $B = E(bc)$. It is possible to construct examples in which $A \cup B$ is strongly extremal but $a, b, c$ fail the independence condition. The unrestricted inequality therefore holds in a strictly larger class of configurations than the one singled out by independence criteria of \cite{Kaya2025HollowGrams}. In upcoming work \cite{GW3D}, one of us shows that in vacuum AdS$_3$ the unrestricted strong subadditivity of generalized wedges follows from the strong subadditivity of von Neumann entropy in the dual CFT, giving further evidence for the validity of the unrestricted version of strong subadditivity.

The authors of \cite{Kaya2025HollowGrams} further suggest that the validity of SSA for unrestricted Bousso--Penington wedges is a special feature of semiclassical holographic states. Our algebras are fixed by the network and do not depend on the boundary state in which they are evaluated. Moreover, our SSA result holds at the algebraic level. When $A \cup B$ is strongly extremal, the inequality for the algebraic entropies therefore holds for any state on the boundary Hilbert space, not only for holographic states.

Recent works assign a pure state directly to a gravitating region, prepared by a gravitational path integral with fixed data at the edge of the region \cite{Bousso:2026btw, Bousso:2026txw, Wei:2026syv}. The entropies of subregions of this state are computed by minimal surface prescriptions, which in \cite{Bousso:2026btw} agree with the Bousso--Penington prescription. In our proposal \cite{Sahu:2025upe}, if the algebra of a GW is semifinite, the GW is associated with a density matrix obtained by restricting the global state to the algebra. A natural possibility is then that the pure states of \cite{Bousso:2026btw, Bousso:2026txw, Wei:2026syv} purify these density matrices, with additional degrees of freedom associated with the GW as the purifying system. This intuition is realised in random tensor networks, where the hollowed state of \cite{Kaya:2025vof} purifies the reduced state, obtained by restricting the tensor network state to the corner algebra. In gravity, these additional degrees of freedom may be edge modes on the boundary of the GW, as in the path integral over embeddings of \cite{Balasubramanian:2023dpj}. Gauge and diffeomorphism invariant tensor networks, in which the area operator is related to the entanglement of edge modes, provide a natural setting to test this idea \cite{Dong:2023, Akers_2024, Qi_2022, Balasubramanian:2025, Balasubramanian:2026vlr}.

\section*{Acknowledgements}

We would like to thank Marc Klinger, Nima Lashkari and Antony Speranza for valuable discussions. This work is supported in part by the National Science and Engineering Research Council of Canada (NSERC) and the Simons Foundation via a Simons Investigator Award. AS thanks Pratik Rath and Onkar Parrikar for discussions, and TIFR for its hospitality during a visit in which part of this work was completed. ChatGPT 5 and 6 were used in proofreading parts of this work and suggested the result in Footnote 2, as well as formalism and proof ideas used in sections 4, 5, and appendix B. 
\appendix

\section{Proof details}

\subsection{A reformulation using the Jones basic construction}
\label{app:proof1}

The entropy difference (\ref{Sdiff}) admits a natural reformulation in terms of
the Jones basic construction associated with the inclusion $\mathcal{A}_{W_2}\subset\mathcal{A}_{W_1}$. We first recall the necessary operator-algebraic ingredients.

Let $\mathcal N\subset\mathcal M$ be a finite-index inclusion of von Neumann algebras with faithful normal
conditional expectation
\[
E:\mathcal M\rightarrow\mathcal N~.
\]
The associated \emph{Jones basic construction} is the von Neumann algebra
\[
\mathcal M_1=\langle\mathcal M,e_{\mathcal N}\rangle~,
\]
generated by \(\mathcal M\) together with the Jones projection
\(e_{\mathcal N}\), the orthogonal projection from
\(L^2(\mathcal M)\) onto \(L^2(\mathcal N)\) in the standard
representation. The basic construction is itself a finite-index extension of
\(\mathcal M\),
\[
\mathcal N\subset\mathcal M\subset\mathcal M_1~,
\]
and is equipped with a canonical faithful normal conditional expectation
\[
\widehat E:\mathcal M_1\rightarrow\mathcal M~,
\]
called the \emph{dual conditional expectation}. It is characterized by
\[
\widehat E(xe_{\mathcal N}y)
=
(\Ind E)^{-1}xy ~,
\qquad x,y\in\mathcal M~,
\]
and satisfies $\Ind(\widehat E)=\Ind(E)$. Iterating this construction gives rise to the Jones tower
\[
\mathcal N
\subset
\mathcal M
\subset
\mathcal M_1
\subset
\mathcal M_2
\subset
\cdots~.
\]
For further details on the Jones basic construction and the dual conditional
expectation we refer to
\cite{Jones1983,PimsnerPopa1986,kosaki1986, Longo1989}, see also \cite{vanderHeijden:2024tdk,Corbeel:2026xqb} for a recent application in the context of the black hole information problem.

A second ingredient is the index-relative entropy identity for finite-index
inclusions. Let
\[
E':\mathcal N'\rightarrow\mathcal M'
\]
denote the conditional expectation associated with the commutant inclusion $\mathcal M'\subset\mathcal N'$. Then every faithful normal state \(\omega\) satisfies
\begin{equation}
S_{\mathcal M}(\omega\|\omega\circ E)
+
S_{\mathcal N'}(\omega'\|\omega'\circ E')
=
\log\Ind(E)~,
\label{entropy-index}
\end{equation}
where \(\omega'\) is the corresponding state on the commutant algebra.
More details on the above identity can be found in, e.g., \cite{kosaki1998type, Xu2020}. For related work on studying entropic quantities in terms of the Jones index of algebra inclusions, the associated ``entropic certainty principle'' \eqref{entropy-index} and its application to QFT, see e.g., \cite{Casini:2019kex, Magan:2020ake}.

We now show that the right-hand side of (\ref{Sdiff}) may be interpreted as a
single relative entropy on the Jones basic construction.

\begin{proposition}
Let
\[
M
=
\left\langle
\mathcal{A}_{W_1},e_{\mathcal{A}_{W_2}}
\right\rangle
\]
be the Jones basic construction associated with the inclusion $\mathcal{A}_{W_2}\subset\mathcal{A}_{W_1}$
and let $\widehat E:M\rightarrow\mathcal{A}_{W_1}$ be the corresponding dual conditional expectation. Then there exists a faithful
normal state \(\omega_M\) on \(M\), canonically associated with the state
\(\omega\), such that
\begin{equation}
\label{eq:Jones_basic_construction}
S_{\rm gen}(W_1)-S_{\rm gen}(W_2)
=
S_M
\left(
\omega_M
\middle\|
\omega_M\circ\widehat E
\right).
\end{equation}
\end{proposition}

\begin{proof}
Applying the index-relative entropy identity (\ref{entropy-index}) to the
inclusion $\mathcal{A}_{W_2}\subset\mathcal{A}_{W_1}$ gives
\begin{equation}
\log\Ind(E)
-
S_{\mathcal{A}_{W_1}}
(\omega_1\|\omega_1\circ E)
=
S_{\mathcal{A}_{W_2}'}
(\omega_2'\|\omega_2'\circ E')~,
\label{certainty-rewritten}
\end{equation}
where
\[
E
=
E_{\mathcal{A}_{W_1}\rightarrow\mathcal{A}_{W_2}}~,
\qquad
E'
=
E_{\mathcal{A}_{W_2}'\rightarrow\mathcal{A}_{W_1}'}~.
\]

Let \(J_1\) denote the modular conjugation associated with the standard form of
\(\mathcal{A}_{W_1}\). A fundamental property of the Jones basic construction
is the canonical identification
\[
M
=
J_1\mathcal{A}_{W_2}'J_1~,
\]
which extends the standard relation $\mathcal{A}_{W_1}'
=
J_1\mathcal{A}_{W_1}J_1$. Equivalently, introducing the involutive normal
\(*\)-anti-isomorphism
\[
j(x)=J_1x^*J_1~,
\]
one obtains an anti-isomorphism $j:\mathcal{A}_{W_2}'
\longrightarrow
M$ under which the conditional expectation on the commutant inclusion is carried
to the dual conditional expectation,
\begin{equation}
\widehat E
=
j\circ E'\circ j~.
\label{dual-expectation}
\end{equation}

We now define the state on the basic construction by $\omega_M
=
\omega_2'\circ j$. Using \(j^2=\mathrm{id}\) together with
(\ref{dual-expectation}), we have $\omega_M\circ\widehat E
=
\omega_2'\circ E'\circ j$. Since Araki's relative entropy is invariant under normal
\(*\)-isomorphisms and \(*\)-anti-isomorphisms, we find:
\begin{align}
S_M
(\omega_M\|\omega_M\circ\widehat E)
&=
S_M
(\omega_2'\circ j
\|
\omega_2'\circ E'\circ j)
\nonumber\\
&=
S_{\mathcal{A}_{W_2}'}
(\omega_2'\|\omega_2'\circ E')~.
\label{entropy-j}
\end{align}
Combining (\ref{certainty-rewritten}) with (\ref{entropy-j}) immediately yields
the required relation (\ref{eq:Jones_basic_construction}).
\end{proof}

\subsection{Intersection of extremal regions}

\label{app:intersection}

We show that in the tensor network the intersection of two extremal regions is
again extremal. For a set of
vertices \(A\), we define its area by the number of network edges with exactly
one endpoint in \(A\),
\[
\area(A)\equiv |\partial A|~.
\]
This area satisfies the submodularity relation:
\begin{equation}
\area(A)+\area(B)
\geq
\area(A\cup B)+\area(A\cap B)~.
\label{eq:area-submodularity}
\end{equation}
More precisely,
\[
\area(A)+\area(B)
=
\area(A\cup B)+\area(A\cap B)
+
2\,|\partial(A\setminus B,B\setminus A)|~,
\]
where \(\partial(A\setminus B,B\setminus A)\) denotes the set of edges connecting
\(A\setminus B\) to \(B\setminus A\). Equality in
\eqref{eq:area-submodularity} holds whenever there are no such edges.

\begin{proposition}
Let \(A\) and \(B\) be extremal regions in the tensor network.
Then \(A\cap B\) is also extremal. In other words, for
every region \(C\) satisfying
\[
A\cap B\subset C~,
\]
one has
\[
\area(C)\geq\area(A\cap B)~.
\]
\end{proposition}

\begin{proof}
Let \(C\supset A\cap B\). Since \(A\) is extremal
\[
\area(A)\leq\area(A\cup C)~.
\]
Applying submodularity to \(A\) and \(C\) gives
\[
\area(A\cap C)+\area(A\cup C)
\leq
\area(A)+\area(C)~.
\]
Combining these two inequalities, we obtain
\begin{equation}
\area(A\cap C)\leq\area(C)~.
\label{eq:first-intersection-bound}
\end{equation}

We now use that \(B\) is extremal. Since
\[
B\subset B\cup(A\cap C)~,
\]
we have
\[
\area(B)\leq\area\bigl(B\cup(A\cap C)\bigr)~.
\]
Submodularity applied to \(B\) and \(A\cap C\) then implies
\[
\area\bigl(B\cap A\cap C\bigr)
+
\area\bigl(B\cup(A\cap C)\bigr)
\leq
\area(B)+\area(A\cap C)~,
\]
and hence
\[
\area(A\cap B\cap C)\leq\area(A\cap C)~.
\]
Because \(A\cap B\subset C\), one has $A\cap B\cap C=A\cap B$. Together with \eqref{eq:first-intersection-bound}, this gives
\[
\area(A\cap B)
\leq
\area(A\cap C)
\leq
\area(C)~.
\]
Therefore, \(A\cap B\) is again extremal.
\end{proof}

\subsection{Relation between entropy formulas for general von Neumann factors}

\label{app:entropy_equivalence}
We now investigate to what extent the two entropy proposals agree for more
general von Neumann algebras. The two definitions have different domains of
applicability. The entropy defined by coarse-graining to the commutant is
well-defined whenever the corresponding conditional expectation exists,
whereas our original proposal requires a finite-index inclusion of von
Neumann factors
\[
\mathcal A\subset\Omega
\]
equipped with a faithful normal finite-index conditional expectation
\[
E_{\Omega\to\mathcal A}:\Omega\to\mathcal A~.
\]
In the remainder of this section, we restrict to this common setting and
identify the additional conditions under which the two entropy formulas
coincide. 

We work in the standard
representation of \(\Omega\) on \(\mathcal H_{\Omega }\) and set
\[
\mathcal M\equiv B(\mathcal H_{\Omega })~.
\]
Let \(\omega\) be a faithful normal state on \(\mathcal M \), and \(\omega_{\Omega}\), \(\omega_{\mathcal{A}}\) its restrictions to the algebras \(\Omega\) and \(\mathcal{A}\) respectively. We assume that the algebraic entropies $S(\omega|\Omega)$ and  $S(\omega|\mathcal{A})$ are well-defined and finite.

We define the commutants of $\mathcal A$ and $\Omega$ inside $\mathcal M$ respectively. The first proposition involves the
commutant of \(\mathcal A\) taken in the GNS algebra
\(\mathcal M\). We assume there exists a conditional expectation\footnote{Since \(\mathcal M\) is a type I von Neumann algebra, a conditional expectation
$E_{\mathcal M\to\mathcal A'}:\mathcal M\to\mathcal A'$ can exist only if \(\mathcal A'\) is itself type I. The statement might be generalized by replacing \(E_{\mathcal M\to\mathcal A'}\) with a normal faithful semifinite operator-valued weight $T_{\mathcal M\to\mathcal A'}$, provided that the relative entropy $S\!\left(
\omega
\middle\|
\omega|_{\mathcal A'}
\circ T_{\mathcal M\to\mathcal A'}
\right)$ is well-defined and finite. We leave such an investigation for future work.}
\[
E_{\mathcal M\to\mathcal A'}
:
\mathcal M\to\mathcal A'~.
\]
We then have the following identity:

\begin{proposition} \label{prop:algebra_generalization}
Let $E_{\Omega\to\mathcal A}
:
\Omega\to\mathcal A$
be a finite conditional expectation with corresponding dual 
\[
E_{\mathcal A'\to\Omega'}
:
\mathcal A'\to\Omega',
\]
and define $E_{\mathcal M\to\Omega'}
=
E_{\mathcal A'\to\Omega'}\circ E_{\mathcal M\to\mathcal A'}$.
We have
\begin{align}
S(\omega|\Acal)
-
\log\Ind(E_{\Omega\to\mathcal A}) + K_{\Omega}
=
S\!\left(
\omega
\middle\|
\omega_{\mathcal A'}
\circ E_{\mathcal M\to\mathcal A'}
\right)~,
\label{eq:entropy_equality_general}
\end{align}
where 
\[
K_{\Omega} := S(\omega\|\omega\circ E_{\mathcal M \to \Omega'}) -S(\omega |\Omega)~.
\]
\end{proposition}

\begin{proof}
This follows by rearranging the index relation
\eqref{entropy-index} for the relevant inclusions:
\begin{equation}
S\!\left(
\omega_{\Omega}
\middle\|
\omega_{\Omega}\circ E_{\Omega\to\mathcal A}
\right)
+
S\!\left(
\omega_{\mathcal A'}
\middle\|
\omega_{\Omega'}
\circ E_{\mathcal A'\to\Omega'}
\right)
=
\log\Ind(E_{\Omega\to\mathcal A})~.
\label{eq:standard-form-index-relation}
\end{equation}
The important point is that the second relative entropy is computed on
\(\Omega'\subset\mathcal A'\), where both commutants are taken in the
purification algebra
\(\mathcal M=B(\mathcal H_{\Omega})\), rather than inside
\(\Omega\). We then use the Pythagorean identity to rewrite
\[
S\!\left(
\omega_{\mathcal A'}
\middle\|
\omega_{\Omega'}
\circ E_{\mathcal A'\to\Omega'}
\right)
=
S\!\left(
\omega
\middle\|
\omega_{\Omega'}
\circ E_{\mathcal M\to\Omega'}
\right)
-
S\!\left(
\omega
\middle\|
\omega_{\mathcal A'}
\circ E_{\mathcal M\to\mathcal A'}
\right)~.
\]
Combining both equations gives 
\begin{align}
&
S\!\left(
\omega_{\Omega}
\middle\|
\omega_{\Omega}\circ E_{\Omega\to\mathcal A}
\right)
-
\log\Ind(E_{\Omega\to\mathcal A})
\nonumber\\
&\qquad
+
S\!\left(
\omega
\middle\|
\omega_{\Omega'}
\circ E_{\mathcal M\to\Omega'}
\right)
=
S\!\left(
\omega
\middle\|
\omega_{\mathcal A'}
\circ E_{\mathcal M\to\mathcal A'}
\right)~.
\nonumber 
\end{align}
We can use the following expression to relate the relative entropy to a difference in algebraic entropies:
\[
S(\omega\|\omega\circ E_{\Omega \to \mathcal A}) =S(\omega|\Acal)- S(\omega|\Omega)
~,
\]
which gives us the required result \eqref{eq:entropy_equality_general}.
\end{proof}
The constant $K_{\Omega}$ in Proposition \ref{prop:algebra_generalization} depends on the algebra $\Omega$, and the choice of global state $\omega$. As a sanity check, we can evaluate it in the case that the algebras are finite-dimensional factors of the form $\mathcal{M}=B(\mathcal{H}\otimes \overline{\mathcal{H}}), \Omega=B(\mathcal{H})\otimes 1$ and $ \Omega'=1\otimes B(\overline{\mathcal{H}})$. Recall that for an algebra of the form  $\mathcal{A}=B(\mathcal{H}_A)$ with $D_{\mathcal{A}} = \dim \mathcal{H}_A$ we have
\[
S(\omega|\mathcal{A}) = S(\rho_{\mathcal{A}})-\log D_{\mathcal A}~,
\]
where $S(\rho_{\mathcal{A}})$ is the von Neumann entropy of the reduced density matrix to the algebra $\mathcal{A}$. 
Taking the global state $\omega$ on $\mathcal{M}$ to be pure, we thus find:
\[
S(\omega\|\omega\circ E_{\mathcal M \to \Omega'}) = S(\omega|\Omega')- S(\omega|\mathcal M) = S(\rho_{\Omega'})-\log D_{\Omega'} -(S(\rho_{\mathcal{M}})-\log D_{\mathcal{M}} ) ~,
\]
Plugging this into the expression for $K_{\Omega}$, we obtain
\begin{align*}
K_{\Omega} &= S(\omega\|\omega\circ E_{\mathcal M \to \Omega'}) - S(\omega |\Omega) \\
&= S(\rho_{\Omega'})-S(\rho_{\Omega})-S(\rho_{\mathcal{M}}) + \log \frac{D_{\mathcal{M}}D_{\Omega}}{D_{\Omega'}} \\
& = \log D_{\Omega}^2~,
\end{align*}
where we used that  $D_{\mathcal{M}}=D_{\Omega}D_{\Omega'}$ and $S(\rho_{\Omega'})=S(\rho_{\Omega})$. In this case, $K_{\Omega}$ is independent of $\omega$.

\section{Averaging over Haar ensembles  }

\label{app:average}

\subsection{Tensor network averaging}
\label{app:haar}

In this section, we shall derive the master formula \eqref{eq:master-formula1}.

The assignment \eqref{eq:CJiso} defines, for each region $S$, the linear map built from the vertices in the complement of $S$,
\[
\hat T_S:\mathcal H_{S}\to\mathcal H_\partial~,
\qquad
\hat T_S=\bigotimes_{e\,\in\,\edge\setminus\edge_S}\bra{\Phi_e}
\bigotimes_{v\,\in\,\verts\setminus S}\ket{V_v}~.
\]
The domain $\mathcal H_S$ is spanned by the indices left dangling by the removal of the vertices of $S$, namely those on the edges of $\bdy S$. Hence
\begin{equation}
D_S\equiv\dim\mathcal H_S=D^{|\bdy S|}~.
\label{eq:DA-dim}
\end{equation}
After the normalization of \eqref{eq:normalized_TA}, $\hat T_A$ defines the positive operator
\begin{equation}
P_A=\frac{D_A}{\tr\big(\hat T_A\hat T_A^\dagger\big)}\,\hat T_A\hat T_A^\dagger~,
\qquad
\tr(P_A)=D_A~,
\label{eq:PA-def}
\end{equation}
where the second equality is exact.

For a collection of regions $A_{i_{\alpha,k}}$
where the indices take the value $(i_{1,1},\dots,i_{m,k_m})=(1,\dots,n)$, we have using \eqref{eq:PA-def}, 
\begin{align}
\prod_{\alpha=1}^{m}\tr\!\left(P_{A_{i_{\alpha,1}}}\cdots P_{A_{i_{\alpha,k_\alpha}}}\right)
&=
\left(\prod_{i=1}^{n}\frac{D_{A_i}}{\tr\big(\hat T_{A_i}\hat T_{A_i}^\dagger\big)}\right)
\prod_{\alpha=1}^{m}
\tr\!\left(
\hat T_{A_{i_{\alpha,1}}}\hat T_{A_{i_{\alpha,1}}}^{\dagger}
\cdots
\hat T_{A_{i_{\alpha,k_\alpha}}}\hat T_{A_{i_{\alpha,k_\alpha}}}^{\dagger}
\right).
\label{eq:Master_Trace}
\end{align}
 We approximate the average of \eqref{eq:Master_Trace} by the ratio of the averages of numerator and denominator. We justify this in \ref{app:ratio-of-averages approximation}.

Given a permutation $\sigma \in S_n$, let $\Sigma(\sigma)$ act on an $n$-fold replicated space by
\begin{equation}
\Sigma(\sigma)\ket{i_1}\otimes\cdots\otimes\ket{i_n}
=\ket{i_{\sigma(1)}}\otimes\cdots\otimes\ket{i_{\sigma(n)}}~,
\label{eq:rep-operator}
\end{equation}
which is a representation of $S_n$. A trace of a product of operators on $ B(\Hcal_\partial)$ is reproduced by placing one operator on each replica and gluing the replicas cyclically. With the convention \eqref{eq:rep-operator} one has $\Tr[(O_1\otimes\cdots\otimes O_n)\Sigma(\sigma)]=\prod_{\text{cycles}}\tr(\cdots)$. Thus we have
\begin{equation}
\prod_{\alpha=1}^{m}\tr\!\left(
\hat T_{A_{i_{\alpha,1}}}\hat T_{A_{i_{\alpha,1}}}^{\dagger}
\cdots
\hat T_{A_{i_{\alpha,k_\alpha}}}\hat T_{A_{i_{\alpha,k_\alpha}}}^{\dagger}
\right)
=
\Tr\!\left[
\bigotimes_{i=1}^{n}\big(\hat T_{A_i}\hat T_{A_i}^\dagger\big)\,
\Sigma\big(\kappa\big)
\right],
\label{eq:Replica_master}
\end{equation}
where
\begin{equation}
\kappa=\prod_{\alpha=1}^{m}\big(i_{\alpha,1}\;i_{\alpha,2}\;\cdots\;i_{\alpha,k_\alpha}\big)
\label{eq:kappa}
\end{equation}
is the product of one cycle per trace factor. The operator $\Sigma(\kappa)$ acts on $\Hcal_\partial^{\otimes n}$, which factorises over the dangling legs. It therefore factorises as $\bigotimes_{\ell\in\legs}\Sigma_\ell(\kappa)$. Each leg may be regarded as an edge joining a network vertex to an external vertex $b\in\extv$ carrying the fixed permutation $\pi_b=\kappa$. This is the origin of the boundary condition imposed in \cref{sec:main}.

For a Haar random state $\ket V$ in dimension $d$, the average $\overline{(\ket V\bra V)^{\otimes k}}$ commutes with $U^{\otimes k}$ for every unitary $U$ which restricts its form to be proportional to a projector onto the symmetric subspace of the $k-$fold
replicated Hilbert space \cite{Hayden:2016cfa}. Normalizing the trace to unity gives
\begin{equation}
\overline{\left(\ket{V}\bra{V}\right)^{\otimes k}}
=\frac{1}{(d)_{k}}\sum_{\sigma\in S_{k}}\Sigma(\sigma)~,
\qquad
(d)_{k}=d(d+1)\cdots(d+k-1)~,
\label{eq:haar-moment}
\end{equation}
where we have used $\sum_{\sigma\in S_k}d^{\,\mathrm{cyc}(\sigma)}=(d)_k,$ the Pochhammer symbol.

The vertex $v$ is absent from $\hat T_{A_i}$ whenever $v\in A_i$. On the replicated space, therefore, the projector $\ket{V_v}\bra{V_v}$ appears once for each $i$ with $v\notin A_i$, a total of
\begin{equation}
n_v=\#\{i:v\notin A_i\}\le n
\label{eq:nv}
\end{equation}
times. Applying \eqref{eq:haar-moment} with $k=n_v$ produces permutations of those $n_v$ replicas alone. Extending them to $S_n$ by fixing the remaining labels, the permutation at $v$ is constrained to the set
\begin{equation}
H_v=\left\{\sigma\in S_n:\sigma(i)=i\ \text{whenever}\ v\in A_i\right\}\cong S_{n_v}~.
\label{eq:Hv}
\end{equation}
When $v$ is not in any region, $H_v=S_n$. Averaging vertex by vertex thus replaces the random tensors by a single permutation variable per vertex,
\begin{equation}
\overline{\Tr\!\left[\bigotimes_{i=1}^{n}\big(\hat T_{A_i}\hat T_{A_i}^\dagger\big)\Sigma(\kappa)\right]}
=\left(\prod_{v\in\verts}\frac{1}{(d_v)_{n_v}}\right)
\sum_{\substack{\pi_v\in H_v,\;v\in\verts\\ \pi_b=\kappa,\;b\in\extv}}
w[\{\pi_v\}]~,
\label{eq:replica-to-sum}
\end{equation}
where $w[\{\pi_v\}]$ is the contraction of the operators $\Sigma_v(\pi_v)$ along the edges of the extended network. Note that the boundary vertices are fixed to carry the permutation $\kappa$.

An edge $e=\langle xy\rangle$ is absent from replica $i$ whenever $e\in\edge_{A_i}$, so it carries $n_e=\#\{i:e\notin\edge_{A_i}\}$ copies of $\ket{\Phi_e}$. Contracting the two permutation operators through them gives
\begin{equation}
w_e(\pi_x,\pi_y)
=\bra{\Phi_e}^{\otimes n_e}\Sigma_x(\pi_x)\otimes\Sigma_y(\pi_y)\ket{\Phi_e}^{\otimes n_e}
=D^{\,n_e-d_{\mathrm C}(\pi_x,\pi_y)}~,
\label{eq:edge-weight}
\end{equation}
where
\begin{equation}
d_{\mathrm C}(\sigma,\rho)=n-\#\,\mathrm{cyc}\big(\sigma^{-1}\rho\big)
\label{eq:cayley}
\end{equation}
is the \emph{Cayley distance}, equal to the minimal number of transpositions carrying $\sigma$ into $\rho$. For \eqref{eq:edge-weight}, each cycle of the relative permutation closes one
index loop and contributes a factor of $D$. The exponent is therefore the
number of cycles that act within the $n_e$ replicas on which $\ket{\Phi_e}$ is
actually present. Now the replicas on which $e$ is absent are precisely those
fixed by both $\pi_x$ and $\pi_y$, and each contributes a fixed point of
$\pi_x\pi_y^{-1}$. Counting all of $\#\,\mathrm{cyc}(\pi_x\pi_y^{-1})$ thus
overstates the exponent by exactly $n-n_e$. Equation \eqref{eq:edge-weight}
records the corrected count,
$D^{\,\#\,\mathrm{cyc}(\pi_x\pi_y^{-1})-(n-n_e)}=D^{\,n_e-d_{\mathrm C}(\pi_x,\pi_y)}$.

Collecting the exponents, \eqref{eq:replica-to-sum} becomes a classical spin model on the extended graph in which each vertex carries an element of $S_n$ and $\log D$ plays the role of an inverse temperature,
\begin{equation}
w[\{\pi_v\}]=\left(\prod_{e\in\widehat\edge}D^{\,n_e}\right)e^{-\mathcal E[\{\pi_v\}]\log D}~,
\qquad
\mathcal E[\{\pi_v\}]=\sum_{(x,y)\in\widehat\edge}d_{\mathrm C}(\pi_x,\pi_y)~.
\label{eq:perm-hamiltonian}
\end{equation}
A configuration is a collection of domains, one for each permutation occurring in it, and only the edges crossing a domain wall contribute to $\mathcal E$. As $D\to\infty$ the sum is dominated by the configurations of least energy,
\begin{equation}
\mathcal E_{\kappa}(A_1,\ldots,A_n)
=\min_{\substack{\pi_v\in H_v,\;v\in\verts\\ \pi_b=\kappa,\;b\in\extv}}
\ \sum_{(x,y)\in\widehat\edge}d_{\mathrm C}(\pi_x,\pi_y)~,
\label{eq:min-energy}
\end{equation}
with the remaining configurations suppressed by integer powers of $D^{-1}$, since the energies are integers. The favoured configurations are those with the fewest and shortest domain walls. These walls are the tensor network counterparts of extremal surfaces, and their total weight is the area term in the entropy.
Assembling the results so far, we have the average of the numerator 
\begin{equation}
\label{eq:Real_Master_Average}
\overline{\prod_{\alpha=1}^{m}\tr\!\left(
\hat T_{A_{i_{\alpha,1}}}\hat T_{A_{i_{\alpha,1}}}^{\dagger}
\cdots
\hat T_{A_{i_{\alpha,k_\alpha}}}\hat T_{A_{i_{\alpha,k_\alpha}}}^{\dagger}
\right)} = \mathcal{N}D^{-\mathcal{E}_\kappa(A_{1},\dots,A_n) }(g_{\kappa}+O(D^{-1}))~,
\end{equation}
where
\begin{equation}
\mathcal N=\left(\prod_{v\in\verts}\frac{1}{(d_v)_{n_v}}\right)\prod_{e\in\widehat\edge}D^{\,n_e}~,
\label{eq:normalisation}
\end{equation}
$g_\kappa$ is the number of configurations attaining the minimum in \eqref{eq:min-energy} and $\kappa$ is defined in \cref{eq:kappa}.

We can rewrite $\mathcal N$ as a factor over replicas since $n_v=\#\{i:v\notin A_i\}$ and $n_e=\#\{i:e\notin\edge_{A_i}\}$, and at large $D$ we may use $(d_v)_{n_v}^{-1}=d_v^{-n_v}\big(1+O(d_v^{-1})\big)$. The replica $i$ factor is then
\[
\left(\prod_{v\notin A_i}d_v^{-1}\right)\prod_{e\notin\edge_{A_i}}D~,
\]
which is exactly the right hand side of \eqref{eq:replica-to-sum} evaluated at $n=1$ for the single region $A_i$. In this case the relevant permutation group is $S_1$ which is trivial and the energy vanishes. Hence
\begin{equation}
\mathcal N=\prod_{i=1}^{n}\overline{\Tr\big(\hat T_{A_i}\hat T_{A_i}^\dagger\big)}\ \big(1+O(D^{-1})\big)~,
\label{eq:N-factorises}
\end{equation}
so the normalisation in \eqref{eq:Master_Trace} cancels the prefactor of the averaged numerator identically. Using \eqref{eq:DA-dim} for the surviving factors $D_{A_i}$ we obtain
\begin{equation}
\begin{aligned}
&\overline{\prod_{\alpha=1}^{m}\tr\!\left(P_{A_{i_{\alpha,1}}}\cdots P_{A_{i_{\alpha,k_\alpha}}}\right)}
\\[2mm]
&\hspace{1.5cm}
=D^{\,\sum_{i=1}^{n}|\bdy A_i|-\mathcal E_{\kappa}(A_1,\ldots,A_n)}\big(g_\kappa+O(D^{-1})\big)~.
\end{aligned}
\label{eq:master-result}
\end{equation}
 The corrections are of two kinds, the expansion of the Pochhammer symbols in \eqref{eq:N-factorises} and the subleading permutation configurations, and both are suppressed by a single power of $D$.

At $n=2$, relevant to the R\'enyi entropy of \cref{eq:purity}, the group is $S_2=\{\Id,\,\text{swap}\}$ with $d_{\mathrm C}(\Id,\Id)=0$ and $d_{\mathrm C}(\Id,\text{swap})=1$. Identifying the two elements with spin values $\pm1$ reduces \eqref{eq:perm-hamiltonian} to the Ising model of \cref{eq:Ising}, and the minimal domain wall is the Ryu--Takayanagi surface.

\subsection{Minimum energy configurations}
\label{app:averageTN}

Let us now find the minimal energy configuration for some trace computations of interest.

\subsubsection{Moments of  $P_A$}

Let $A$ be a strongly extremal subregion in the tensor network. We want to show that the operator ${P}_A = {T}_A{T}_A^\dagger$ is a projector in the large $D$ limit. This is equivalent to saying that the leading order contribution in $1/D$ to
\begin{equation*}
    \tr{({P}_A^n)} = D_A^n\frac{\Tr( (\otimes^n \hat P_A \Sigma({c_n}))}{\Tr( \otimes^n \hat P_A)}~
\end{equation*}
must be independent of $n$. To calculate the average of the numerator, we perform the Haar average as described in the previous section, which produces a sum over all configurations with the following boundary condition: each boundary leg contains $\Sigma(c_n)$ and each edge on the boundary of $A$ contains $\Id$.

With the given boundary condition, the dominant contribution comes from a configuration where there is a $c_n$ domain homologous to the boundary separated from an identity domain homologous to the region $A$. 
For a generalized entanglement wedge, the domain wall is the boundary of the region $A$ itself. In this case each edge on the boundary of $A$ contributes $D_e^{1-n}$, since $d_C(\mathbb{1},c_n) = n-1$. The total contribution from all the legs on the boundary of $A$ is therefore given by $D_e^{(1-n)|\partial A|} = D_A^{1-n}$. 

Using the master formula \eqref{eq:master-result}, we thus obtain the leading order contribution \[
\overline{\tr(P_A^n)} =D_A^n  D_A^{(1-n)}= D_A~,
\]
which is independent of $n$. We conclude that $P_A$ is a projector in the large bond dimension limit. 

\subsubsection{Moments of $P_AP_B$}
Let $A$ and $B$ be strongly extremal regions with $A\cup B$ strongly extremal. We claim that the minimum of \eqref{eq:min-energy} is attained at
\begin{equation}
\sigma_v=
\begin{cases}
\Id~, & v\in A\cap B~,\\
c_{\rm even}=(2\,4\,\cdots\,2n)~, & v\in A\setminus B~,\\
c_{\rm odd}=(1\,3\,\cdots\,2n-1)~, & v\in B\setminus A~,\\
\tau=(1\,2\,\cdots\,2n)~, & v\in\overline{A\cup B}~,
\end{cases}
\label{eq:pApB-config}
\end{equation}
where $\overline{A \cup B}$ denotes the complement of $A\cup B$ in the tensor network. We prove this by showing that any admissible configuration can be changed into \eqref{eq:pApB-config} without raising the energy.
For future reference, we record the relevant Cayley distances for the above permutations
\begin{equation}
d_{\mathrm C}(\Id,\tau)=2n-1~,\quad
d_{\mathrm C}(\Id,c_{\rm even})=d_{\mathrm C}(\Id,c_{\rm odd})=n-1~,\quad d_{\mathrm C}(c_{\rm even},\tau)=d_{\mathrm C}(c_{\rm odd},\tau)=n~.
\label{eq:cycle-distances}
\end{equation}
We first prove the following lemma which is a consequence of extremality of a region $C$:
\begin{lemma}
     Let $f\ge0$ be a function on
$\verts\setminus C$ that vanishes on the external vertices. Then
\begin{equation}
\sum_{\substack{\{v,w\}\\ v,w\notin C}}\big|f(v)-f(w)\big|
\ \ge\
\sum_{\substack{\{u,v\}\\ u\in C,\ v\notin C}} f(v)~,
\label{eq:coarea}
\end{equation}
where the sum runs over edges of the tensor network.
\end{lemma}

\begin{proof}
For $t\ge0$ let $S_t=\{v\in\verts\setminus C: f(v)>t\}$. Since $f$ vanishes
on $\extv\setminus C$, the set $S_t$ is disjoint from $\extv$, so
extremality of $C$ gives $|\partial(C\cup S_t)|\ge|\partial C|$. Enlarging
$C$ by $S_t$ removes from the boundary those edges of $\partial C$ whose
outer endpoint lies in $S_t$, and adds those edges outside $C$ with exactly
one endpoint in $S_t$. Hence
\begin{equation}
\sum_{\substack{\{v,w\}\in\edge\\ v,w\notin C}}
\big|\mathbf 1_{S_t}(v)-\mathbf 1_{S_t}(w)\big|
\ \ge\
\sum_{\substack{\{u,v\}\in\edge\\ u\in C,\ v\notin C}}\mathbf 1_{S_t}(v)~.
\label{eq:coarea-slice}
\end{equation}
Integrating \eqref{eq:coarea-slice} over $t\in[0,\infty)$ and exchanging the
finite sums with the integral yields \eqref{eq:coarea}, by the identities
\begin{equation}
\int_0^\infty\big|\mathbf 1_{\{f(v)>t\}}-\mathbf 1_{\{f(w)>t\}}\big|\,dt
=\big|f(v)-f(w)\big|~,
\qquad
\int_0^\infty\mathbf 1_{\{f(v)>t\}}\,dt=f(v)~,
\label{eq:layer-cake}
\end{equation}
the first of which holds because the two indicators differ precisely for $t$
between $f(w)$ and $f(v)$.
\end{proof}

Let $\{\sigma_v\}$ be an admissible configuration. We first replace it by the
configuration that agrees with it on $A\cup B$ and satisfies $\sigma_v=\tau$
for every $v\in\overline{A\cup B}$, and we estimate the resulting change in
energy. Set $f(v)=d_{\mathrm C}(\sigma_v,\tau)$ for
$v\in\overline{A\cup B}$, which is non-negative and vanishes on the external
vertices. Edges internal to $A\cup B$ are unaffected. Edges internal to
$\overline{A\cup B}$ lose their entire energy, and by the triangle
inequality this loss is at least 
\[
\sum_{\substack{\{v,w\}\\ v,w\in\overline{A\cup B}}}
\big|f(v)-f(w)\big|~.
\]
Edges joining $A\cup B$ to $\overline{A\cup B}$ gain
$d_{\mathrm C}(\sigma_u,\tau)-d_{\mathrm C}(\sigma_u,\sigma_v)\le f(v)$
each, so the total gain is at most
\[
\sum_{\substack{\{u,v\}\\ u\in A\cup B,\ v\notin A\cup B}} f(v)~.
\]
By \eqref{eq:coarea} applied with $C=A\cup B$, the first sum dominates the
second, and the energy does not rise. We next replace the configuration so
obtained by the one satisfying $\sigma_v=c_{\rm even}$ for every
$v\in A\setminus B$, leaving it unchanged elsewhere, and we put
\begin{equation}
a_x=n-1-d_{\mathrm C}(\Id,\sigma_x)\ \ge\ 0~,
\qquad x\in A\setminus B~,
\label{eq:ax}
\end{equation}
the inequality holding because $\sigma_x$ permutes only the $n$ even labels
and therefore lies within distance $n-1$ of the identity. The energy lost on
an edge internal to $A\setminus B$ is at least $|a_x-a_{x'}|$, by the
triangle inequality applied to $\Id$, $\sigma_x$ and $\sigma_{x'}$. The
energy lost on an edge joining $A\setminus B$ to $B$ equals $-a_x$, since the
permutation at the far endpoint permutes only odd labels, whether that
endpoint lies in $B\setminus A$ or in $A\cap B$, so that both the old and the
new distance decompose by disjointness of supports. The energy lost on an
edge joining $A\setminus B$ to $\overline{A\cup B}$ is at least $a_x$,
since the new distance equals $n$ by \cref{eq:cycle-distances} while
\[
d_{\mathrm C}(\sigma_x,\tau)\ \ge\
d_{\mathrm C}(\Id,\tau)-d_{\mathrm C}(\Id,\sigma_x)=n+a_x~,
\]
and it is here that the preceding replacement on $\overline{A\cup B}$ is
required. Summing the three contributions, the total energy lost is at least
\begin{equation}
\sum_{\substack{\{x,x'\}\\ x,x'\in A\setminus B}}\big|a_x-a_{x'}\big|
+\sum_{\substack{\{x,r\}\\ x\in A\setminus B,\ r\in\overline{A\cup B}}} a_x
-\sum_{\substack{\{x,v\}\\ x\in A\setminus B,\ v\in B}} a_x~,
\label{eq:step-two-sum}
\end{equation}
which is non-negative by \eqref{eq:coarea} applied with $C=B$ and with $f$
taken to equal $a$ on $A\setminus B$ and to vanish on
$\overline{A\cup B}$. This choice is legitimate, since
$\verts\setminus B=(A\setminus B)\cup\overline{A\cup B}$ and $f$ vanishes on
the external vertices. Hence the energy again does not rise. Exchanging the
roles of $A$ and $B$ and repeating the argument, we may further set
$\sigma_v=c_{\rm odd}$ for every $v\in B\setminus A$ without raising the
energy. Since every vertex of $A\cap B$ necessarily carries the identity, the
resulting configuration is \eqref{eq:pApB-config}, which proves the claim.
Evaluating \eqref{eq:min-energy} on \eqref{eq:pApB-config} by means of
\eqref{eq:cycle-distances}, we obtain
\begin{equation}
\mathcal E_\kappa
=(n-1)\big(|\partial A|+|\partial B|\big)+\big|\partial(A\cup B)\big|~.
\label{eq:pApB-energy}
\end{equation}
The minimiser is unique provided $|\partial A|<|\partial(A\cup B)|$ and
$|\partial B|<|\partial(A\cup B)|$, and all three regions $A,B,A\cup B$ are strongly extremal. Should either inequality be an equality,
the minimum is degenerate and the factor $g_\kappa$ in
\eqref{eq:master-result} exceeds unity.

\subsubsection{Variance of the moments}

Consider fluctuations of
$
X_n \equiv \tr\!\left[(P_AP_B)^n\right]
$
as captured by the variance
$
    \overline{X_n^2}-\overline{X_n}^{\,2}.
$
The first term is
\[
    \overline{X_n^2}
    =
    \overline{
    \tr\!\left[(P_AP_B)^n\right]
    \tr\!\left[(P_AP_B)^n\right]
    }~.
\]
This is described by a $4n$-replica model with boundary permutation
\begin{equation}
r=(1\,2\,\cdots\,2n)(2n+1\,2n+2\,\cdots\,4n)~.
\end{equation}
The analogue of \eqref{eq:pApB-config} is
\begin{equation}
\sigma_v=
\begin{cases}
\Id~, & v\in A\cap B~,\\
c_{\rm even}^{(2)}~, & v\in A\setminus B~,\\
c_{\rm odd}^{(2)}~, & v\in B\setminus A~,\\
r~, & v\in\overline{A\cup B}~,
\end{cases}
\label{eq:variance-disc-config}
\end{equation}
where
\begin{align}
c_{\rm even}^{(2)}
&=(2\,4\,\cdots\,2n)(2n+2\,2n+4\,\cdots\,4n)~,
\nonumber\\
c_{\rm odd}^{(2)}
&=(1\,3\,\cdots\,2n-1)(2n+1\,2n+3\,\cdots\,4n-1)~.
\end{align}
Their relevant Cayley distances are
\begin{equation}
d_{\mathrm C}(\Id,r)=4n-2~,\qquad
d_{\mathrm C}(\Id,c_{\rm even}^{(2)})
=
d_{\mathrm C}(\Id,c_{\rm odd}^{(2)})
=2n-2~,
\end{equation}
and
\begin{equation}
d_{\mathrm C}(c_{\rm even}^{(2)},r)
=
d_{\mathrm C}(c_{\rm odd}^{(2)},r)
=2n~.
\end{equation}
The proof that \eqref{eq:variance-disc-config} minimizes the energy follows
the same sequence of replacements as in the previous section. First,
extremality of $A\cup B$ and \eqref{eq:coarea} allow us to set
$\sigma_v=r$ throughout $\overline{A\cup B}$ without increasing the
energy. One may then set $\sigma_v=c_{\rm even}^{(2)}$ on
$A\setminus B$ and, by exchanging $A$ and $B$,
$\sigma_v=c_{\rm odd}^{(2)}$ on $B\setminus A$.

There is only one additional possibility in the $4n$-replica problem.
A permutation supported on the $2n$ even labels can have
$d_{\mathrm C}(\Id,\sigma)=2n-1$, rather than at most $2n-2$, if it is a
single cycle joining the even labels of the two replica copies. In this
case $\sigma$ and $r$ act transitively on the $4n$ replicas, and the
cycle inequality gives
\begin{equation}
d_{\mathrm C}(\sigma,r)\ge 2n+1
>
d_{\mathrm C}(c_{\rm even}^{(2)},r)=2n~.
\end{equation}
Thus such replica-connecting configurations are not energetically favored. 

It is therefore convenient to define
\begin{equation}
a_x^+
=
\max\!\left\{
2n-2-d_{\mathrm C}(\Id,\sigma_x),\,0
\right\},
\qquad x\in A\setminus B~.
\end{equation}
With this replacement, the same estimates as in the previous subsection
give a lower bound on the energy lost of
\begin{equation}
\sum_{\substack{\{x,x'\}\\x,x'\in A\setminus B}}
|a_x^+-a_{x'}^+|
+
\sum_{\substack{\{x,r'\}\\
x\in A\setminus B,\ r'\in\overline{A\cup B}}}
a_x^+
-
\sum_{\substack{\{x,v\}\\
x\in A\setminus B,\ v\in B}}
a_x^+~,
\end{equation}
which is non-negative by \eqref{eq:coarea} applied with $C=B$, with $f$ taken to be equal to $a^{+}$ on $A\setminus B$ and zero elsewhere. Hence setting $\sigma_v=c_{\rm even}^{(2)}$ on $A\setminus B$ does not
raise the energy, and similarly for $c_{\rm odd}^{(2)}$ on
$B\setminus A$.

Evaluating the energy on \eqref{eq:variance-disc-config} gives
\begin{align}
\mathcal E_{\mathrm{min}}
&=
(2n-2)(|\partial A|+|\partial B|\big)
+2\big|\partial(A\cup B)\big|
\nonumber\\
&=
2\Big[
(n-1)(|\partial A|+|\partial B|\big)+\big|\partial(A\cup B)\big|
\Big]~.
\label{eq:variance-disc-energy}
\end{align}
Thus this saddle is the disconnected product of two copies of the
leading saddle for $\overline{X_n}$.
The leading disconnected contribution is
\begin{align}
\overline{X_n^2}
&=
D^{2n(|\partial A|+|\partial B|)-\mathcal E_{\mathrm{min}}}
\big(1+O(D^{-1})\big)
=
D^{2(|\partial A|+|\partial B|-|\partial(A\cup B)|)}
\big(1+O(D^{-1})\big)~.
\label{eq:variance-leading}
\end{align}
On the other hand, from \eqref{eq:pApB-n-leading} we find
\[
\overline{X_n}^{\,2}
=
D^{\,2(|\partial A| + |\partial B| - |\partial(A \cup B)|)}
\big(1+O(D^{-1})\big)~.
\]
Thus, the leading disconnected contribution cancels in the variance:
\[
\overline{X_n^2}-\overline{X_n}^{\,2}
=
D^{\,2(|\partial A| + |\partial B| - |\partial(A \cup B)|)}
O\!\left(D^{-1}\right)~,
\]
or equivalently
\[
\frac{
\overline{X_n^2}-\overline{X_n}^{\,2}
}{
\overline{X_n}^{\,2}
}
=
O\!\left(D^{-1}\right)~.
\]

\subsection{The ratio-of-averages approximation}
\label{app:ratio-of-averages approximation}

We now justify the replacement of the average of the normalized expression in \eqref{eq:Master_Trace} by the ratio of the corresponding averages. Define
\begin{equation}
Q_i\equiv \hat T_{A_i}\hat T_{A_i}^{\dagger},
\qquad
X\equiv
\prod_{\alpha=1}^{m}
\tr\left(
Q_{i_{\alpha,1}}\cdots Q_{i_{\alpha,k_\alpha}}
\right),
\qquad
Y\equiv
\prod_{i=1}^{n}\tr(Q_i)~.
\label{eq:ratio-XY-def}
\end{equation}
The quantity appearing in \eqref{eq:Master_Trace} is therefore
\begin{equation}
\left(\prod_{i=1}^{n}D_{A_i}\right)\frac{X}{Y}~,
\end{equation}
where the factors \(D_{A_i}\) are non-random and hence play no role in the following discussion.
Expanding around the mean of \(Y\), with
$$
\delta Y=Y-\overline{Y},
$$
gives
\begin{align}
\overline{\left(\frac{X}{Y}\right)}
={}&
\frac{\overline{X}}{\overline{Y}}\left[1
-\frac{\operatorname{Cov}(X,Y)}{\overline{X}\,\overline{Y}}
+\frac{\operatorname{Var}(Y)}{\overline{Y}^{2}}\right]
+\text{higher central moments}~,
\label{eq:ratio-expansion}
\end{align}
where the higher moments involve at least three error terms.
Thus it is sufficient to show that the connected fluctuations involving
\(Y\) are suppressed at large \(D\). 
We first consider the covariance
\begin{equation}
\operatorname{Cov}(X,Y)=\overline{XY}-\overline X \ \overline Y~.
\label{eq:cov-XY}
\end{equation}
The average \(\overline{XY}\) is a \(2n\)-replica calculation. We label the replicas associated with \(X\) by
\(1,\ldots,n\)
and those associated with \(Y\) by
\(n+1,\ldots,2n\).
The first \(n\) replicas are glued at the boundary according to the permutation
\(\kappa\) defined in \eqref{eq:kappa}, whereas each of the additional replicas arising from \(Y\) is traced individually. The boundary condition is therefore
\begin{equation}
\widetilde\kappa
=
\kappa (n+1)(n+2)\cdots(2n)~.
\label{eq:cov-boundary-condition}
\end{equation}
In other words, \(\kappa\) acts on the first replica sector while the second sector is fixed pointwise.
At a vertex \(v\), each region \(A_i\) now occurs once in each replica sector. Hence the vertex carries a permutation in
\begin{equation}
\widetilde H_v
=
\left\{
\sigma\in S_{2n} :
\sigma(i)=i,\;
\sigma(n+i)=n+i
\quad\text{whenever } v\in A_i
\right\}
\cong S_{2n_v}~.
\label{eq:cov-Hv}
\end{equation}
where \(n_v\) is defined in \eqref{eq:nv}. The averaged product therefore takes the form
\begin{equation}
\overline{XY}
=
\widetilde{\mathcal N}
\sum_{\substack{
\widetilde\pi_v\in\widetilde H_v;v\in\verts\
\widetilde\pi_b=\widetilde\kappa;b\in\extv}}
D^{-\widetilde{\mathcal E}[{\widetilde\pi_v}]}~,
\label{eq:cov-replica-sum}
\end{equation}
with
\begin{equation}
\widetilde{\mathcal E}[{\widetilde\pi_v}]
=
\sum_{(x,y)\in\widehat\edge}
d_{\mathrm C}(\widetilde\pi_x,\widetilde\pi_y)
\label{eq:cov-energy}
\end{equation}
and
\begin{equation}
\widetilde{\mathcal N}
=
\left(
\prod_{v\in\verts}\frac{1}{(d_v)_{2n_v}}
\right)
\prod_{e\in\widehat\edge}D^{2n_e}~.
\label{eq:cov-normalisation}
\end{equation}
Let
\(\{\pi_v^\star\}\)
be a configuration attaining the minimum
\(\mathcal E_\kappa(A_1,\ldots,A_n)\)
in \eqref{eq:min-energy}. There is an immediate \(2n\)-replica configuration
\begin{equation}
\widetilde\pi_v^{\rm disc}
=
\pi_v^\star (n+1)(n+2)\cdots(2n)~.
\label{eq:cov-disc-saddle}
\end{equation}
It obeys the boundary condition
$
\widetilde\pi_b^{\rm disc}
=
\tilde{\kappa}
$,
and, since the Cayley distance is unchanged by adjoining fixed points, its energy is therefore
\begin{equation}
\widetilde{\mathcal E}_{\rm disc}
=
\mathcal E_\kappa(A_1,\ldots,A_n)~.
\label{eq:cov-disc-energy}
\end{equation}
This is the disconnected saddle: the first replica sector reproduces the leading configuration of \(\overline X\), while the replicas belonging to \(Y\) remain everywhere in the identity permutation.

To see that no configuration can have lower energy, consider the map
\(R:S_{2n}\rightarrow S_n\)
obtained by deleting the labels
\(n+1,\ldots,2n\)
from the cycle decomposition of a permutation. It satisfies
\begin{equation}
R(\widetilde\kappa)=\kappa~,
\qquad
R(\widetilde H_v)= H_v~,
\end{equation}
Moreover, deleting replica labels cannot increase the Cayley distance. Indeed, \(d_{\mathrm C}(\sigma,\rho)\) is the minimum number of transpositions required to transform \(\sigma\) into \(\rho\). If
$$
(\sigma=\gamma_0)\rightarrow \gamma_1 \rightarrow \ldots \rightarrow (\gamma_k=\rho)~,
\qquad
k=d_{\mathrm C}(\sigma,\rho)~,
$$
is a shortest path in the Cayley graph, then each consecutive pair
\(\gamma_{j-1},\gamma_j\) differs by a single transposition. After deleting the labels \(n+1,\ldots,2n\), such a step either remains a single transposition on the surviving labels or becomes trivial. Hence the projected path from \(R(\sigma)\) to \(R(\rho)\) contains at most \(k\) nontrivial transpositions, and therefore
\begin{equation}
d_{\mathrm C}\big(R(\sigma),R(\rho)\big)
\leq
d_{\mathrm C}(\sigma,\rho)~.
\label{eq:cayley-projection}
\end{equation}
Consequently, for any admissible \(2n\)-replica configuration,
\begin{align}
\widetilde{\mathcal E}[{\widetilde\pi_v}]
&\geq
\mathcal E[{R(\widetilde\pi_v)}]
\geq
\mathcal E_\kappa(A_1,\ldots,A_n)~.
\label{eq:cov-energy-bound}
\end{align}
Together with \eqref{eq:cov-disc-energy}, this shows that
\begin{equation}
\widetilde{\mathcal E}_{\widetilde\kappa}
=
\mathcal E_\kappa(A_1,\ldots,A_n)~.
\label{eq:cov-min-energy}
\end{equation}

A configuration which connects the two replica sectors, or which acts nontrivially on the second sector, cannot be of the disconnected form
\(\pi_v (n+1)(n+2)\cdots(2n)\)
everywhere. Since the boundary condition on the second sector is the identity, such a configuration necessarily introduces an additional domain wall. Equivalently, the inequality in \eqref{eq:cayley-projection} is strict on at least one edge. 
The genuinely connected contribution to \(\overline{XY}\) is consequently suppressed by extra powers of $1/D$. 

It remains to compare the normalization of the disconnected \(2n\)-replica saddle with the product of the two separate averages. Using
\begin{equation}
(d_v)_{2n_v}^{-1}
=
d_v^{-2n_v}\big(1+O(d_v^{-1})\big)~,
\qquad
(d_v)_{n_v}^{-2}
=
d_v^{-2n_v}\big(1+O(d_v^{-1})\big)~,
\end{equation}
we have
\begin{equation}
\widetilde{\mathcal N}
=
\mathcal N^2\big(1+O(D^{-1})\big)~.
\label{eq:cov-normalisation-factorises}
\end{equation}
On the other hand, \(Y\) itself corresponds to the boundary condition
\(\kappa=\Id\). Its leading configuration has every permutation equal to the identity and zero energy, so
\begin{equation}
\overline Y
=
\mathcal N\big(1+O(D^{-1})\big)~.
\label{eq:Y-average}
\end{equation}
Combining this with \eqref{eq:Real_Master_Average}, the disconnected saddle of \(\overline{XY}\) reproduces
\begin{equation}
\overline{XY}_{\rm disc}
=
\overline X \ \overline Y
\big(1+O(D^{-1})\big)~.
\label{eq:cov-disc-factorisation}
\end{equation}
The leading disconnected contribution therefore cancels in
\(\overline{XY}-\overline X\,\overline Y\), leaving only the finite-\(D\)
corrections to the Haar normalization and configurations with positive
energy gap. Hence,
\begin{equation}
\frac{\operatorname{Cov}(X,Y)}
{\overline X \ \overline Y}
=
O(D^{-1})~.
\label{eq:cov-relative-suppression}
\end{equation}

The fluctuations of the denominator are suppressed in the same way. The quantity
\(\overline{Y^2}\)
is again a \(2n\)-replica calculation, now with boundary condition
\begin{equation}
\widetilde\kappa=\Id
\end{equation}
on all \(2n\) replicas. The leading configuration has every vertex in the identity permutation and has zero energy. Any configuration correlating the two copies of \(Y\), or acting nontrivially in either replica sector, has a strictly positive integer energy. Together with the \(O(D^{-1})\) correction from the Pochhammer factors, this gives
\begin{equation}
\frac{\operatorname{Var}(Y)}{\overline Y^{2}}
=
O(D^{-1})~.
\label{eq:Y-relative-variance}
\end{equation}
The same replica argument applies to higher connected moments involving additional copies of \(Y\): any contribution which does not factorize into
independent identity-boundary replica sectors requires additional
domain-wall energy and is suppressed by a positive power of \(D^{-1}\).

Returning to \eqref{eq:ratio-expansion}, the covariance, variance, and
higher-fluctuation terms therefore contribute only beyond leading order.
We conclude that
\begin{equation}
\overline{\left(\frac{X}{Y}\right)}
=
\frac{\overline X}{\overline Y}
\big(1+O(D^{-1})\big)~.
\label{eq:ratio-of-averages}
\end{equation}
Thus, to the order retained in \eqref{eq:master-result}, it is sufficient to
average the numerator and denominator of \eqref{eq:Master_Trace}
separately.

\subsection{A tensor network example with $A\cup B$ not extremal}

\label{app:nonextremal-counterexample}

We now evaluate the first two moments for the tensor network shown in Figure \ref{fig:nonextremal-counterexample}, which has the property that $A, B$ are extremal but $A\cup B$ is not. All edges have bond dimension $D$. We denote the permutations at the top, left, right, and central vertices by $\sigma_t$, $\sigma_\ell$, $\sigma_r$, and $\sigma_c$, respectively. The five physical boundary legs are fixed to $c_{2n}$, while the central vertex belongs to $C=A\cap B$ and is fixed to
$
\sigma_c=\Id~.
$
The permutations at the remaining random vertices are summed over the allowed replica groups.
For this network,
\begin{equation}
D_A=D_B=D^4~,
\qquad
D_E=D^5<D_{A\cup B}=D^6~,
\end{equation}
where $E$ is the region obtained by including the top vertex. Thus $A\cup B$ is not extremal.

\paragraph{First moment.}

For $n=1$, the odd- and even-replica groups are trivial, so
\begin{equation}
\sigma_\ell=\sigma_r=\sigma_c=\Id~,
\qquad
\sigma_t\in S_2~.
\end{equation}
The boundary permutation is $c_2=(12)$, and the action for the two possible values of $\sigma_t$ is
\begin{equation}
S_1(\sigma_t)
=
d_C(\sigma_t, c_2)
+2d_C(\sigma_t, \Id)
+4d_C(c_2, \Id)~.
\label{eq:nonextremal-n1-action}
\end{equation}
The three terms arise from the top boundary leg, the two edges connecting the top vertex to the lower vertices, and the four boundary legs attached to the lower vertices, respectively. The two edges connecting the lower vertices to $C$ have zero cost.
Using $d_C(c_2, \Id)=1$, we find
\begin{equation}
S_1(\Id)=5~,
\qquad
S_1(c_2)=6~.
\end{equation}
The dominant saddle therefore has $\sigma_t=\Id$. Geometrically, the identity domain extends throughout $E$, and the domain wall lies on $\partial E$. Its contribution is $D^{-5}=D_E^{-1}$, so
\begin{equation}
\overline{\Tr(P_AP_B)}
\sim
D_AD_BD_E^{-1}
=
D^8D^{-5}
=
D^3~.
\label{eq:nonextremal-example-n1}
\end{equation}

\paragraph{Second moment.}

For $n=2$, define
\begin{equation}
c=c_4=(1234)~,
\qquad
a=(24)~,
\qquad
b=(13)~.
\end{equation}
The permutations at the left and right vertices are restricted to the even- and odd-replica subgroups,
\begin{equation}
\sigma_\ell\in \{\Id,a\}~,
\qquad
\sigma_r\in \{\Id,b\}~,
\end{equation}
while $\sigma_t$ is summed over $S_4$.
The action is
\begin{align}
S_2(\sigma_t,\sigma_\ell,\sigma_r)
={}&
d_C(\sigma_t ,c)
+d_C(\sigma_t,\sigma_\ell)
+d_C(\sigma_t,\sigma_r)
\nonumber\\
&+
d_C(\sigma_\ell, \Id)
+d_C(\sigma_r, \Id)
+2d_C(\sigma_\ell, c)
+2d_C(\sigma_r, c)~.
\label{eq:nonextremal-n2-action}
\end{align}
The first term comes from the boundary leg attached to the top vertex. The next two terms come from the edges connecting the top vertex to the left and right vertices, while $d_C(\sigma_\ell)$ and $d_C(\sigma_r)$ come from the edges connecting these vertices to $C$. The final two terms account for the two boundary legs attached to each lower vertex.

The relevant distances are
\begin{equation}
d_C(c, \Id)=3~,
\qquad
d_C(a, \Id)=d_C(b, \Id)=1~,
\end{equation}
and
\begin{equation}
d_C(a,c)=d_C(b,c)=2~,
\qquad
d_C(a,b^)=2~.
\end{equation}
Minimising \eqref{eq:nonextremal-n2-action} over $\sigma_t\in S_4$ for each allowed pair $(\sigma_\ell,\sigma_r)$ gives

\begin{equation}
\begin{array}{c|c|c}
(\sigma_\ell,\sigma_r)
&
\displaystyle \min_{\sigma_t\in S_4}
S_2(\sigma_t,\sigma_\ell,\sigma_r)
&
\text{minimising }\sigma_t
\\ \hline
(\Id,\Id) & 15 & \Id \\
(a,\Id)   & 14 & a \\
(\Id,b)   & 14 & b \\
(a,b)     & 14 & a,\ b,\ c
\end{array}
\label{eq:nonextremal-minimisation}
\end{equation}

These minima follow directly from the triangle inequality for the Cayley distance. For example, when $(\sigma_\ell,\sigma_r)=(a,\Id)$, the part of the action depending on $\sigma_t$ is $d_C(\sigma_t ,c)
+d_C(\sigma_t ,a)
+d_C(\sigma_t, \Id)~.$
The triangle inequality implies
\begin{align}
2\Big[
d_C(\sigma_t ,c)
+d_C(\sigma_t ,a)
+d_C(\sigma_t, \Id)
\Big]
\geq{}&
d_C(c,a)
+d_C(c, \Id)
+d_C(a, \Id)
\nonumber\\
={}&
2+3+1=6~.
\end{align}
The bound is saturated by $\sigma_t=a$, giving a minimum total action of $14$. The case $(\sigma_\ell,\sigma_r)=(\Id,b)$ follows in the same way.
For $(\sigma_\ell,\sigma_r)=(a,b)$, the $\sigma_t$-dependent part is
\begin{equation}
d_C(\sigma_t ,c)
+d_C(\sigma_t ,a)
+d_C(\sigma_t, b)~.
\end{equation}
Its minimum is $4$, attained by $\sigma_t=a$, $\sigma_t=b$, and $\sigma_t=c$. The global minimum is therefore
\begin{equation}
S_{2,\min}=14~,
\end{equation}
and it is attained by the five configurations
\begin{equation}
(\sigma_t,\sigma_\ell,\sigma_r)
=
(a,a,\Id)~,
\qquad
(b,\Id,b)~,
\end{equation}
\begin{equation}
(\sigma_t,\sigma_\ell,\sigma_r)
=
(a,a,b)~,
\qquad
(b,a,b)~,
\qquad
(c,a,b)~.
\end{equation}
Each of these configurations contributes $D^{-14}$ at leading order. Restoring the projector normalization gives
\begin{equation}
\overline{\tr\left[(P_AP_B)^2\right]}
\sim
5D_A^2D_B^2D^{-14}
=
5D^{16}D^{-14}
=
5D^2~.
\label{eq:nonextremal-example-n2}
\end{equation}
Comparing \eqref{eq:nonextremal-example-n1} and \eqref{eq:nonextremal-example-n2}, we find
\begin{equation}
\overline{\tr(P_AP_B)}
\sim D^3~,
\qquad
\overline{\tr\left[(P_AP_B)^2\right]}
\sim 5D^2~,
\end{equation}
so the leading moments are dependent on $n$ in this non-extremal example.

\subsection{Corner algebra from averaging}
\label{app:CornerAvg}

We average the factor subalgebra map $\varphi^{(B)}(\mathcal{O}_A) = 
U_A^{(B)}(\mathcal{O}_A\otimes\mathbb{1})U_A^{(B)\dagger}$ over the Haar measure on 
$U(D_\partial - D_A)$, the group parametrising the choice of auxiliary isometries 
$T^{(a)}$, $a \geq 2$. Using the explicit form of $U_A^{(B)}$,
\begin{equation}
    U^{(B)}_A = \sum_j \ket{\phi^1_j}\bra{j, e_1} 
    \;+\; 
    \sum_{\substack{a,a'\geq 2\\j,j'}} B_{(a,j),(a',j')}\,\ket{\phi^a_j}\bra{j', e_{a'}},
\end{equation}
we expand $\varphi^{(B)}(\mathcal{O}_A)$ in blocks. The $a=1$ sector contributes
\begin{equation}
    \sum_{j,j'} (\mathcal{O}_A)_{jj'}\ket{\phi^1_j}\bra{\phi^1_{j'}},
\end{equation}
which is independent of $B$ and therefore unaffected by the averaging. The $a \geq 2$ 
sectors contribute terms of the form
\begin{equation}
    \sum_{\substack{a,a' \geq 2 \\ j,j'}} \sum_{\substack{b,b' \geq 2 \\ k,k'}}
    B_{(a,j),(a',j')}\,\overline{B_{(b,k),(b',k')}}\,
    (\mathcal{O}_A)_{j'k'}\,\ket{\phi^a_j}\bra{\phi^b_k}.
\end{equation}
Note that mixed terms between the $a=1$ sector and the $a\geq 2$ sectors vanish, 
since $\mathcal{O}_A \otimes \mathbb{1}$ preserves the $a$-label.

Applying the Haar average over $B \in U(D_\partial - D_A)$,
\begin{equation}
    \mathbb{E}_B\!\left[B_{IJ}\overline{B_{I'J'}}\right] 
    = \frac{1}{D_\partial - D_A}\,\delta_{II'}\delta_{JJ'}~,
\end{equation}
where $I = (a,j)$ and $J = (a',j')$ are collective indices with $a,a' \geq 2$, 
the averaged contribution from the $a \geq 2$ sectors becomes
\begin{equation}
    \frac{1}{D_\partial - D_A} \sum_{\substack{a \geq 2\\j}} \sum_{\substack{a' \geq 2\\j'}}
    (\mathcal{O}_A)_{j'j'}\,\ket{\phi^a_j}\bra{\phi^a_j}
    = \frac{(m-1)\mathrm{Tr}(\mathcal{O}_A)}{D_\partial - D_A}
    \sum_{\substack{a \geq 2\\j}}\ket{\phi^a_j}\bra{\phi^a_j}.
\end{equation}
The sum $\sum_{a \geq 2, j}\ket{\phi^a_j}\bra{\phi^a_j}$ is precisely the projector 
onto the orthogonal complement of the code subspace, $\mathbb{1}_\partial - P_A$, 
which has dimension $D_\partial - D_A$. Hence the prefactor simplifies and we obtain
\begin{equation}
    \phi(\mathcal{O}_A) 
    \equiv \mathbb{E}_B\!\left[\varphi^{(B)}(\mathcal{O}_A)\right]
    = T_A \mathcal{O}_A T_A^\dagger + \tau(\mathcal{O}_A)(\mathbb{1}_\partial - P_A)~,
\end{equation}
where $\tau(\mathcal{O}_A) = \frac{1}{D_A}\mathrm{tr}(\mathcal{O}_A)$ is the 
normalised trace and $P_A = T_A T_A^\dagger$ projects onto the code subspace. 
This operator manifestly belongs to the corner algebra $\mathcal{N}_A$.

\begin{proposition}
The map $\phi : B(\mathcal{H}_A) \to \mathcal{N}_A$ is linear, unital, and 
$*$-preserving. Its failure to be multiplicative is central:
\begin{equation}
    \phi(\mathcal{O}_A \mathcal{O}_B) - \phi(\mathcal{O}_A)\phi(\mathcal{O}_B)
    = \bigl(\tau(\mathcal{O}_A \mathcal{O}_B) - \tau(\mathcal{O}_A)\tau(\mathcal{O}_B)
    \bigr)(\mathbb{1}_\partial - P_A)~,
\end{equation}
so $\phi$ is a unital $*$-homomorphism after quotienting the central ideal $ \mathbb{C}(\mathbb{1}_\partial - P_A)$.
\end{proposition}

\begin{proof}
Linearity and $*$-preservation are immediate. Unitality follows from 
$\phi(\mathbb{1}_A) = T_A T_A^\dagger + (\mathbb{1}_\partial - T_A T_A^\dagger) 
= \mathbb{1}_\partial$. Using $T_A^\dagger T_A = \mathbb{1}_A$ and 
$P_A(\mathbb{1}_\partial - P_A) = 0$,
\begin{equation}
    \phi(\mathcal{O}_A)\phi(\mathcal{O}_B)
    = T_A \mathcal{O}_A \mathcal{O}_B T_A^\dagger 
    + \tau(\mathcal{O}_A)\tau(\mathcal{O}_B)(\mathbb{1}_\partial - P_A)~,
\end{equation}
while $\phi(\mathcal{O}_A \mathcal{O}_B) = T_A \mathcal{O}_A \mathcal{O}_B T_A^\dagger 
+ \tau(\mathcal{O}_A \mathcal{O}_B)(\mathbb{1}_\partial - P_A)$. Subtracting gives 
the stated formula.
\end{proof}

The failure of $\phi$ to be multiplicative is thus entirely supported on the 
orthogonal complement of the code subspace. Forming the quotient by the central 
ideal $\mathbb{C}(\mathbb{1}_\partial - P_A)$ recovers the exact representation 
$\mathcal{O}_A \mapsto T_A \mathcal{O}_A T_A^\dagger$ on the code subspace.

\bibliographystyle{jhep}
\bibliography{references}
\end{document}